\documentclass[12pt]{article}
\usepackage{shortcuts}
\usepackage{amsmath}
\usepackage{amssymb,amsmath}
\usepackage{amsfonts,mathrsfs,bbm}
\usepackage{bbm}
\usepackage{url}
\usepackage{setspace,float}
\usepackage{natbib}
\usepackage{subcaption} 
\usepackage[symbol]{footmisc}
\newcommand{\blind}{1}

\title{Representation Transfer Learning for Semiparametric Regression}

\begin{document}

\def\spacingset#1{\renewcommand{\baselinestretch}%
{#1}\small\normalsize}
\spacingset{1}


\if1\blind
{
  \title{\bf Representation Transfer Learning for Semiparametric Regression}
   \author{Baihua He\thanks{Equal contribution}\
   \thanks{Department of Statistics and Finance, School of Management, University of Science and Technology of China, Hefei, China. Email: baihua@ustc.edu.cn},
 \   Huihang Liu$^*$\thanks{International Institute of Finance, School of Management, University of Science and Technology of China, Hefei, China. Email: huihang@mail.ustc.edu.cn},
 \  Xinyu Zhang\thanks{ Academy of Mathematics and Systems Science, Chinese Academy of Sciences, Beijing, China. Email: xinyu@amss.ac.cn},
   \ {and Jian Huang}\thanks{Corresponding author.
    Department of Applied Mathematics, The Hong Kong Polytechnic University, Hong Kong SAR, China. Email: j.huang@polyu.edu.hk}
}
 \maketitle
\fi

\if0\blind
{
  \bigskip
  \begin{center}
    {\LARGE\bf Representation Transfer Learning for Semiparametric Regression}
\end{center}
  \medskip
} \fi

\begin{abstract} 
We propose a transfer learning method that utilizes data representations in a semiparametric regression model. Our aim is to perform statistical inference on the parameter of primary interest in the target model while accounting for potential nonlinear effects of confounding variables.
We leverage knowledge from source domains, assuming that the sample size of the source data is substantially larger than that of the target data. This knowledge transfer is carried out by the sharing of data representations, predicated on the idea that there exists a set of latent representations transferable from the source to the target domain. We address model heterogeneity between the source and target domains by incorporating domain-specific parameters in their respective models. We establish sufficient conditions for the identifiability of the models and demonstrate that the estimator for the primary parameter in the target model is both consistent and asymptotically normal. These results lay the theoretical groundwork for making statistical inferences about the main effects. Our simulation studies highlight the benefits of our method, and we further illustrate its practical applications using real-world data.
\end{abstract}

\noindent%

\spacingset{1.0} 
\section{Introduction} \label{sec:intro}

In many practical scenarios, the availability of data can be limited, posing difficulties for estimating effective models for statistical inference. Often, there is an abundance of data in  related but not identical source domains, while the specific target domain suffers from data scarcity.
Transfer learning provides a solution to this issue by leveraging knowledge from  similar source tasks to enhance model performance in the target task \citep{pan2010survey}.
Over the past decade, transfer learning has been widely adopted in various machine learning tasks, including computer vision \citep{yosinski2014how}, natural language processing \citep{mou2016how}, and speech recognition \citep{huang2013crosslanguage}.

The challenges of {domain heterogeneity} and the risks of {negative transfer} in utilizing auxiliary source data have led to the advancement of {statistical transfer learning}, which seeks to develop novel transfer learning methodologies that address these challenges {and establish their theoretical properties}. 
Researchers have proposed transfer learning methods for a variety of models, including high-dimensional linear models \citep{bastani2021predicting,li2022transfer}, generalized linear models \citep{tian2022transfer,li2023estimation}, functional regression \citep{lin2022transfer}, semi-supervised classification \citep{zhou2022doubly}, and basis-type models \citep{cai2024transfer}, among others.
\citet{tian2023learning} introduced a linear representation multi-task method
for estimating a similar representation. However, their reliance on linear associations and linear representation constraints the method's applicability to complex data structures.
In the context of transfer learning, \citet{hu2023optimal} proposed a model averaging approach for semiparametric regression models. However, their method only leverages the linear component's knowledge for transfer learning, overlooking the non-linear component similarities.
\lhhcomment{Term ``statistical transfer learning'' is not widely used in literature.}

Although these methods have shown promising results, there are certain issues that remain unexplored.
First, the challenges in balancing the trade-off between model flexibility and parameter interpretability. Most existing research on transfer learning and multi-source data integration falls into two categories: studies that focus on parametric models \citep{tian2022transfer,li2023estimation}, which offer simplicity and interpretability, and those that explore complex non-linear models \citep{tan2018survey} lack interpretability due to their ``black-box'' nature. Striking a balance between the interpretability of parametric models and the flexibility of non-parametric models remains a significant challenge. Second, the challenges remain in {constructing and identifying the transferable knowledge} across data domains. Most existing statistical transfer learning methods transfer the model parameters directly to the target domain. However, these methods depends on parametric model assumptions, which may limit the transferability of the knowledge and fail to utilize the {latent} shared information among domains. The knowledge is not transferable if the model parameters are different, even if there are {latent} structures in the data. The identifiability of transferable knowledge is also crucial for the statistical inference {before} transferring. These challenges motivate us to develop a novel transfer learning method that uses the shared knowledge among domains and strike a balance between model interpretability and flexibility,

We propose a representation transfer learning (RTL) method for knowledge transfer within the context of semiparametric regression model \citep{engle1986semiparametric}. The semiparametric regression model enables the interpretability of the treatment parameters, or parameters of main interest, while capturing the flexible data structures through the nonparametric components. Our main objective is to facilitate statistical inference for the treatment effects in the target model, taking into possibly nonlinear effects of the confounding variables. We achieve this by accommodating multi-dimensional confounding variables in a flexible manner and by incorporating information from source domains. The scenario we address involves a target model of interest, together with independent heterogeneous source models that share a higher-level data representation. The representation transfer learning mechanism enables the latent shared knowledge to be transferred. We use deep neural networks in the estimation of the effects of confounding variables through a set of representation functions. The transfer of knowledge from source data to the target model occurs via these representation functions, which we estimate by capitalizing on the ample sample sizes available in the source domains.

The main contribution of our paper are threefold. First, a critical concern is the {identifiability} of both the transferred representation function and the domain-specific parametric components. Given that the representation is a high-level abstraction of the data, it is often non-unique and non-identifiable \citep{chen2023leveraging}. Within the semiparametric regression framework, the identifiability of the parametric components can be compromised due to their interaction with the representation. To tackle this challenge, we formulate novel and interpretable conditions that ensure the identifiability of both the representation function and the linear coefficients, provided there is sufficient diversity among the source domains. Second, our theoretical analysis shows that the proposed method can consistently estimate the representation functions via deep neural networks with ReLU activation. We demonstrate that representation transfer learning can reduce approximation bias and enhance sample efficiency. Third, we establish the asymptotic normality of the estimated primary parameter in the target model, providing the basis for statistical inference regarding the effects of the variable of primary interest. Consequently, RTL adeptly balances model interpretability with model flexibility and improve the estimation accuracy of the primary parameter.

Our proposed RTL approach marks a significant departure from existing statistical methods for transfer learning and semiparametric regression models. Unlike the traditional distance-based transfer learning frameworks, we account for model heterogeneity between the source and target domains by employing flexible representation functions and domain-specific parameters. These learned representation functions serve as conduits for knowledge transfer, capturing intrinsic information that is often the most challenging aspect to estimate in a model.

The rest of the article is organized as follows. We introduce the model framework and develop
the proposed RTL method  in Section~\ref{sec:model}. We provide the theoretical guarantees in Section~\ref{sec:theory}. We present the numerical studies including simulation, semi-synthetic data analysis based on MNIST hand writing dataset, and illustrate RTL using the Pennsylvania reemployment bonus experiment and housing rental information data in Section~\ref{sec:emp-results}. We give concluding remarks in Section~\ref{sec:discuss}, and relegate all technical proofs to the Supplementary Materials.

\section{Model and methodology} \label{sec:model}

In this section, we present our proposed Representation Transfer Learning (RTL) method. We denote the data by the triplet $(Y, \mathbf{X}, \mathbf{Z})$, where $Y \in \mathbb{R}$ is the response variable, $\mathbf{X} \in \mathbb{R}^d$ corresponds to the $d$-dimensional covariate of primary interest, and $\mathbf{Z} \in \mathbb{R}^q$ a  $q$-dimensional confounding variable. Typically, $\mathbf{X}$ is a low-dimensional treatment variable, implying that $d$ is small. We permit the confounding variable $\mathbf{Z}$ to be of a moderately high dimension, with its dimension $q$ allowed to grow as the sample size increases.
We consider $K$ distinct source domains, each with its own dataset $(Y_k, \mathbf{X}_k, \mathbf{Z}_k)$ for $k = 1, \ldots, K$. Additionally, we have the target domain data, which is denoted as $(Y_0, \mathbf{X}_0, \mathbf{Z}_0)$. Our approach is designed to leverage the information from these multiple source domains to enhance inference in the target domain via shared data representation.

\subsection{Model} 
We begin by considering distinct semiparametric regression models for each source domain and the target domain:
\begin{align}
\text{Sources:} \quad
Y_k& = \bbeta_{k}^\top \bX_k + g_k(\bZ_k) + \varepsilon_k, \quad k=1, \ldots, K,
  \label{s1}    \\
\text{Target:} \quad   Y_0 & =  \bbeta_{ 0}^\top \bX_0 + g_0(\bZ_0) + \varepsilon_0.
  \label{eq:model}
\end{align}
In these models, for $k=0, 1, \ldots, K$, $\bbeta_{k}$ represents the effects associated with $\bX_k$, $g_k: \mathbb{R}^{q} \to \mathbb{R}$ denotes an unspecified nonparametric function capturing the potential nonlinear impact of the confounding variable $\bZ_k$, and $\varepsilon_k$ is the random noise component with $E(\varepsilon_k)=0$ and $E(\varepsilon_k^2)=\sigma_k^2$. This model allows us to systematically address the influence of both the covariates of interest and the confounding variables across different domains.

Our main goal is to conduct statistical inference on the parameter $\bbeta_{0}$, which measures the effect of the covariate of interest, $\bX_0$, on the outcome variable $Y_0$. This task is challenging for two main reasons: the limited availability of data from the target domain and the complex influence that nonparametric estimation of nuisance functions, related to multivariate confounding variables, has on the estimation of $\bbeta_{0}$. Although the use of flexible neural networks to approximate these functions may appear to be a feasible approach, it complicates the inference process for $\bbeta_{0}$. Furthermore, this method in itself does not overcome the fundamental obstacle known as the ``curse of dimensionality", which arises during the nonparametric estimation of a multidimensional function.

Transfer learning offers a solution to the ``curse of dimensionality" in the target domain by utilizing data from multiple sources. We are particularly interested in scenarios where the combined sample size from the source domains significantly exceeds that of the target domain. To enable the effective transfer of knowledge from the source domains to the target domain, it is crucial to establish specific assumptions about the relationship between the source and target data models. Our approach capitalizes on the source data to assist in estimating the relevant function within the target data model. This method assumes the existence of a latent representation of the confounding effect that is invariant across both the source and target data.

Specifically, we propose the following expression for the confounding effects:
\begin{align*}
g_k(\bZ) =  \bgamma_k^\top \bR(\bZ), \quad k=0, 1, \ldots, K,
\end{align*}
where $\bR: \mathbb{R}^q \to \mathbb{R}^p$ functions as a representation of the confounding variables, and $\bgamma_k$ represents domain-specific coefficients. This composite model structure aligns with the approaches suggested by \citet{du2020few} and \citet{tripuraneni2020theory}. The representation $\bR$ can be interpreted as a set of basis functions, with $\bgamma$ acting as the corresponding weights. This strategy differs from traditional basis expansion techniques, such as spline methods, which rely on a predetermined set of basis functions for approximating nonparametric functions. Instead, our approach estimates the representation function $\bR$ from the data.

By focusing on the differences in the coefficients $(\bbeta, \bgamma)$ across the source and target domains, we can effectively capture domain heterogeneity. This approach simplifies the challenging tasks of function estimation and heterogeneity detection. Consequently, we posit that the representation function $\bR$ is a shared element across different domains, representing the transferable knowledge from source tasks to the target task.

The above discussion leads to the proposed RTL model as follows:
\begin{align}
  \text{Sources:} \quad   Y_k&=\bbeta_{k}^\top \bX_k  + \bgamma_{k}^\top \bR(\bZ_k) +\varepsilon_k,\ k=1, \ldots, K,   \label{srtr1}\\
  \text{Target:}  \quad    Y_0&=\bbeta_{0}^\top \bX_0 +\bgamma_{0}^\top \bR (\bZ_0) +\varepsilon_0, \label{srtr2}
  \end{align}
where $\bbeta_{k}$ and $\bgamma_{k}$ are source-specific coefficients of dimensions $d$ and $p$, respectively. The function $\bR: \mathbb{R}^q \rightarrow \mathbb{R}^p$ serves as the shared representation function across domains.

\subsection{Estimation method}

Based on the RTL models (\ref{srtr1}) and (\ref{srtr2}), at the population level, our proposed
RTL method proceeds in two steps:
\begin{description}
\item[Step P1:]
In the source domain,  we consider the minimizers of the population risk function
\begin{equation}\label{deff0}
 \{(\bbeta_{\ast k}, \bgamma_{\ast k})_{k=1}^K, \bR_*\}\in
  \argmin\limits_{\{(\bbeta_k, \bgamma_k)_{k=1}^K, \bR\}}\frac{1}{K}
  \sum_{k=1}^{K}\mathbb{E}\{Y_k-\bbeta_k^\top \bX_k - \bgamma_k^\top \bR(\bZ_k)\}^2,
\end{equation}
where $\bR_*$ is the shared representation function that will be used in the target domain,
the coefficients $(\bbeta_{\ast k}, \bgamma_{\ast k})_{k=1}^K$ take into account possible heterogeneity
across source domains and the target domain. {The confounding effects $\bgamma_{\ast k}^\top \bR_*(\bZ_k)$ are not separable and thus are not identifiable, as discussed in Section~\ref{sec:theory}. Fortunately, the effect $\bbeta_{\ast k}$ is uniquely identifiable.}

\item [Step P2:]
In the target domain, given the representation function $\bR_*$ from the source domain, we solve
\begin{equation}
\label{defft}
\{\bbeta_{\ast 0}, \bgamma_{\ast 0}\} = \argmin_{(\bbeta_0, \bgamma_0)}
\mathbb{E} \{Y_0 - \bbeta_0^\top \bX_0 - \bgamma_0^\top \bR_*(\bZ_0)\}^2.
\end{equation}
\end{description}

Now suppose we have a random sample of independent and identically distributed observations from the target domain, denoted as $\{(Y_{0i}, \bX_{0i}, \bZ_{0i}), i=1, \ldots, n_0\}$. Additionally, we have access to the datasets from $K$ source domains,  $\{(Y_{ki}, \bX_{ki}, \bZ_{ki}), i=1, \ldots, n_k\},$ where $k=1, \ldots, K.$  Let $N=n_1+\cdots n_K$ be the combined sample size of the source domains.
Although there are no explicit constraints on the sample sizes across the target and source domains, it is usually the case that $N$ significantly exceeds the sample size $n_0$ of the target domain.
Our main goal is to leverage the data from these source domains to enhance the estimation accuracy within the target domain. This can be achieved by using the empirical version of the formation at the population given in (\ref{deff0}) and (\ref{defft}).
We first focus on estimating the representation function $\bR$ using the source data, followed by estimating the regression parameters within the target model using the target data. These two steps correspond to the Steps P1 and P2 at the population level and are
as follows:

\begin{description}
\item [Step E1: Estimation of the shared representation function.]
This step involves estimating a shared representation function $\bR$, which is formulated as an optimization problem:
\begin{align}\label{estMul}
  \{(\widehat{\bbeta}_k,\widehat{\bgamma}_k)_{k=1}^K,\widehat{\bR}\} = \argmin_{\{(\bbeta_k,\bgamma_k)_{k=1}^K, \bR\in \mathcal{R}\} }
  \Big\{\frac{1}{K}\sum_{k=1}^{K}\frac{1}{n_k}\sum_{i=1}^{n_k}
  (Y_{ki}- \bbeta_k^\top \bX_{ki} - \bgamma_k^\top \bR(\bZ_{ki}))^2
  \Big\},
\end{align}
where the estimation of the representation function $\bR$ is conducted over a specified class of neural networks, denoted as $\mathcal{R}$. This step is crucial for capturing the underlying
representations shared across the source domains.

\item [Step E2: Estimation of parameters in the target model.]
After estimating the representation function $\widehat{\bR}$ from the source data, the next step is to estimate the parameters within the target model. This is achieved by solving
\begin{equation}\label{estTar}
 \{\widehat{\bbeta}_0,\widehat{\bgamma}_0\}=\argmin_{\bbeta_0, \bgamma_0}
 \frac{1}{n_0} \sum_{i=1}^{n_0}
  \{Y_{0i}-\bbeta_0^\top \bX_{0i}-\bgamma_0^\top\widehat{\bR}(\bZ_{0i})\}^2.
\end{equation}
Given that $\widehat{\bR}$ remains fixed in this step, the pair $(\widehat{\bbeta}_0,\widehat{\bgamma}_0)$ is effectively obtained through a least squares estimation process.
\end{description}

The proposed RTL method, which involves pre-training on multiple source domains before transferring the estimated representations to the target domain, 
enhances data and computational efficiency. The abundance of source data ensures that the representation function $\bR$ can be estimated at a
much faster convergence rate than that when only target domain data is available.

\subsection{Implementation}
We approximate the representation functions by feedforward neural networks defined as:
\begin{equation*}
  \bR(\bz)=\bA_D\sigma(\bA_{D-1}\sigma(\cdots\sigma(\bA_0\bz+\bb_0)\cdots)+\bb_{D-1})+\bb_D,
\end{equation*}
where $\bA_i\in \mbR^{p_{i+1}\times p_i}$ and $\bb_i\in \mbR^{p_{i+1}}$ for $i=0,\ldots, D$, $p_0=q$ is the dimension of the input variables, $p_{D+1}=p$ is the dimension of the output layer, and $\sigma(\cdot)$ is the activation function. We consider the ReLU activation function $\sigma(x)=\max\{0,x\}$, applied component-wise.
The parameters of the representation function $\bR(\cdot)$ are denoted as
$\btheta=\{\bA_0,\ldots,\bA_D, \bb_0,\ldots,\bb_{D}\}$.
The number $W=\max\{p_0,\ldots,p_D\}$ and $D$ are the width and depth of the neural network, respectively. The weight matrices together with the bias vectors contain $S=\sum_{i=0}^{D}p_{i+1}(p_i+1)$ entries in total. The parameters including weights and biases
are assumed to be bounded by a constant $B_{\theta}>0.$
 We denote the set of  the neural network functions defined above by
 $\mathcal{R} =  \calN \calN (W,D,B_{\btheta}).$

Figure~\ref{fig:nn} illustrates the architecture of a demo neural network with $d=3$, $q=5$ and $p=3.$ We train the representation network and the linear layer in an iterative style for $400$ epochs.
In each epoch, we first update the representation network and then the linear layer.
The weights in the representation network are optimized using the SGD optimizer with a learning rate of $10^{-3}$ and batch size of $n_k$.
The weights in the last layer are obtained using the least square estimation.
We use early stopping method during the training process, 
and use other i.i.d. observations as a validation dataset for model selection with a sample size of $30\%$ of training dataset.
That is, we select the model with a minimum prediction error on the validation set for evaluation.

\begin{figure}[ht]
  \centering
  \includegraphics[width=4.5 in, height=2 in]{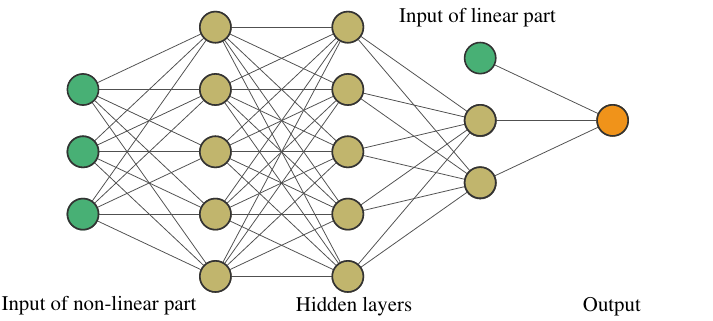}
  \caption{The architecture of a partially linear neural network with input dimension $d=1$, $q=3$, representation dimension $p=2$, depth $D=2$ and width $W=5$.}
  \label{fig:nn}
\end{figure}

To train the representation network across different datasets, we
compute the total loss function (\ref{estMul}) on each epoch and then backward the gradients to update the weights $\btheta$ in the representation network.
After the estimated representation $\widehat{\bR}$ is computed,  we obtain the estimator of $\bbeta_{\ast 0}$ and $\bgamma_{\ast 0}$ by solving (\ref{estTar}).

\section{Theoretic results} \label{sec:theory}
In this section, we study the theoretical properties of the proposed RTL method.
We first provide sufficient conditions under which the model parameters are identifiable.
Next, we derive the convergence rate of the estimated representation function in terms of the
source data sample size. Then we show that the estimator of the parameter of main interest in the target domain model is
asymptotically normal. We also provide a consistent estimator of the asymptotic covatiance matrix. These results make it possible to conduct statistical inference about the main parameter in the target domain model.

\subsection{Identifiability}
Identifiability is a fundamental question in statistical modeling problems.
Usually, the parameters in a model are required to be uniquely identifiable so that their consistent estimation is possible. In the proposed model, this requires careful consideration because of the term $\bgamma_{\ast k}^{\top} \bR_\ast(\bZ_k)$ representing the confounding effect in the model. Since $\bgamma_{\ast k}$ and $\bR_\ast$ are both unknown, they are not identifiable in the usual sense. In this subsection, we provide a set of conditions to guarantee the identifiability of the parameter of main interest $\bbeta_{\ast k}$ and the confounding effects  represented by $\bgamma_{\ast k}^{\top} \bR_{\ast}(\bZ_k)$.

We first state the definition of a notion of identifiability, linear identifiability, for $\bgamma_k$ and $\bR$.

\begin{definition}
  The data representations are said to be linearly identifiable if there exists an invertible matrix $\Lambda$ such that $\bR^{\prime}(\bZ)=\Lambda^{-1}\bR(\bZ)$ and $\bgamma_k^{\prime}=\Lambda\tr \bgamma_k$ for all $\bZ\in \calZ$ and $k\in[K].$
\end{definition}

Based on this definition, we have $\bgamma_k^{\prime\top}\bR^{\prime} = \bgamma_k^\top \Lambda \Lambda^{-1}\bR=\bgamma_k^\top \bR.$ Therefore, although $\bgamma_k$ and $\bR$ are only linearly identifiable, the confounding effects, represented by $\bgamma_k^\top \bR,$ are uniquely identifiable in the usual sense.

We impose the following conditions to ensure the identifiability of parameters and representations.
\begin{condition} \label{assumKL}
   The matrix $\mathbb{E}[\{\bX_k- \mathbb{E}(\bX_k|\bZ_k)\}\{\bX_k- \mathbb{E}(\bX_k|\bZ_k)\}^\top ]$ is invertible.
\end{condition}

\begin{condition}\label{assumind}
 (a)  There exists 
 $\{k_i\}_{i=1}^{p}\subseteq[K]$ such that the coefficients $\{\bgamma_{\ast k_i}\}_{i=1}^p$ are linearly independent.
 (b)  There exists $\bZ_1,\ldots,\bZ_p\in \calZ$ such that the matrix $
  \left[\bR_\ast(\bZ_1),\ldots,\bR_\ast(\bZ_p)\right]$ is invertible.
\end{condition}
Condition~\ref{assumKL} is a common assumption in regression analysis in the presence of confounding variables,  which assumes that the main variables $\bX_k$ have significant variation across different tasks after projecting out all variation that can be explained by the nuisance variables $\bZ_k$, for each $k\in[K].$
When confounding variables are present, they can introduce bias or distortions that obscure the true relationship between the variables of interest. By requiring that $\bX_k$ maintains significant variation independent of these confounders, it ensures that the effects of $\bX_k$ on the response variable $Y_k$ can be properly estimated.

Condition \ref{assumind} (a)  requires that the support of the distribution for the coefficients $\{\bgamma_{\ast k}\}_{k=1}^K$, is sufficiently rich.
Similar assumption was also imposed in the analysis of panel data \citep{ahn2001gmm,bai2009panel,moon2015linear}.
Condition \ref{assumind} (b)
stipulates  that $\bR_\ast$ exhibits a sufficient degree of variability. This variability is essential to ensure that the image of $\bR_\ast$—the set of all possible outputs it can generate—does not become confined within a proper subspace of its potential range. In simpler terms, the function must be versatile enough in its transformations to avoid being restricted to a limited portion of the space it operates within.

\begin{theorem}\label{theidenti}
  Suppose Conditions \ref{assumKL}-\ref{assumind} 
  hold.
  Let $\{(\bbeta_{k}, \bgamma_{k})_{k=1}^K, \bR\}$ and $\{(\bbeta_{k}^{\prime}, \bgamma_{k}^{\prime})_{k=1}^K, \bR^{\prime}\}$
  be sets of parameters satisfy (\ref{deff0}).
  Then $\bbeta_{k}^{\prime}=\bbeta_{k}$ and there exists an invertible matrix $\Lambda$ such that
  $\bR^{\prime}=\Lambda^{-1}\bR$ and $\bgamma_{k}^{\prime}=\Lambda\tr \bgamma_{k}$ for $k\in[K]$.
\end{theorem}

Theorem \ref{theidenti} shows that the representation function $\bR$ is identifiable up to a multiplicative matrix transformation if Conditions \ref{assumKL}-\ref{assumind} are satisfied.

In the following, we use a simple example to illustrate the identifiability of our proposed model.
We set the representation dimension as $p=2$, $3$, and $5$, the dimension of non-linear part as $q=p$, and the dimension of linear part as $d=1$.
The data generating process is set as $Y = \beta X + \bgamma^\top \bR(\bZ) + \epsilon$, where $\bbeta$ and $\bgamma$ are generated from standard normal distribution, $X\in\mathbb{R}$ and $\bZ\in\mathbb{R}^q$ are from standard normal distribution and $\epsilon\sim\mathcal{N}(0, 0.3^2)$. The representation functions $\bR(\cdot)$ are generated by the following univariate functions: $\sin(\pi x)$, $\cos(\pi x)$, $2 \sqrt{|x|} - 1$, $(1 - |x|)^2$, $1 / (1 + \exp(-x))$ and $- \sin(x)$.
The linear coefficients are heterogeneous in the source dataset.
We set the sample size in source dataset as $n_{k}=2000$ for all $k=1,\dots,K$ and let $K=8$.
The estimated representation function is transformed by a linear transformation which is identified by minimizing the distance between the transformed representation and the true one.
Figure~\ref{fig:toy-example} shows the transformed learned representation (solid line) by the proposed method and the true representation function (dashed line).

\begin{figure}[ht]
  \centering
  \includegraphics[width=0.32\linewidth]{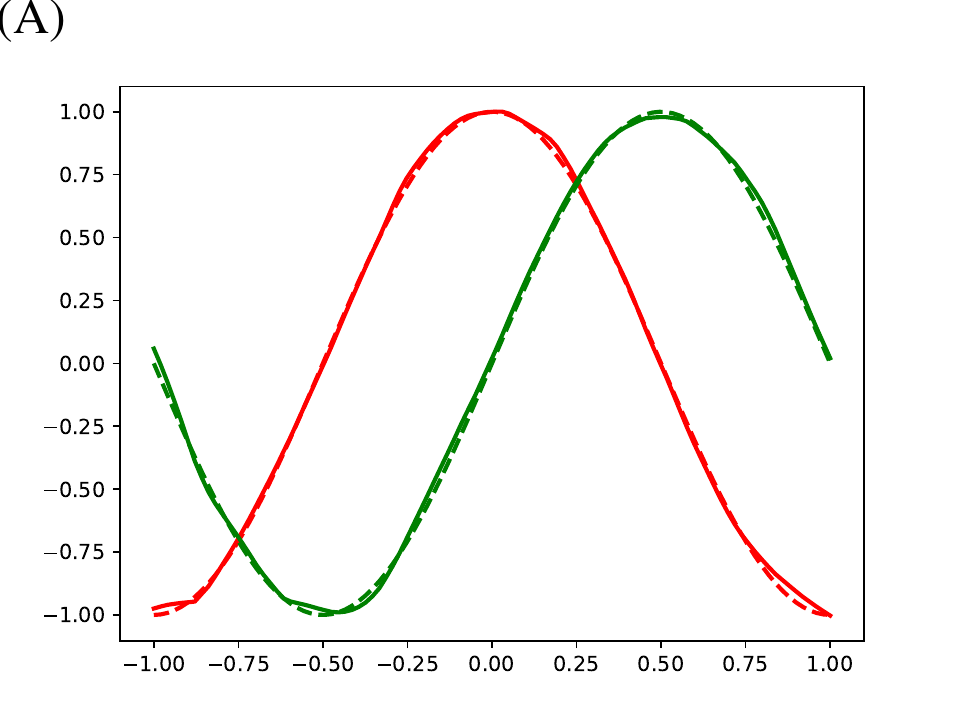}
  \includegraphics[width=0.32\linewidth]{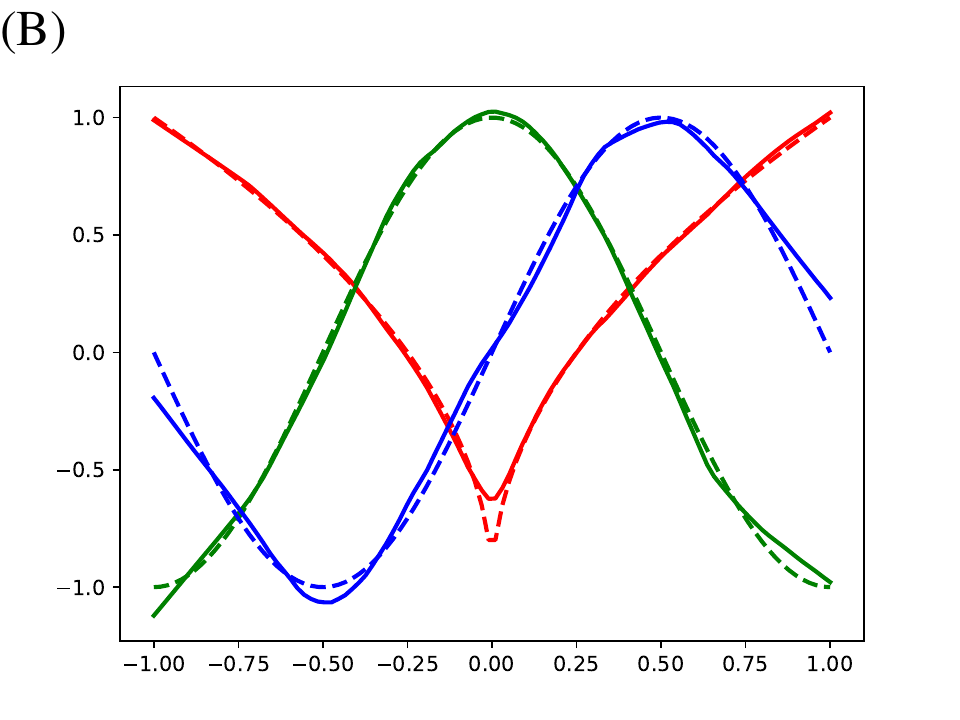}
  \includegraphics[width=0.32\linewidth]{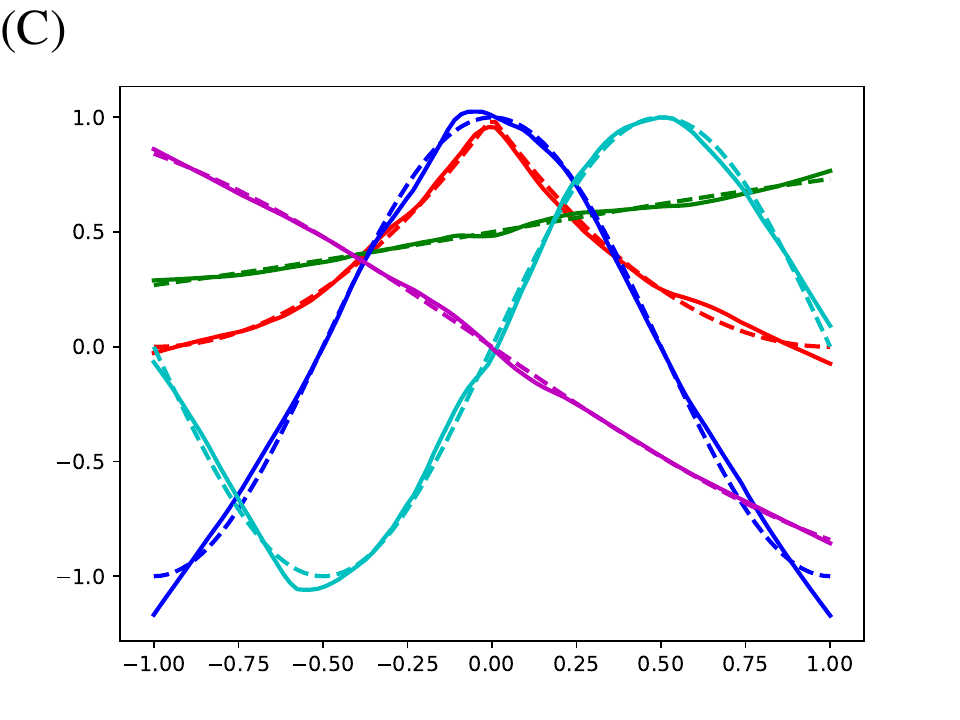}
  \caption{Demonstration of learned representation by RTL.
  The left panel uses $\sin(\pi x)$ and $\cos(\pi x)$ as the true representation functions.
  The middle panel uses $2 \sqrt{|x|} - 1$, $\sin(\pi x)$ and $\cos(\pi x)$.
  The right panel uses $(1 - |x|)^2$, $1 / (1 + \exp(-x))$, $- \sin(x)$, $\sin(\pi x)$ and $\cos(\pi x)$.
  The solid line represents the learned representation function and the dashed line represents the true representation function.}
  \label{fig:toy-example}
\end{figure}

\subsection{Convergence rate}
\label{sec:rate of h}
Based on the identifiability of the representations, we derive the convergence rate of the estimated representation function in this section.
We impose a set of regularity conditions to guarantee the consistency of the estimated representation function.

\begin{condition}\label{conb}
  (a) For any $\bbeta\in \calB$, there exists a constant $B_{\beta}$ such that $\|\bbeta\|_2\leq B_{\beta}$, where for any $d$-vector ${\ba}$, $\|{\ba}\|_2=\sqrt{\sum_{j=1}^d a_j^2}$.
  For any $\bgamma\in \Gamma$, there exists a constant $B_{\gamma}$ such that $\|\bgamma\|_2\leq B_{\gamma}$.  (b) The covariate $\bX\in\calX\subseteq \mathbb{R}^d$ satisfies that $\|\bX\|_2\leq B_X$.
  The representation function $\bR(\cdot)\in \calR$ satisfies $\|\bR(\bZ)\|_2\leq B_R$ for any $\bZ\in\calZ$.
  (c) The response variable $Y_{k}\in \calY\subseteq \mathbb{R}$ is subexponentially distributed for $k\in[K]$.
\end{condition}

Condition \ref{conb} includes standard assumptions in regression problems \citep{tripuraneni2020theory,jiao2023deep}. We further assume that the representation function is in the H\"{o}lder class.
\begin{definition}
  Let $\kappa=s+\nu>0$, $\nu\in(0,1]$ and $s=\lfloor\kappa\rfloor\in \mbN_0$, where $\lfloor\kappa\rfloor$ denotes the largest integer strictly smaller than $\kappa$ and $\mbN_0$ denotes the set of nonnegative integers. For a finite constant $B_0>0$, define the H\"{o}lder class $\calH^{\kappa}(\calZ, B_0)$ as
  \begin{align*}
    &\calH^{\kappa}(\calZ, B_0)\\&=\left\{R(\cdot):\calZ\rightarrow \mbR: \max\limits_{\|\omega\|_1\leq s} \|\partial^\bomega R\|_{\infty}\leq B_0, \max\limits_{\|\bomega\|_1=s}\sup_{\bZ\neq \bZ^\prime} \frac{|\partial^\bomega R(\bZ)-\partial^\bomega R(\bZ^\prime)|}{\|\bZ-\bZ^\prime\|_2^{\nu}}\leq B_0\right\},
  \end{align*}
  where $\partial^\bomega =\partial^{\omega_1}\cdots\partial^{\omega_q}$ with $\bomega=(\omega_1,\ldots,\omega_q)\tr \in \mbN_0^q$, and $\|\bomega\|_1=\sum_{j=1}^q \omega_j$.
\end{definition}
\begin{condition}\label{con5}
  Each element of the representation function $\bR$ belongs to the H\"{o}lder class $\calH^{\alpha}(\calZ, B_0)$.
\end{condition}

\begin{condition}\label{conRieZ}
{
The support $\calZ$ of the representation function $\bR: \mathbb{R}^q \to  \mathbb{R}^p$ belongs to a compact $p$-dimensional Riemannian manifold isometrically embedded in $\mathbb{R}^p$ with $p \le q.$
}
\end{condition}

\begin{condition}\label{con4}
{
The dimensions $\{p, q\}$ of $\{\bZ, \bR\}$ satisfy the following condition:
\begin{align*}
  \underline{n}^{-1/2} p^{1/2} \log N=o(1)\;\mbox{and}\;
p{(D+2+\log q)^{1/2}}\prod_{i=0}^D (p_i+1) (\log N)^2N^{-1/2} =o(1),
\end{align*}
where $D$ is the depth of the neural network, $\underline{n}=\min_{1\le k\le K} n_k,$
and $N=\sum_{k=1}^Kn_k.$
}
\end{condition}

Condition~\ref{conRieZ} is a low-dimensional manifold condition on $\calZ$.
Actually, Condition \ref{conRieZ} is not necessary for the convergence rate of the representation function. But it would guarantee a higher convergence rate of the representation function when $q$ is very large.
Condition~\ref{con4} pertains to the dimensionality of the model in relation to the sample sizes. This condition accommodates the presence of a moderately high-dimensional covariate vector, allowing the dimensions $\{p, q\}$ to increase indefinitely, provided that their rate of divergence meets the specified constraints. While this condition is met in numerous applications, it does not cover sparse, high-dimensional scenarios where the number of covariates exceeds the sample size.

For representation functions $\bR$ and $\bR^{\prime}$, denote
$d_2({\bR},\bR^{\prime})=(\expect \|{\bR}(\bZ)-\bR^{\prime}(\bZ)\|_2^2)^{1/2}.$
Let
$\Delta_N =  p^{1/2}N^{-{s}/({2s+p\log q})}+s_1^{1/2}(\log N)^2 N^{-1/2}$, where
$s_1=\max\{Kd,Kp,S\}.$  Typically, the size of the neural network $S > \max\{Kd, Kp\},$
thus $s_1$ is simply the network size used in the estimation.

{
\begin{theorem}\label{thediffh}
  Suppose Conditions \ref{assumKL}--\ref{con4} hold, there exists an invertible matrix $\Lambda_{\ast}$ such that
  \begin{align*}
    d_2(\widehat{\bR},\bR^{\ast})=O_p(\Delta_N),
  \end{align*}
  where $\bR^\ast=\Lambda_{\ast}^{-1}\bR_\ast$.
\end{theorem}
}
Theorem~\ref{thediffh} establishes the convergence rate of the estimated representation.
The rate is determined by two terms. The first term represents the approximation error,
which minimizes the distance from $\bR^\ast$ to $\calR$.
The second term represents the stochastic error.

\subsection{Asymptotic normality}

In this section, we establish the asymptotic normality of the estimated primary parameter within the target domain. It is a common scenario in transfer learning that the total sample size $N$ from the source domains significantly exceeds the sample size $n_0$ in the target domain. Our derivation of asymptotic distribution is conducted with this disparity in sample sizes taken into consideration.

We need the following condition.

\begin{condition}\label{inv}
  The matrix $\bJ_0=\expect[\{\bX_0-\bfm^\ast(\bZ_0)\}\{\bX_0-\bfm^\ast(\bZ_0)\}\tr ]$ is invertible and $\expect[\{\bX_0-\bfm^\ast(\bZ_0)\}\tr \{\bX_0-\bfm^\ast(\bZ_0)\}]<\infty$, where $\bfm^\ast(\bZ_0)=\ExRaz$.
\end{condition}
Condition \ref{inv} is fairly standard in the semi-parametric regression literature, which
is needed for constructing a semi-parametric efficient estimator of $\bbeta_0.$
We note that the independence between $\bX_k$ and $\bZ_k$ is not required for $k\in[K]$, throughout the paper.

To remove the confounding effect of $\bR^{\ast}(\bZ_0)$, we consider finding an $d\times p$
matrix $\bmu^\ast$ that satisfies the orthogonality equation,
\begin{align*}
    \expect\left[\left\{\bX_0-\bmu\bR^\ast(\bZ_0)\right\} \bR^{\ast\top}(\bZ_0) \right]=\bf{0}.
\end{align*}
This is equivalent to find a  $\bmu^\ast$ that is a  minimizer of $\expect[\|\bX_0-\bmu\bR^\ast(\bZ_0)\|_2^2]$. Moreover, the efficient score for $\bbeta_{\ast 0}$ is $\{\bX_0-\bmu^{\ast}\bR^{\ast}(\bZ_0)\}\epsilon_0$.
The following theorem establishes the asymptotic normality of $\bbeta_{\ast 0}$.

\begin{theorem}\label{thenorm}
  Suppose Conditions \ref{assumKL}-\ref{inv} hold, we have
  \begin{align}\label{temnorm0}
    \sqrt{n_0}(\wbbetaz-\bbeta_{\ast 0})&=\bJ_0^{-1}\Big[\frac{1}{\sqrt{{n}_0}}\sum_{i=1}^{n_0}
    \{\bxiz-\bmu^{\ast}\bRaiz\}\epsiz\Big]+O_p(\sqrt{n_0}\Delta_N^2).
  \end{align}
  where $\Delta_N =  p^{1/2}N^{-{s}/({2s+p\log q})}+s_1^{1/2}(\log N)^2 N^{-1/2}$.
Therefore, if $n_0^{1/2} \Delta_N^2\to 0$ as $N \to \infty$, that is, $n_0^{1/2}s_1(\log N)^4 N^{-1}\to 0$ and ${n_0}^{1/2}pN^{-{2s}/({2s+p\log q})}\to 0$ as $N \to \infty$,
 we have,
 \begin{align}\label{temnorm}
 \sqrt{n_0}(\wbbetaz-\bbeta_{\ast 0})    \stackrel{D}{\to} N(0,\sigma^2_0 {\bJ_0}^{-1}),  \text{ as } n_0 \to \infty \text{ and } N \to \infty.
  \end{align}
\end{theorem}

Through the data augmentation by the abundant source dataset, we prove that the estimator for $\bbeta_{\ast 0}$ attains $\sqrt{n_0}$-consistent and asymptotic normality.
The asymptotic expression of $\wbbetaz-\bbeta_{\ast 0}$ in Theorem \ref{thenorm} indicates the estimator $\wbbetaz$ attains the information bound, so it is semiparametrically efficient. When the variance term $\sigma_0^2\bJ_0^{-1}$ is unknown, we use a plug-in estimator. Based on Theorem \ref{thenorm}, a natural estimator of $\bmu^\ast$ is given by solving the equation
 $ \sum_{i=1}^{n_0}\bxiz \widehat{\bR}\tr (\bZ_{0i}) -\bmu \sum_{i=1}^{n_0}\wbRiz\widehat{\bR}\tr (\bZ_{0i}) =0 ,$
which leads to
\begin{equation}\label{estmu}
  \wbmu=\sum_{i=1}^{n_0}\bxiz\widehat{\bR}\tr (\bZ_{0i}) \left\{\sum_{i=1}^{n_0}\wbRiz\widehat{\bR}\tr (\bZ_{0i}) \right\}^{-1}.
\end{equation}
Combining \eqref{temnorm} and \eqref{estmu}, we can estimate the variance of $\wbbetaz$  by
$
  \widehat{\bSigma} 
  =\wbJz^{-1} \widehat{\mathbf{A}}\wbJz^{-1},
$
where
\begin{align*}
\widehat{\mathbf{A}}& =
 \frac{1}{n_0}\sum_{i=1}^{n_0}\left\{(Y_{i0}- \wbbetaz\tr\bxiz- \hat{\bgamma}_0\tr \widehat{\bR}(\bZ_{0i}) )^2 (\bxiz-\wbmu\wbRiz)(\bxiz-\wbmu\wbRiz)\tr\right\}, \\
 \wbJz &=\frac{1}{n_0} \sum_{i=1}^{n_0}(\bxiz-\wbmu\wbRiz)(\bxiz - \wbmu\wbRiz)\tr.
 \end{align*}

The next corollary shows that  $\widehat{\bSigma}$ is consistent.
\begin{corollary}\label{corvar}
  Under Conditions \ref{assumKL}-\ref{inv}, if $\epsiz$ and the components of $\bX_0-\ExbRaz$ have bounded fourth moments, then
  \begin{align*}
   \widehat{\bSigma} 
   \pover\sigma^2_0 {\bJ_0}^{-1}.
  \end{align*}
\end{corollary}
Theorem \ref{thenorm} and Corollary \ref{corvar}  shows that
the distribution of $\wbbetaz-\bbeta_{\ast 0})$ can be approximated by a normal distribution whose covariance matrix can be consistently estimated,
provide a theoretical basis for making statistical inference about the parameter of main interest in the target domain.

\subsection{Benefits from source data}
We now discuss the benefits of source data for estimating the primary parameter $\bbeta_0$ in the
target domain.

Suppose only target dataset were available. The basic semiparametric partially linear model is \citep{engle1986semiparametric},
$$
Y_0=\bbeta_{0}^\top \bX_0 + g_0(\bZ_0) +\varepsilon_0.
$$
Consider the least squares estimator
  \begin{align*}
      \{\widetilde{\bbeta}_0, \widetilde{g}_0\} = \argmin_{\bbeta, g_0} \frac{1}{n_0} \sum_{i=1}^{n_0}
\{Y_{0i}-\bbeta_0^\top \bX_{0i}-g_0 (\bZ_{0i})\}^2.
  \end{align*}
There is an extensive literature on the asymptotic properties of the least squares estimators
in the semiparametric regression model using various approximation methods such as splines
for dealing with the nonparametric component, see for example, \citet{hardle2000partially} and the references therein.
Under the conditions given in Section \ref{sec:theory}, it holds that \citep{hardle2000partially, farrell2021deep}
\begin{align} \label{diffgg}
      \expect \|\widetilde{g}_0 (\bZ) - g_{0}(\bZ) \|= O_p(n_0^{-{s}/({2s+ q})}).
\end{align}
The convergence rate in (\ref{diffgg}) is optimal \citep{stone1980optimal}.
Furthermore,
  \begin{align*}
    \sqrt{n_0} ( \widetilde{\bbeta}_0 - \bbeta_{*0})
= \bJ_0^{-1} \left[\frac{1}{\sqrt{n_0}}\sum_{i=1}^{n_0} \{\bxiz-\expect(\bxiz|\bziz)\}\epsiz\right] + \sqrt{n_0}\, O_p(n_0^{-{2s}/({2s+q})}).
  \end{align*}
Therefore, to ensure asymptotic normality of $\widetilde{\bbeta}_0,$ we must have
$
n_0^{1/2-{2s}/({2s+q})} \to 0.
$
This necessitates the condition $n_0^{-{s}/({2s+q})}=o(n_0^{-1/4})$. Fulfilling this requirement can be difficult, particularly when dealing with a multi-dimensional confounding variable and lacking source data. For instance, under a standard regularity condition where $g_0$ possesses continuous second-order derivatives and assuming $q=10$, we have $O(n_0^{-2/(4+10)})=O(n_0^{-1/7})$. Consequently, the condition $n_0^{-{s}/({2s+q})}=o(n_0^{-1/4})$ may prove to be quite restrictive.
In contrast, with the inclusion of source data, Theorem \ref{thediffh} indicates that if the conditions $s_1^{1/2}(\log N)^2 N^{-1/2}=o(n_0^{-1/4})$ and $p^{1/2}N^{-{s}/({2s+p\log q})}=o(n_0^{-1/4})$ are met, then the estimator of $\bR$ will achieve a convergence rate faster than $n_0^{-1/4}$. Given that $s_1=\max\{Kd, Kp, S\}$, we can set $s_1=S$ for a sufficiently large network used in the analysis. These conditions are satisfied if the network size $S$ is less than $(\log N)^{-2} N^{1/2}/n_0^{1/4}$ and the total sample size from the source domains $N$ exceeds $n_0^{(2s+p\log q)/(4s)}$. Hence, a sufficiently large amount of source data can ensure the satisfaction of these conditions.

\section{Numerical studies} 
\label{sec:emp-results}

In this section, we evaluate the performance of RTL via numerical studies.
We first describe the simulation results and then illustrate the applications of RTL
on two real-world datasets.

\subsection{Simulation studies}\label{subsec:sim-study}

In this section, we evaluate the finite sample performance of RTL using simulated data.  We generate data under various designs and compare our method with
the existing approaches.

\subsubsection{Data generating models}
We consider the data generation models as described in \eqref{s1} and \eqref{eq:model}.
We consider the following two scenarios:

\begin{description}
  \item [(a)] Homogeneous models: In this scenario, the source and target domain models are the same. Thus, $\bbeta_{*k} = \bbeta$ and $\bgamma_{*k} = \bgamma$ for all $k=0, 1, \dots, K$.
  The coefficients of $\bbeta$ and $\bgamma$ are specified by i.i.d. drawn from standard normal distribution;
\item [(b)] Heterogeneous models: In this scenario, the source and target domain models are different.  The elements of $\bbeta_{*k}$ and $\bgamma_{*k}$ are specified by  drawing i.i.d. random numerbers from standard normal distribution for all $k=0, 1, \dots, K$.
\end{description}
Covariates $\bX_k$ and $\bZ_k$ are drawn from i.i.d. uniform distribution on $[-1,1]$.
We consider two types of representation functions $\bR(\cdot)$:
\begin{description}
  \item [(a)] (Additive Model) $\bR(\bZ) = [f_1(z_1), f_2(z_2), \dots, f_r(z_r)]\tr$ where $f_i$'s are univariate functions;
  \item [(b)] (Additive Factor Model) $\bR(\bZ) = [f_1(\tilde{z}_1), f_2(\tilde{z}_2), \dots, f_r(\tilde{z}_r)]\tr$ where $f_i$'s are univariate functions and $\widetilde{\bZ} = \bm{B} \bZ$ for some transformation matrix $\bm{B}$. We generate $\bm{B}$ by drawing i.i.d. random numbers from $N(0, 1/q)$.
\end{description}

\subsubsection{Evaluation}
The performance of the estimated regression function
$\widehat{\mu}(\bX, \bZ)= \bX\tr \widehat{\bbeta}_0 +\widehat{\bgamma}_0\tr \widehat{\bm R}(\bZ)
$
is evaluated according to the prediction error and estimation error.
The prediction performance is evaluated by the empirical mean squared error
computed on a test set of size $n_\text{test}$
generated from the target data distribution, i.e.,
$
  \widehat{\operatorname{MSE}}_0
  = n_{\text{test}}^{-1} \sum_{i=1}^{n_\text{test}} \left\{ \widehat{\mu}(\bX_i, \bZ_i) - \mu(\bX_i, \bZ_i) \right\}^2 ,
$
which is an estimator of the mean squared error $\operatorname{MSE} = \mathbb{E}\{ \left[ \widehat{\mu}(\bX, \bZ) - \mu(\bX, \bZ) \right]^2 \}$.
The estimation error is reported on the linear part of the target data,
$
  \operatorname{Err}_{\beta_0}
  = \| \widehat{\bbeta}_0 - \bbeta_0 \|_2 ,
$
where $\widehat{\bbeta}_0$ is the estimator of $\bbeta_0$.

\subsubsection{The effect of the source data sample size} \label{subsubsec:exp1}

Given that the main objective of transfer learning is to use the information from source data to enhance the analysis of target data, we initially assess the performance of RTL as the sample size of the source data varies. In the experiments conducted here, the dimension of the linear component $\bX$ is fixed at $d=5$, and the dimension of the non-linear component $\bZ$ is set at $q=10$. We consider $K=6$ source datasets in total. Additionally, the dimension of the representation function $\bR$ is set to $5$.

{Let the univariate functions $f_i$'s be randomly chosen from $\sin(z_1)$, $2\sqrt{|z_2|} - 1$, $(1 - |z_3|)^2$, $1 / \{1 + \exp(-z_4)\}$, $\cos(\pi z_5 / 2)$.}
The dimension of representation function in the working model, denoted as $r,$ is set as $r=1,3,5,7$, and $9$.
When $r=5$, the representation dimension is the same as the true representation function
used in the data generating model.
The under-estimating and over-estimating models are also considered when we set $r=1, 3$ and $r=7, 9$, respectively.
The sample size in the source dataset  $n_0 = 10, 200, 400, 600, 800, 1000$, and  $1200$, and the sample size in the target dataset is set fixed as $50$.

\begin{figure}[H]
  \centering
  \begin{minipage}{\linewidth}
    \centering
    \includegraphics[width=2.8 in, height=2.5 in]{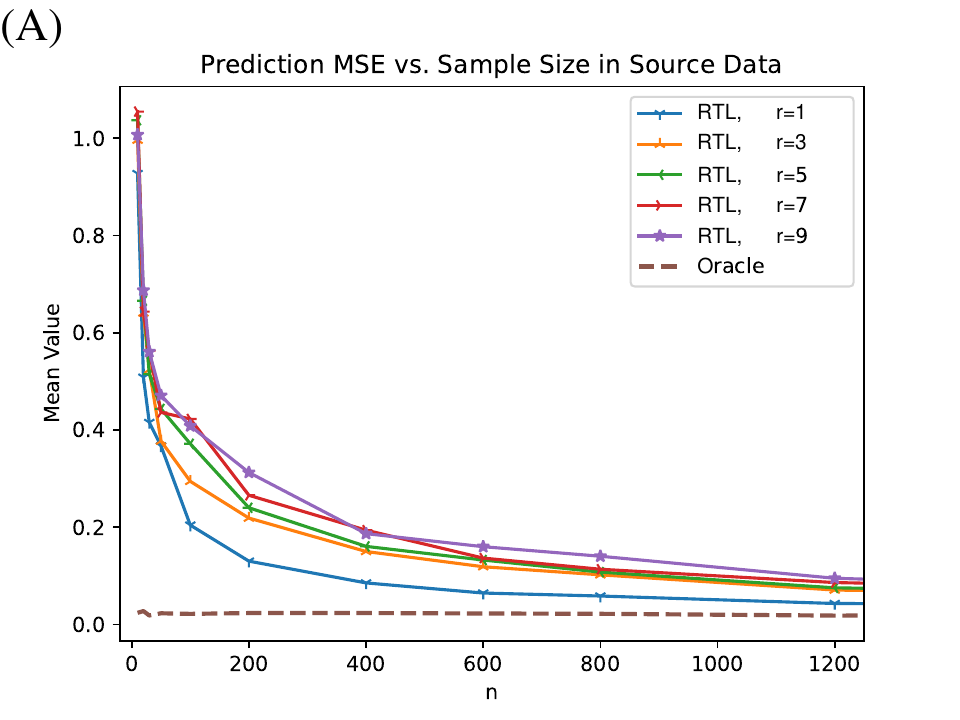}
    \includegraphics[width=2.8 in, height= 2.5 in]{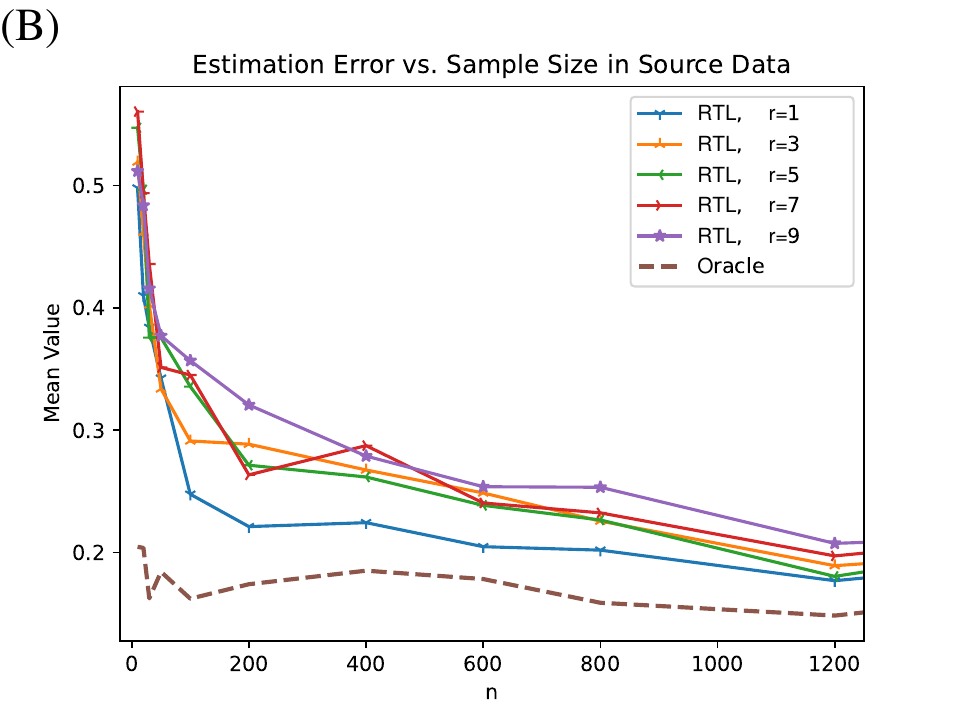}
    \caption{Additive Model with homogeneous coefficients}
    \label{fig:sim.additive.homogeneous}
  \end{minipage}
\end{figure}

\begin{figure}[H]
  \begin{minipage}{\linewidth}
    \centering
    \includegraphics[width=2.8 in, height= 2.5 in]{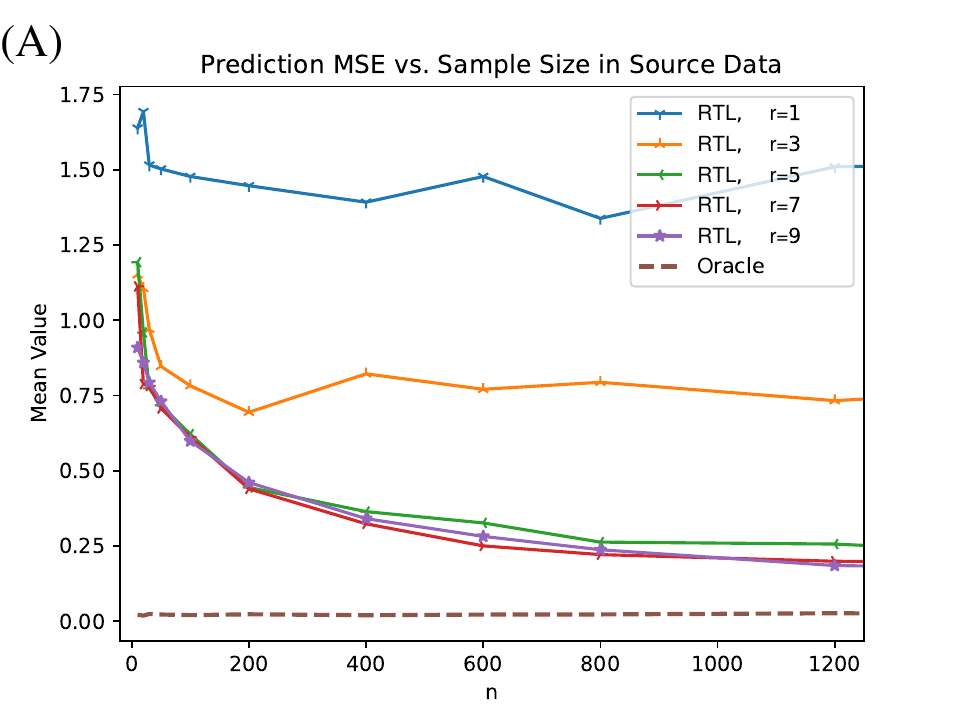}  
    \includegraphics[width=2.8 in, height= 2.5 in]{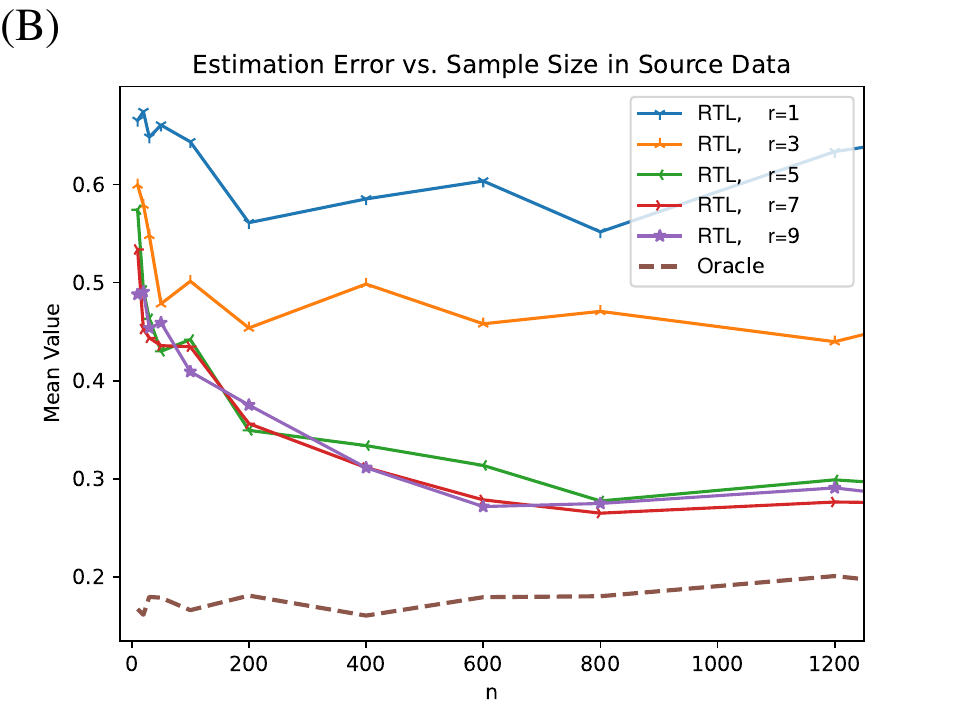}
    \caption{Additive Model with heterogeneous coefficients}
    \label{fig:sim.additive.heterogeneous}
  \end{minipage}
\end{figure}

\begin{figure}[H]
  \centering
   \begin{minipage}{\linewidth}
     \centering
     \includegraphics[width=2.8 in, height= 2.5 in]{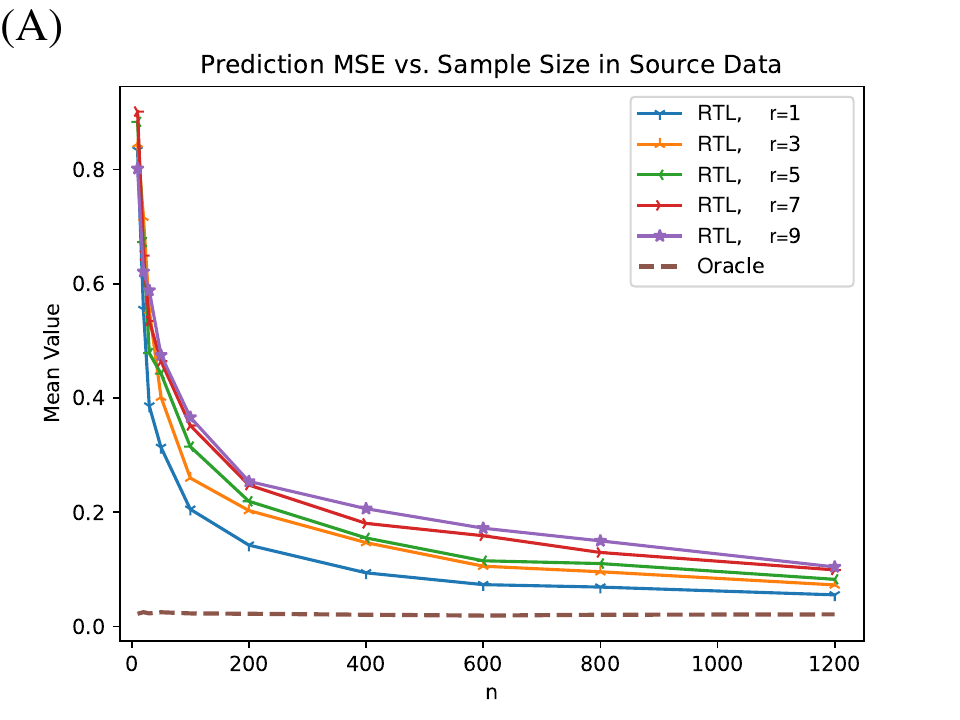} 
     \includegraphics[width=2.8 in, height= 2.5 in]{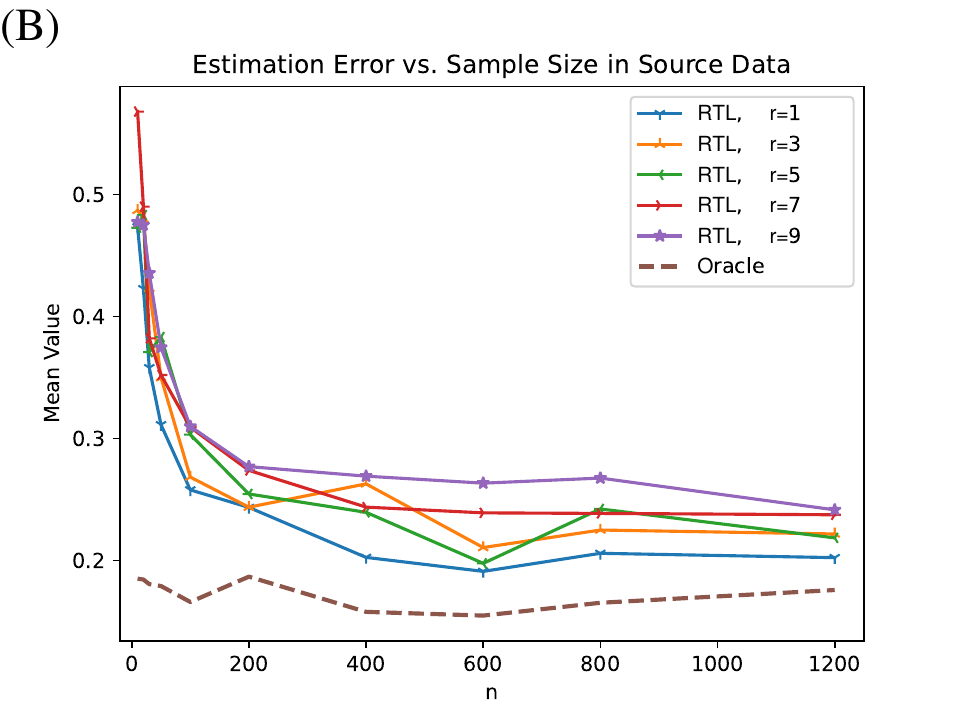}
     \caption{Additive Factor Model with homogeneous coefficients}
     \label{fig:sim.deep.homogeneous}
   \end{minipage}

   \begin{minipage}{\linewidth}
     \centering\includegraphics[width=2.8 in, height= 2.5 in]{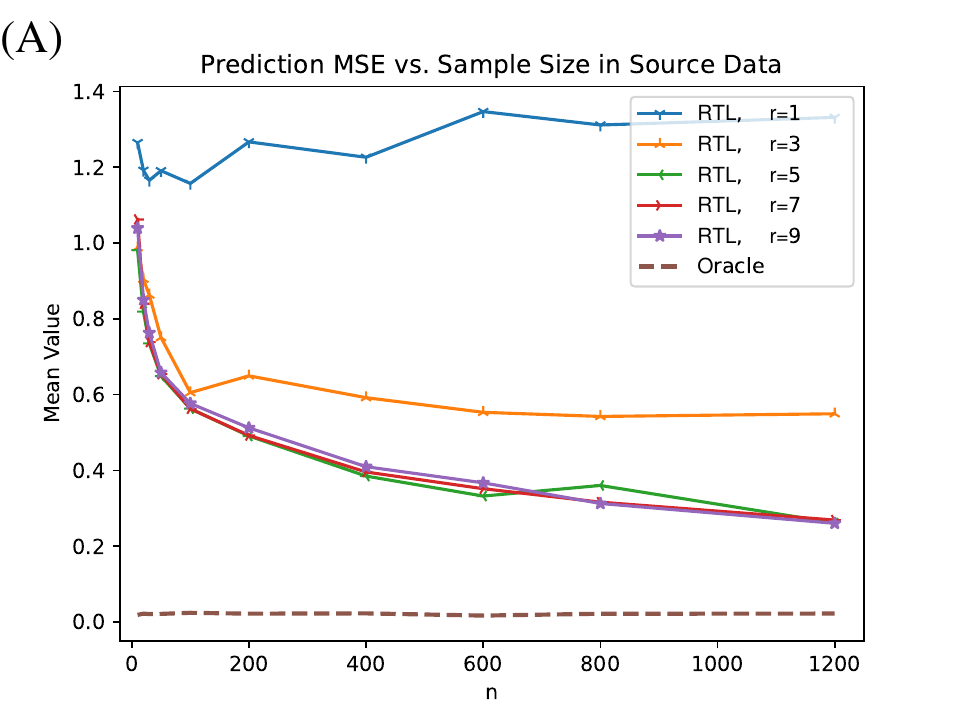} 
     \includegraphics[width=2.8 in, height= 2.5 in]{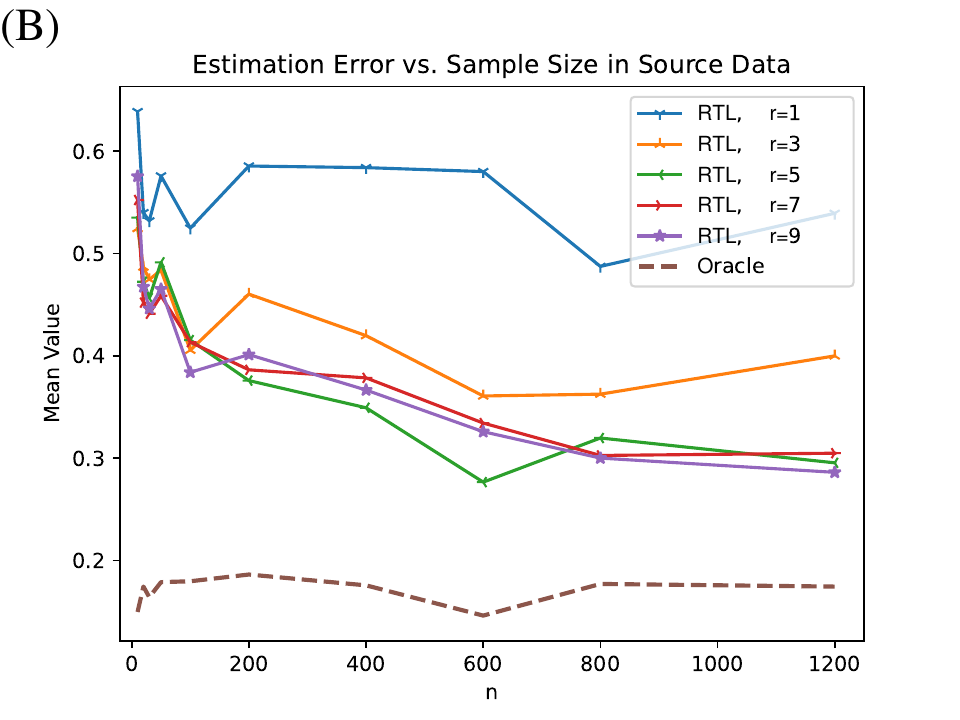}
     \caption{Additive Factor Model with heterogeneous coefficients}
     \label{fig:sim.deep.heterogeneous}
   \end{minipage}
 \end{figure}

 We repeat the experiments $50$ times and report the average performance.
{We  also report the `Oracle' method which uses the true representation function in the target data. }
Figures~ \ref{fig:sim.additive.homogeneous} to  \ref{fig:sim.deep.heterogeneous} present prediction MSEs for RTL with the number of representation functions $r=1,3,5,7,9.$ The true value of $r=5$ in
the generating model.
For the Additive Model Design, the depth of the neural network is set as $D=4$ and the width is set as $W=300$. The results are shown in Figures~\ref{fig:sim.additive.homogeneous} and \ref{fig:sim.additive.heterogeneous}. For the Additive Factor Model Design, the depth of the neural network is set as $D=6$ and the width is set as $W=500$.
The results are shown in Figures~\ref{fig:sim.deep.homogeneous} and \ref{fig:sim.deep.heterogeneous}.

The experimental results indicate that as the sample size increases, the performance of the RTL method approaches that of the {oracle} estimator, provided that the dimension of the representation function is close to or exceeds the true dimension of the representation function in the data-generating model. However, if the chosen dimension of the representation function is less than the true dimension present in the generating model, the proposed method exhibits suboptimal performance. Consequently, in practical applications, it is advisable to set the dimension of the representation function to a higher rather than lower value to ensure better performance.

\subsubsection{Comparison} \label{subsec:numerical_comparison}
We consider both the additive model design and the additive factor model design as previously described. Additionally, we explore a more intricate deep model design to simulate the data generation process, as illustrated in Figure~\ref{fig:deep_model}. In this model, the functions $f_i$ and $h_i$ are selected randomly from a pool of functions that includes $\sin(x)$, $-\cos(x)$, $\cos(2x)$, $\sin(\pi x)$, $\cos(\pi x)$, $2\sqrt{x+0.5}-1$, $(1-|x-0.5|)^2$, $1 / {1 + \exp(x)}$, $\tan(x+0.1)$, $\log(x+1.5)$, $\exp(x)$, $x^2$, and $\arctan(x)$. For instance, the output of the first node in the second layer is computed as $f_1(z_1 + z_2)$. We utilize $K=20$ source datasets and $1$ target dataset, and we assess two configurations regarding the model's dimension and sample size. In the first configuration, each source dataset comprises $200$ samples, and the dimension of the non-linear component is set to $q=20$. In the second configuration, the sample size for each source dataset is increased to $400$, while the dimension of the non-linear component is reduced to $q=10$. For both configurations, the target dataset consists of $n_0=50$ samples, with the dimension of the non-linear component fixed at $d=5$.

\begin{figure}[H]
  \centering
  \begin{tikzpicture}[every node/.style = {align=center}]

    \node (f1) at (0, -1) {$f_{1}()$};
    \node (z1) at (-0.5, -2) {$z_1$};
    \node (z2) at (0.5, -2) {$z_2$};

    \node (f2) at (2, -1) {$f_{2}()$};
    \node (z3) at (1.5, -2) {$z_3$};
    \node (z4) at (2.5, -2) {$z_4$};

    \node (f3) at (4, -1) {$f_{3}()$};
    \node (z5) at (3.5, -2) {$z_{5}$};
    \node (z6) at (4.5, -2) {$z_{6}$};

    \node (f4) at (6, -1) {$f_{4}()$};
    \node (z7) at (5.5, -2) {$z_{7}$};
    \node (z8) at (6.5, -2) {$z_{8}$};

    \node (f5) at (7, -1) {$f_{5}()$};
    \node (z9) at (7.5, -2) {$z_{9}$};
    \node (z10) at (8.5, -2) {$z_{10}$};

    \node (f6) at (8, -1) {$f_6()$};

    \node (h1) at (1, 0) {$h_1()$};
    \node (h2) at (3, 0) {$h_2()$};
    \node (h3) at (5, 0) {$h_3()$};
    \node (h4) at (7, 0) {$h_4()$};
    \node (h5) at (8, 0) {$h_5()$};

    \draw (h1) -- (f1);
    \draw (h1) -- (f2);
    \draw (h2) -- (f2);
    \draw (h2) -- (f3);
    \draw (h3) -- (f3);
    \draw (h3) -- (f4);
    \draw (h4) -- (f4);
    \draw (h4) -- (f5);
    \draw (h5) -- (f5);
    \draw (h5) -- (f6);

    \draw (f1) -- (z1);
    \draw (f1) -- (z2);
    \draw (f2) -- (z3);
    \draw (f2) -- (z4);
    \draw (f3) -- (z5);
    \draw (f3) -- (z6);
    \draw (f4) -- (z7);
    \draw (f4) -- (z8);
    \draw (f5) -- (z8);
    \draw (f5) -- (z9);
    \draw (f6) -- (z9); 
    \draw (f6) -- (z10);
  \end{tikzpicture}
  \caption{The architecture of a deep model with $q=10$ and $p=5$ used in Exp 2.}
  \label{fig:deep_model}
\end{figure}
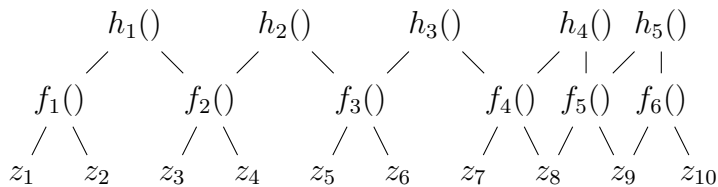

We consider the following competitor methods.
\begin{description}
  \item [(a)] The ``Pool" method is a parametric pooling method which estimates the coefficients using a combined loss. Pooled regression (PR) assumes that all parameters across different individuals are the same. All datasets are pooled together.
  \item [(b)] The ``MAP" method represents the model averaging transfer learning method \citep{zhang2024prediction}.
  \item [(c)] The ``Trans-lasso" method represents the high-dimensional linear regression transfer learning method \citep{li2022transfer}.
  \item [(d)] The ``Meta" method represents the meta-analysis method where the coefficients from different datasets are weighted based on the inverse variance of estimations.
  \item [(e)] The ``STL" method represents the neural network method which  only uses the target data.
\end{description}

To adapt these methods for non-linear models, we express the nonparametric component as a linear combination of cubic spline basis functions. The optimal number of knots is determined through the use of validation samples. For our proposed RTL method, we define the dimension of the representation space to be $p=5$, which reflects the true underlying dimension. The outcomes of this comparison are depicted in Figures~\ref{fig:sim.pred_comparison} and \ref{fig:sim.est_comparison}, where the left side corresponds to the scenario with $n_k=400$ and $q=10$, and the right side pertains to the scenario with $n_k=200$ and $q=20$.
It can be seen that RTL has lower prediction and estimation errors than the existing methods, including Pool, MAP, Trans-Lasso, Meta, and STL. These results show the superiority of our proposed RTL method in terms of both prediction accuracy and estimation quality.

\begin{figure}[H]
  \centering
  \begin{subfigure}{0.45\linewidth}
    \centering
    \caption*{\tiny$\bf{(A)}$\raggedright}
    \includegraphics[width=2.8 in, height=2.4 in]{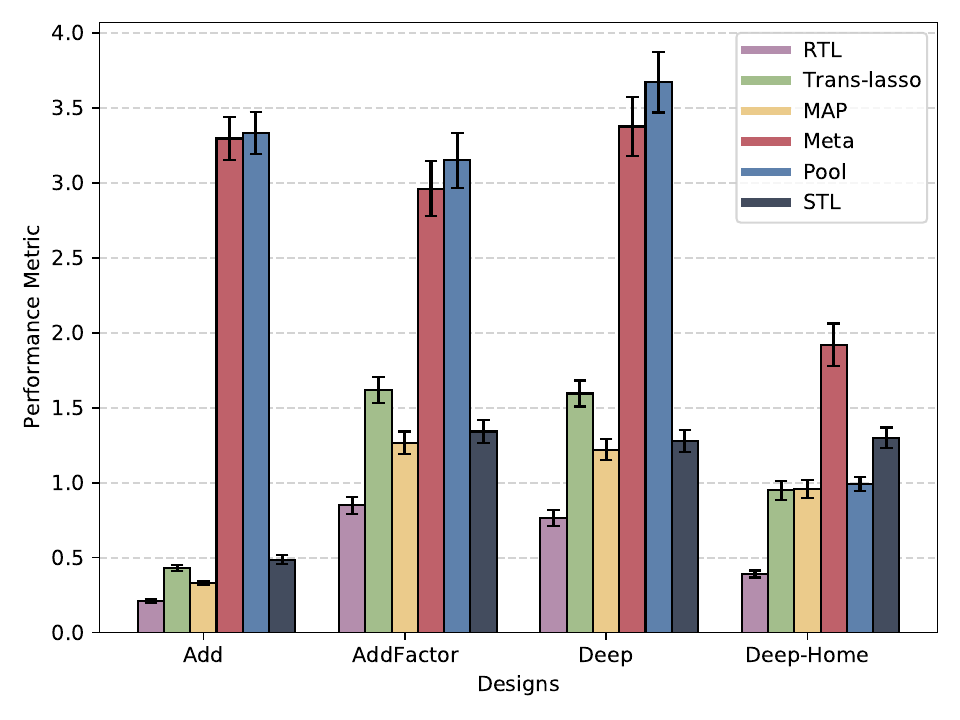}
  \end{subfigure}
  \begin{subfigure}{0.45\linewidth}
    \centering
    \caption*{\tiny$\bf{(B)}$\raggedright}
    \includegraphics[width=2.8 in, height=2.4 in]{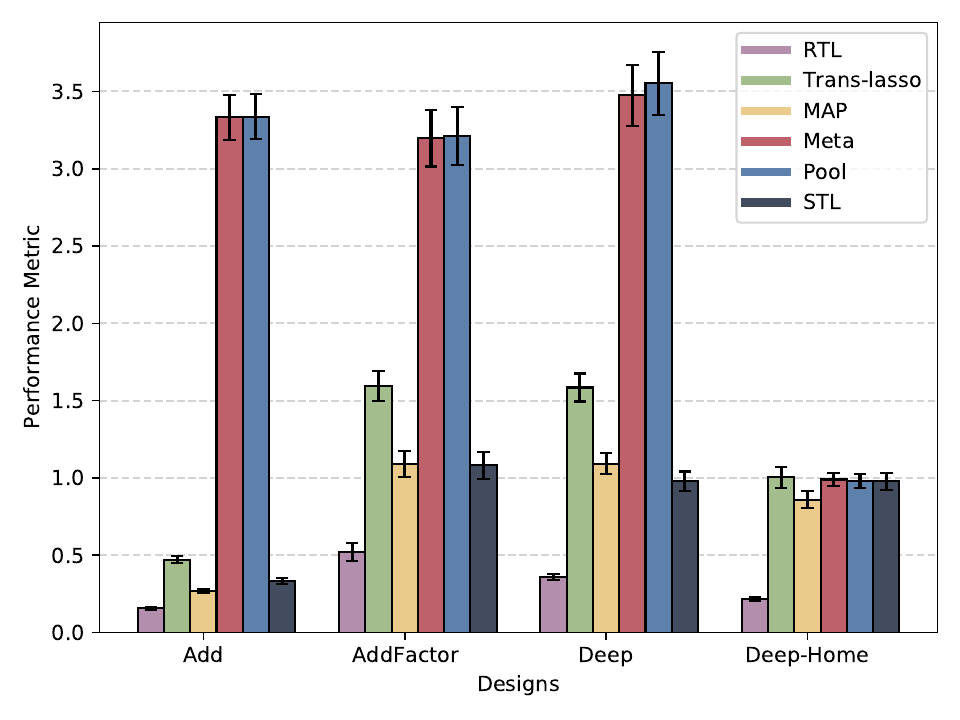}
  \end{subfigure}
  \caption{Prediction performance comparison of different methods across different designs.
  (A) The left-hand-side corresponding to the case of $n_k=200$ and $q=20$; (B) the right-hand-side corresponding to the case of $n_k=400$ and $q=10$. It can be seen that RTL has lower prediction errors than the existing methods, including Pool, MAP, Trans-Lasso, Meta, and STL.
  }
  \label{fig:sim.pred_comparison}
\end{figure}

\begin{figure}[H]
  \centering
  \begin{subfigure}{0.45\linewidth}
    \caption*{\tiny$\bf{(A)}$\raggedright}
    \includegraphics[width=2.8 in, height=2.4 in]{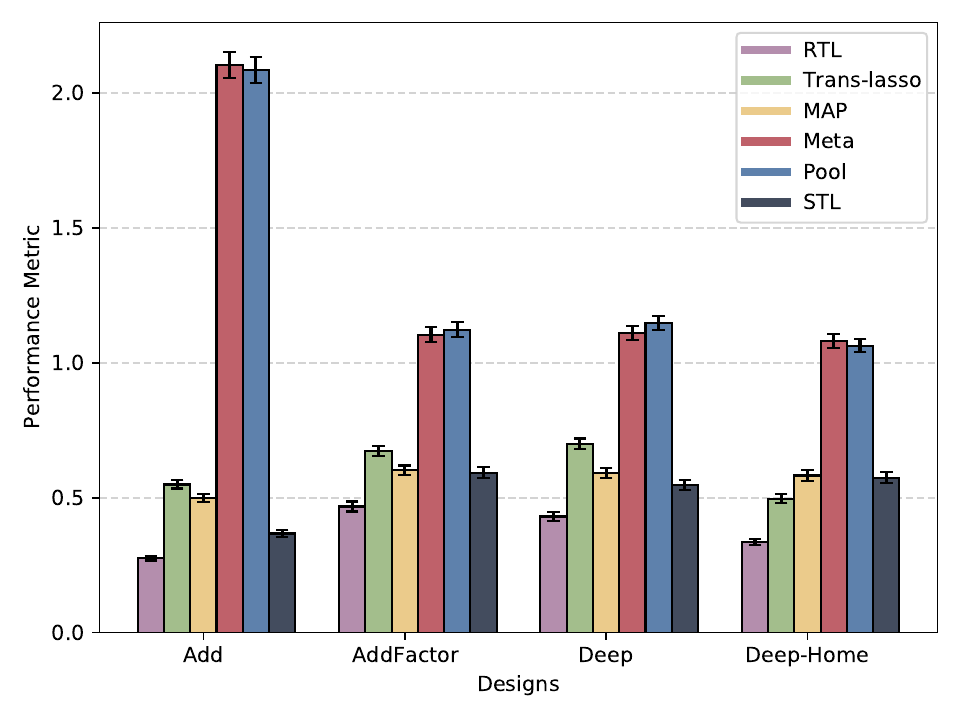}
  \end{subfigure}
  \centering
  \begin{subfigure}{0.45\linewidth}
    \caption*{\tiny$\bf{(B)}$\raggedright}
    \includegraphics[width=2.8 in, height=2.4 in]{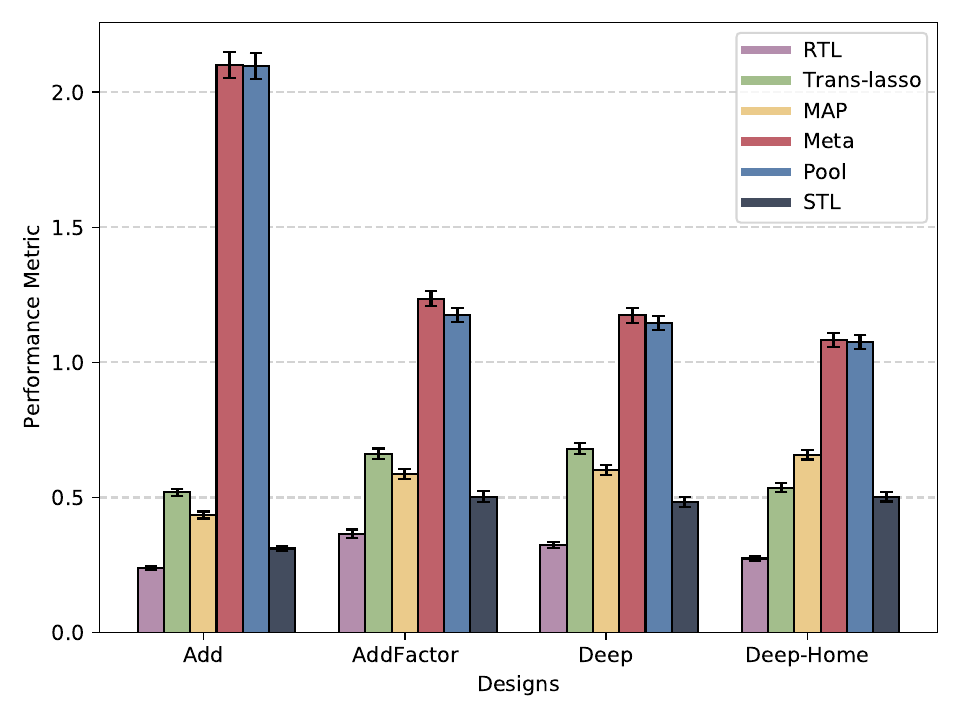}
  \end{subfigure}
  \caption{Estimation performance comparison between RTL and Pool, MAP, Trans-Lasso,
  Meta, and STL across different designs.
  (A) The left-hand-side corresponding to the case of $n_k=200$ and $q=20$;
  (B) the right-hand-side corresponding to the case of $n_k=400$ and $q=10$.
  It can be seen that RTL has lower estimation errors than the existing methods, including Pool, MAP, Trans-Lasso, Meta, and STL.}
  \label{fig:sim.est_comparison}
\end{figure}

\begin{table}[H]
  \centering
  \small
  \caption{Illustration of asymptotic variance and normality. }
  \label{tab:asy_variance}
  \begin{tabular}{crm{1.5cm}m{2cm}m{2.5cm}}
    \toprule
    Design & Avg. Bias & SD  &  SE  & Normality \\
    \midrule
    \multirow{2}{*}{\shortstack{Additive\\$n_k=200$}} & 0.0040 & 0.3016 & $0.2953_{\ 0.0459}$  & \includegraphics[width=1.5cm]{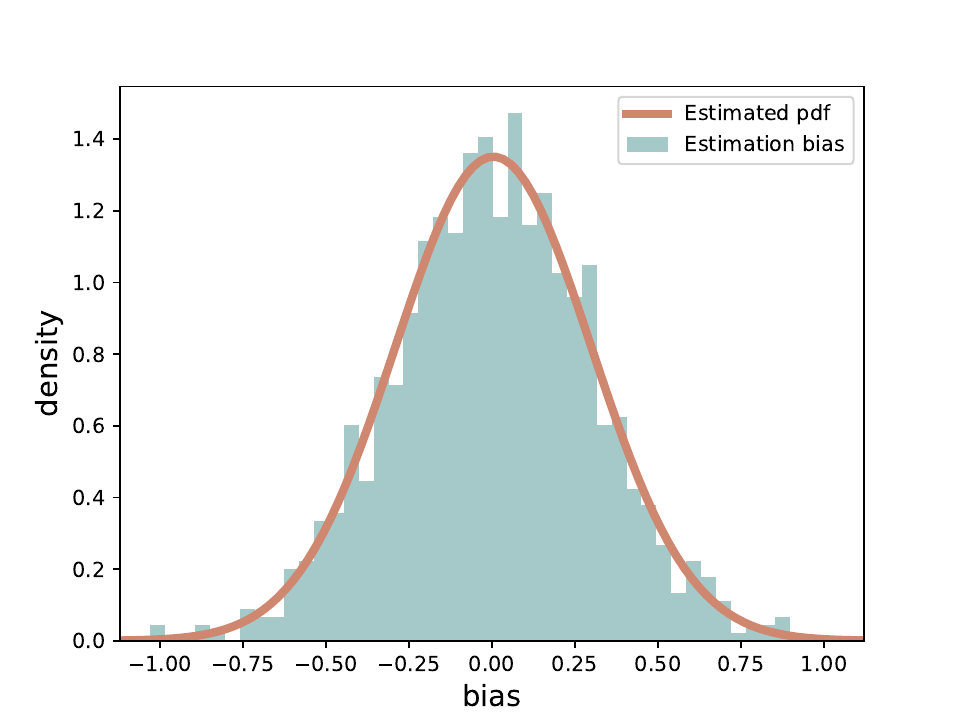}\\
    \multirow{2}{*}{\shortstack{Additive\\$n_k=400$}} & -0.0146 & 0.2321 & $0.2241_{\ 0.0386}$ & \includegraphics[width=1.5cm]{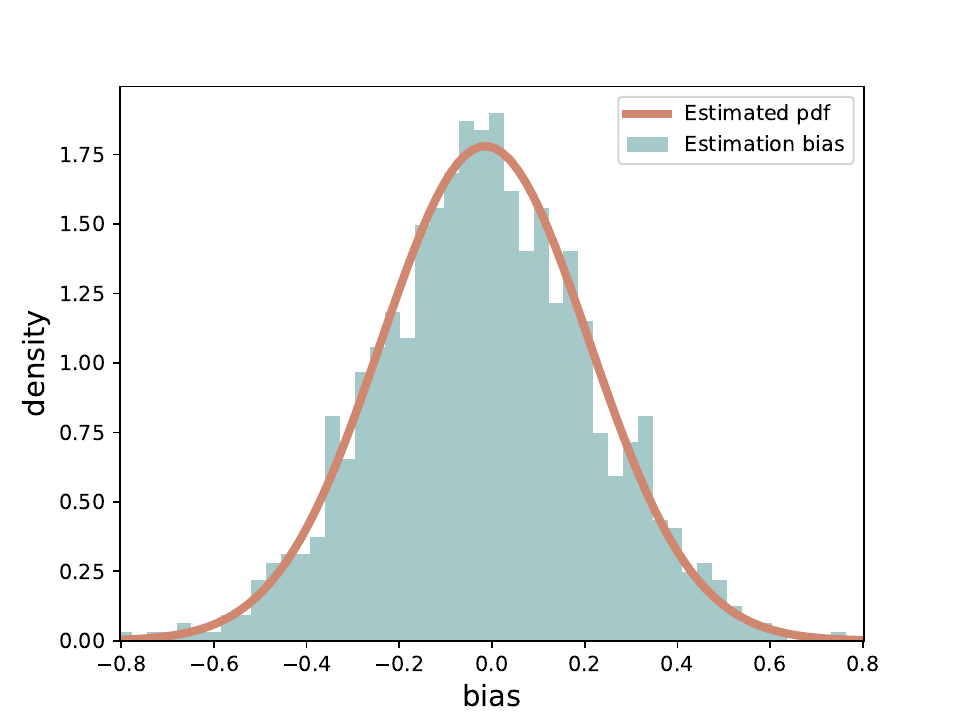}\\
    \multirow{2}{*}{\shortstack{Add-Factor\\$n_k=200$}} & 0.0001 & 0.2945 & $0.2880_{\ 0.0462}$ & \includegraphics[width=1.5cm]{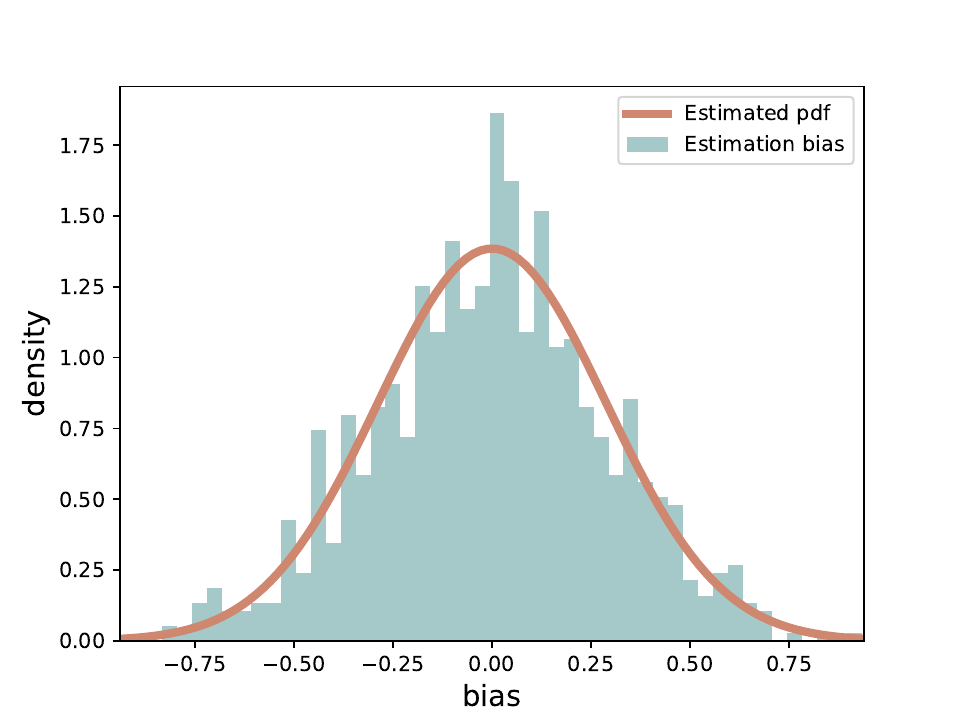}\\
    \multirow{2}{*}{\shortstack{Add-Factor\\$n_k=400$}} & -0.0153 & 0.2135 & $0.2044_{\ 0.0345}$ & \includegraphics[width=1.5cm]{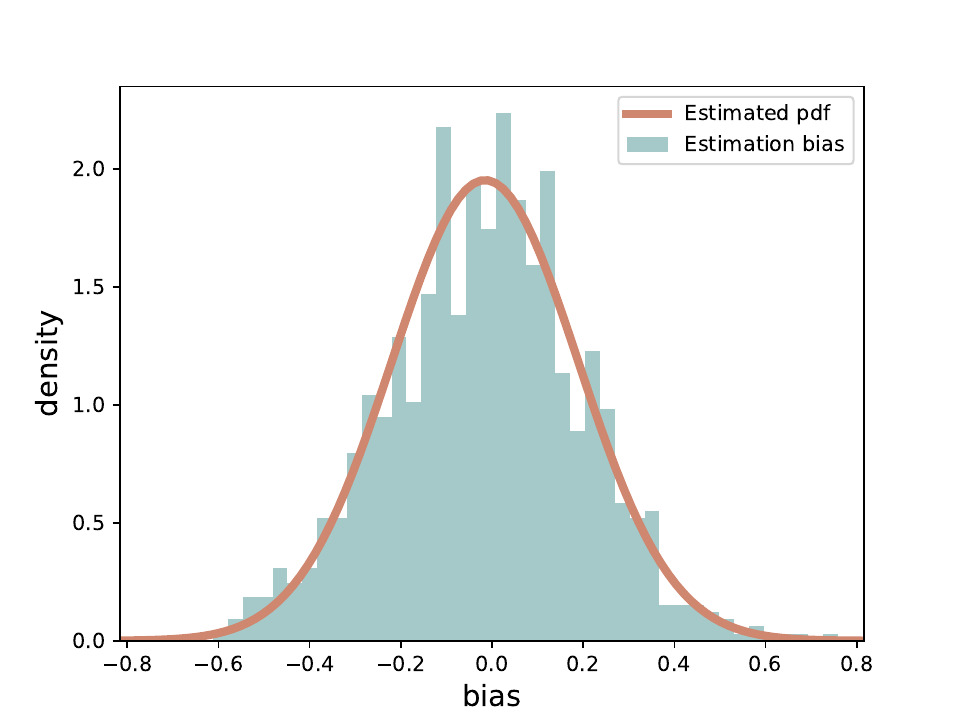} \\
    \multirow{2}{*}{\shortstack{Deep\\$n_k=200$}} & -0.0080 & 0.2155 & $0.2140_{\ 0.0350}$ & \includegraphics[width=1.5cm]{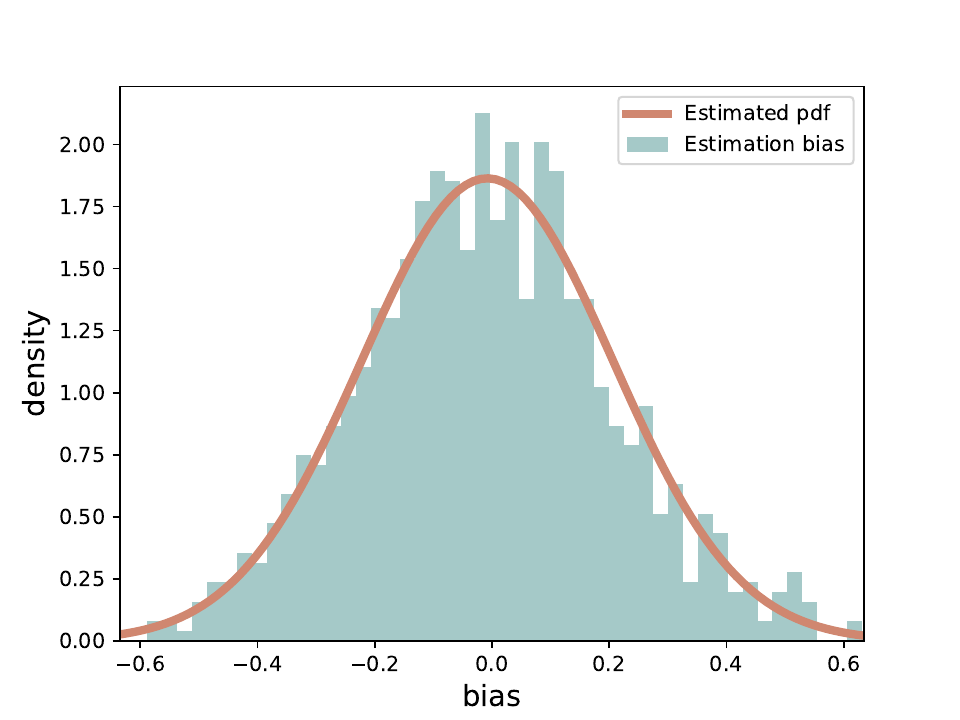}\\
    \multirow{2}{*}{\shortstack{Deep\\$n_k=400$}} & -0.0056 & 0.1350 & $0.1350_{\ 0.0228}$ & \includegraphics[width=1.5cm]{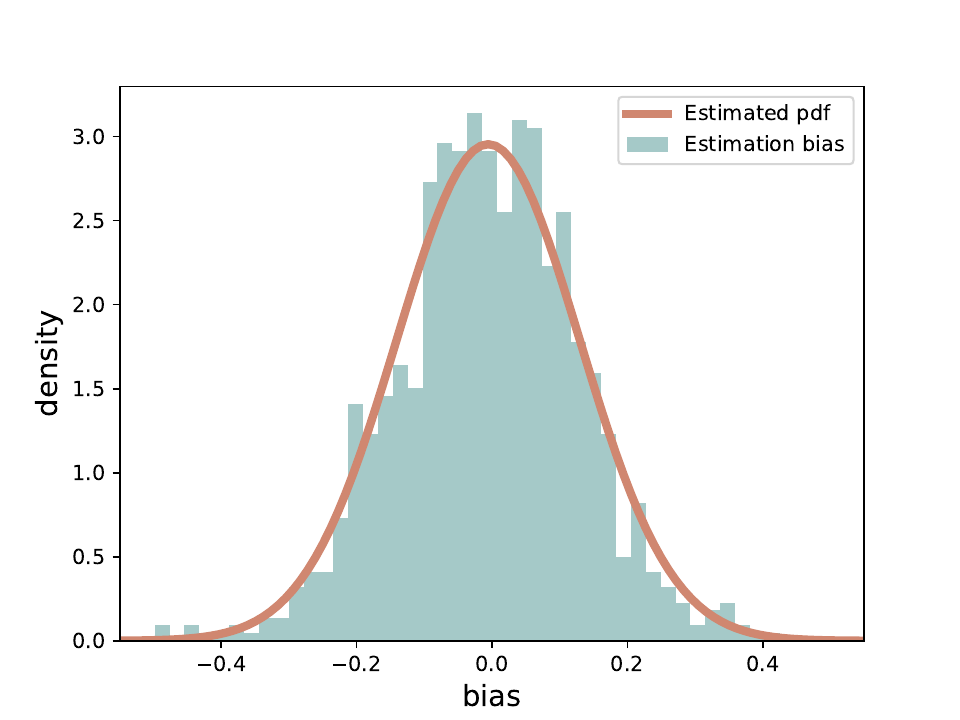}\\
    \multirow{2}{*}{\shortstack{Deep(Homo)\\$n_k=200$}} & -0.0061 & 0.1567 & $0.1527_{\ 0.0242}$ & \includegraphics[width=1.5cm]{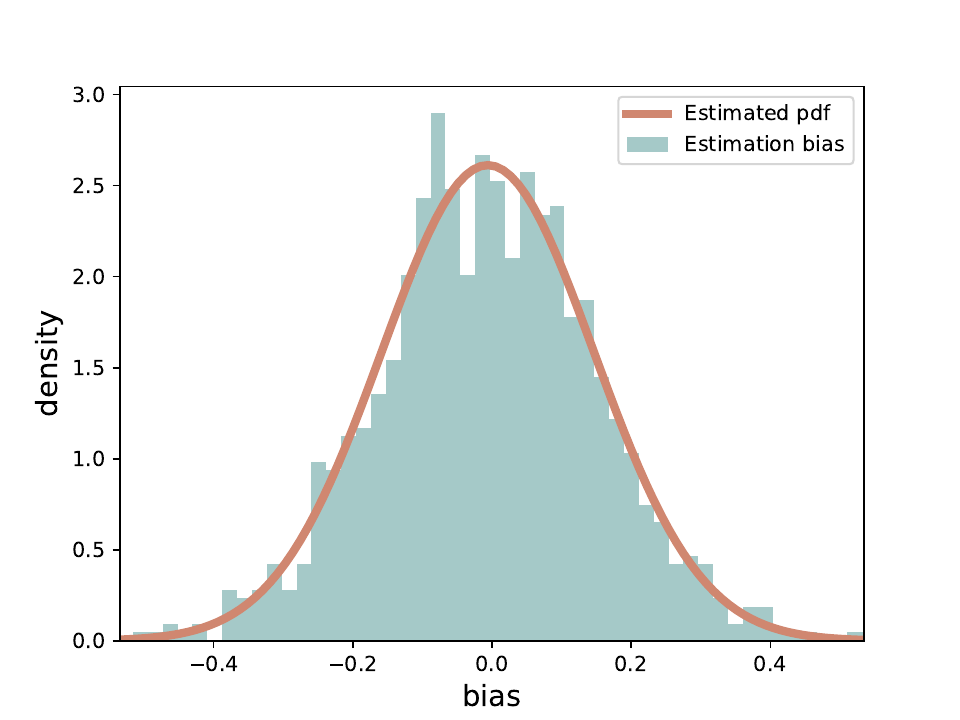} \\
    \multirow{2}{*}{\shortstack{Deep(Home)\\$n_k=400$}} & -0.0003 & 0.1158 & $0.1117_{\ 0.0184}$ & \includegraphics[width=1.5cm]{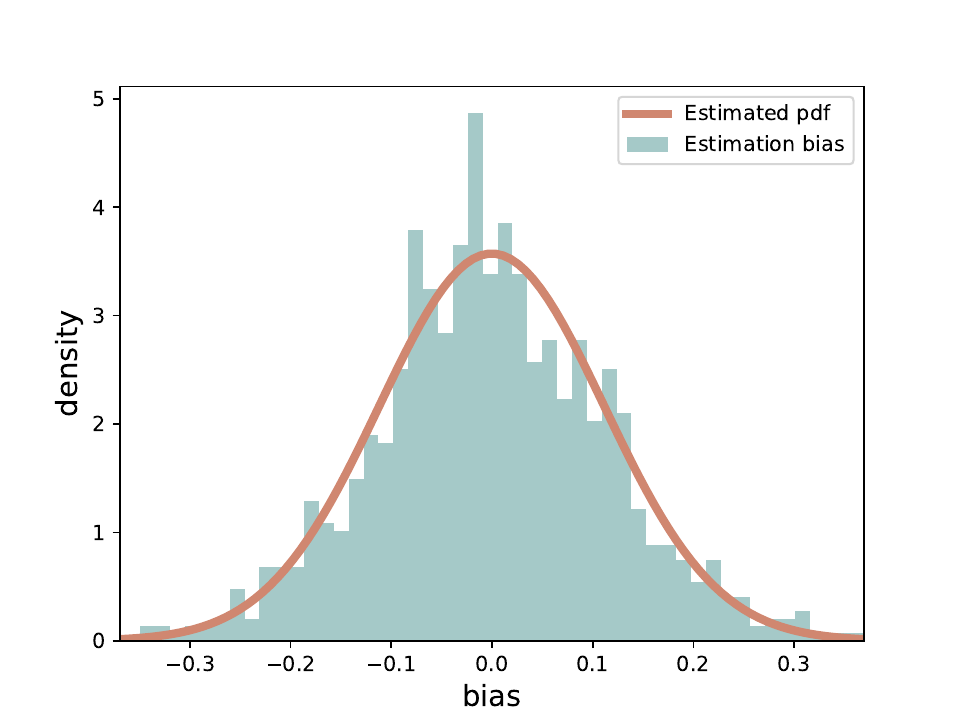} \\
    \bottomrule
  \end{tabular}
\end{table}

\subsubsection{Assessment of variance estimation and asymptotic normality}

To demonstrate the estimated asymptotic variance, we employ the same experimental setup as described in Section~\ref{subsec:numerical_comparison} with the exception that the parameters $\bbeta_{*0}$ and $\bgamma_{*0}$ are both set to $\bm{1},$ and the number of source dataset is set to $6$. Our analysis concentrates on the transformed coefficients $\theta = \bm{\alpha}^{\top}\bbeta_{*0}$, where $\bm{\alpha} = \bm{1} / \sqrt{p}.$

Table~\ref{tab:asy_variance} presents the results, including the average bias and standard deviation (SD) of a combined coefficient, the mean of the estimated variance (SE), and an illustration of normality, all based on $1000$ repetitions. The last column of Table~\ref{tab:asy_variance} displays
the  histograms of the estimation biases alongside the asymptotic normal distributions with the estimated means and variances.
The results suggest that the distribution  RTL estimator is well approximated by normal distribution.

\subsection{Semi-synthetic MNIST data}\label{subsec:synthetic}
In this section, we apply RTL to a semi-synthetic arithmetic dataset constructed from the MNIST dataset \citep{lecun1998mnist}. The MNIST dataset is a collection of handwritten digits, comprising $70,000$ samples with 10 class labels each represented by a $28 \times 28$ grayscale image.
For our semi-synthetic scenario, we define the data generating process as follows:
\begin{align}
  Y = \beta X + \bgamma^\top \bR\bigl(\bZ\bigr) + \epsilon,
\end{align}
where $X$ is a random sample drawn from the standard normal distribution, $\bZ$ is a digit image sampled from the MNIST dataset, $\beta \in \mathbb{R}$ and $\bgamma \in \mathbb{R}^{10}$ are unknown coefficients to be estimated, $\bR \in \mathbb{R}^{10}$ represents the one-hot encoded label corresponding to the input image, and $\epsilon$ is noise that follows the standard normal distribution.
To illustrate, suppose we set $\beta = 1$, $\bgamma = \bm{1}$, $X=1$, $\epsilon=0.2$, and let $\bZ$ correspond to the image $\vcenteredinclude{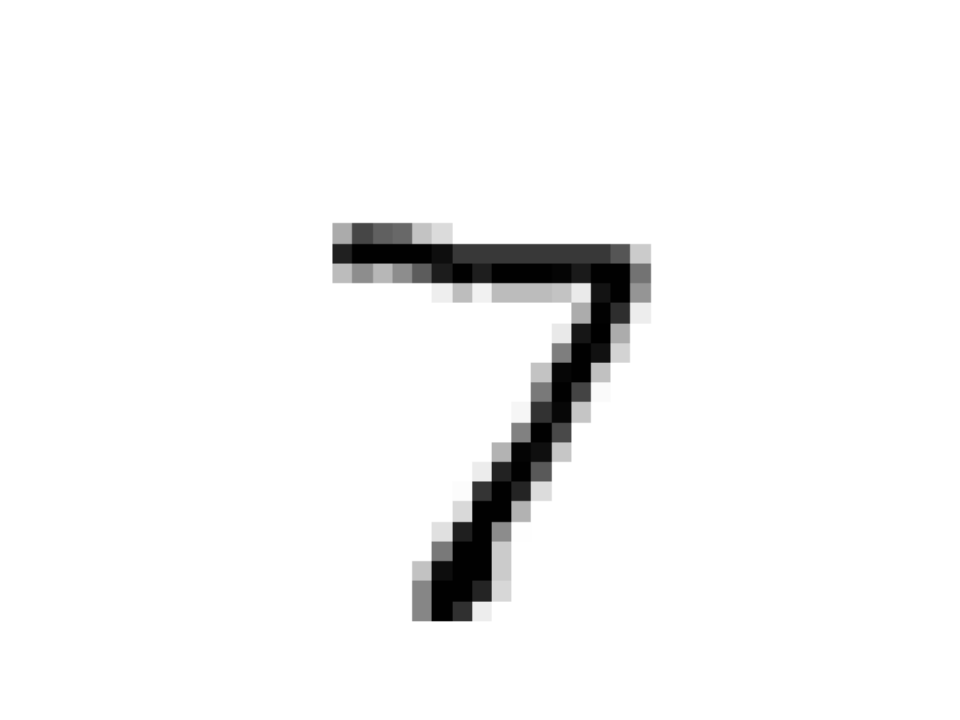}$. The resulting value of $Y$ would be $1 + \bm{1}^\top \bR\bigl(\vcenteredinclude{figs/digits.pdf}\bigr) + 0.2$, which equals $8.2$ if the image represents the digit `7'.

We adopt the following experimental setup:
The total number of source datasets is fixed at $10$. Each source dataset is composed of a training subset, which accounts for $40\%$ of the MNIST training set, and a separate validation subset consisting of $500$ samples randomly selected from the same training set. The coefficients $\beta$ and $\gamma_1$ are randomly assigned for each source dataset, and we define $\bgamma$ as $\gamma_1 \bm{1}$ across all source data. The target dataset is limited to $100$ samples, also drawn from the MNIST training set.

For the representation learning, we utilize a neural network with $7$ hidden layers, which includes $5$ convolutional layers and $2$ fully connected layers.
{More detailed information on the architecture of the neural network is provided in the Supplementary Materials.}
The output dimension of the representation network is set to $10$. We adopt an iterative training approach for the representation network and the subsequent linear layer, spanning $10$ epochs. The learning rate is established at $10^{-4}$, and we use a batch size of $128$.

The average performance of the estimated representation network, in terms of prediction, classification, and estimation based on 5 replications, is presented in Table~\ref{tab:res_mnist}.
It is important to note that the prediction error is evaluated on a test set derived from the MNIST test set, which comprises $10,000$ images.
The classification accuracy is assessed using the MNIST training set.
For this evaluation, the estimated representation network is augmented with two additional linear layers with ReLU activation functions, and the final output undergoes a transformation via a logarithmic softmax function.

\begin{table}[H]
  \centering
  \small
  \caption{The prediction and estimation errors of RTL in the synthetic MNIST data analysis.}
  \label{tab:res_mnist}
  \begin{tabular}{cccc}
    \toprule
    Method & Prediction Error & Estimation Error & Classification Accuracy \\
    \midrule
    RTL & $0.2163 (0.0675)$ & $0.0701 (0.0716)$ & $98.69\% (0.24\%)$ \\
    \bottomrule
  \end{tabular}
\end{table}

The results show that RTL has a good performance in the synthetic MNIST data analysis for the prediction and estimation in target domain.
The high classification accuracy indicates that the estimated representation network is able to capture the label information of the input images.

\subsection{Rental data} \label{subsec:real-world}

We demonstrate the application of the proposed RTL method using an apartment rental dataset from three major Chinese cities: Beijing, Shanghai, and Shenzhen. This dataset was obtained from a publicly accessible website, available at \url{http://www.idatascience.cn/}. The number of available rental apartments across various districts in these cities is detailed in Table~\ref{tab:housing_data}. Additionally, the variables included in the dataset are given in Table~\ref{tab:housing_variables}. The main goal of our analysis is to assess the influence of key factors, such as neighborhood characteristics and the proximity of schools, on rental prices.

\begin{table}[ht]
  \centering
  \scriptsize
  \caption{
  The number of  apartments for rent by district in Beijing, Shanghai and Shenzhen}
  \label{tab:housing_data}
  \begin{tabular}{lr|lr|lr}
  \toprule
  \multicolumn{2}{c|}{Beijing} & \multicolumn{2}{c|}{Shanghai} & \multicolumn{2}{c}{Shenzhen} \\
  \midrule
  Haidian     & 528  &  Pudong    & 1333  &  Nanshan   & 1524  \\
  Chaoyang    & 1241 &  Xuhui     & 566   &  Futian    & 1169 \\
  Changping   & 310  &  Changning & 432   &  Bao'an    & 1108 \\
  Dongcheng   & 315  &  Putuo     & 416   &  Longgang  & 857  \\
  Xicheng     & 308  &  Huangpu   & 393   &  Luohu     & 857  \\
  Fengtai     & 347  &  Baoshan   & 365   &  Longhua   & 778  \\
  Shijingshan & 269  &  Longhuaa  & 360   &  Buji      & 735  \\
  Mentougou   & 264  &  Jing'an   & 349   &  Guangming & 714  \\
  Fangshan    & 249  &  Yangpu    & 316   &  Yantian   & 543  \\
  Shunyi      & 225  &  Tongzhou  & 223   &            &      \\
  Jiading     & 302  &  Hongkou   & 207   &            &      \\
  Daxing      & 291  &  Fengxian  & 204   &            &      \\
  Huairou     & 162  &            &       &            &      \\
  \bottomrule
  \end{tabular}
\end{table}

\begin{table}[H]
  \centering
  \scriptsize
  \caption{The description of variables in the apartment rental dataset.}
  \label{tab:housing_variables}
  \begin{tabular}{ll}
    \toprule
    Variable & Description  \\
    \midrule
    (y) price & monthly rent of the apartment \\
    \midrule
    (z) room & number of rooms \\
    (z) hall & number of halls \\
    (z) toilet & number of toilets \\
    (z) hasbed & has bed \\
    (z) haswardrobe & has wardrobe \\
    (z) hasac & has air conditioner \\
    (z) hasgas & has gas \\
    (z) floor & 4 categories based on height \\
    (z) totalfloor & total floors of the building \\
    (z) numhospital & number of hospitals (within 3km) \\
    \midrule
    (x) neighborhood & neighborhood of the apartment \\
    (x) numschool & number of schools (within 3km) \\
    \bottomrule
  \end{tabular}
\end{table}

For the apartment rental dataset encompassing three major Chinese cities, Beijing, Shanghai, and Shenzhen, it is important to recognize that these cities, being in distinct regions (north, east, and south) of China, have unique rental markets.  Consequently, pooling the data from Beijing, Shanghai, and Shenzhen for analysis without considering regional differences could lead to questionable conclusions. For instance, what characterizes a neighborhood in Beijing may be different from those in Shanghai or Shenzhen. Additionally, the availability of an air conditioner might influence rental prices differently across these cities due to their varying climatic conditions. Additionally, considering the vast size of these cities, there can be significant variations in rental market dynamics even between different districts within the same city. Therefore, pooling data from different districts within a single city is not advisable.
On the other hand, despite the geographical distinctions and the heterogeneity across districts within these cities, there are inherent similarities within their rental markets. These commonalities make it plausible to apply knowledge about factors affecting rental prices from one city to another. Thus, while regional specificities should not be overlooked, there is merit in exploring the transferability of insights across these diverse urban rental markets.

Given this context, transfer learning is a reasonable approach to analyzing data from a specific district in one city, using data from other districts as source data. This method helps overcome the limitations of small sample sizes for district-specific data and enhances the analysis by leveraging the broader patterns and insights from across the dataset. Thus, while respecting regional specificities, transfer learning offers a way to explore the transferability of insights across these diverse urban rental markets.

In our analysis, we focus on the effects  of two variables: \textit{neighborhood} and
\textit{numschool} (the number of schools within a 3km radius), which are widely recognized as having a significant influence on rental prices in China. We use other factors as confounding variables that may also affect rental prices, albeit to a lesser degree.
We examine four target datasets from four randomly selected districts. These districts include Changping in Beijing, Yantian in Shenzhen, and both Putuo and Fengxian in Shanghai.
When analyzing a specific target dataset, such as the one from Changping in Beijing, we incorporate all the remaining data as source data. This enables us to quantify the effects of neighborhood characteristics and the proximity to schools on rental prices, while also taking into account the wider context provided by the comparative data from other districts.

Using the proposed RTL method for the semiparametric regression model as described in Section \ref{sec:model}, we calculate the $95\%$ confidence intervals for the coefficients of \textit{neighborhood} and \textit{numschool} based on the results in Section \ref{sec:theory}. The findings are presented in Table~\ref{tab:housing_variables_est}. These estimated coefficients provide a quantitative assessment of the impact these variables have on rental prices. Moreover, the fact that these confidence intervals do not include zero indicates that the district where the house is located and the proximity to schools significantly increase the monthly rent, when other factors are held constant.

We also consider the prediction performance via randomly dividing the target dataset into training, validation and testing sets. The training set consists of $30\%$ of the data and the testing set consists of other $30\%$ of the data, while the remaining is allocated to the validation set. The prediction performance is evaluated using the mean squared error on the testing set. Figure~\ref{fig:house_pred} displays the prediction errors the proposed RTL and the existing methods Pool, MAP, Trans-Lasso, Meta, and STL on housing rental information data. It shows that RTL outperforms these existing methods in terms of prediction error.

\begin{table}[H]
  \centering
  \small
  \caption{The estimated coefficients and confidence intervals for the variables of main interest in the housing rental data.}
  \label{tab:housing_variables_est}
  \begin{tabular}{llrrr}
    \toprule
    District & Variable &  Estimate & SE & $95\%$ CI \\
    \midrule
    \multirow{2}{*}{Changping} & neighborhood       & $38.08$  & $3.37$   & $[31.48, 44.68]$ \\
                               & numschool  & $16.00$  & $3.93$   &  $[8.30, 23.70]$ \\
    \midrule
    \multirow{2}{*}{Putuo}     & neighborhood       & $59.53$  & $6.18$   & $[47.41, 71.64]$ \\
                               & numschool  & $12.23$  & $2.85$   & $[6.64, 17.82]$ \\
    \midrule
    \multirow{2}{*}{Fengxian}  & neighborhood      & $9.80$  & $3.07$   & $[3.79, 15.82]$ \\
                               & numschool  & $9.36$ & $2.47$   & $[4.52, 14.21]$ \\
    \midrule
    \multirow{2}{*}{Yantian}   & neighborhood       & $66.70$  & $3.92$   & $[59.03, 74.38]$ \\
                               & numschool  & $41.03$  & $5.42$   & $[30.41, 51.65]$ \\
    \bottomrule
  \end{tabular}\\
\end{table}

\begin{figure}[H]
  \centering
  \begin{subfigure}{0.23\linewidth}
    \centering
    \caption*{\tiny$\bf{(A)}$ Changping\raggedright}
    \includegraphics[width=\linewidth]{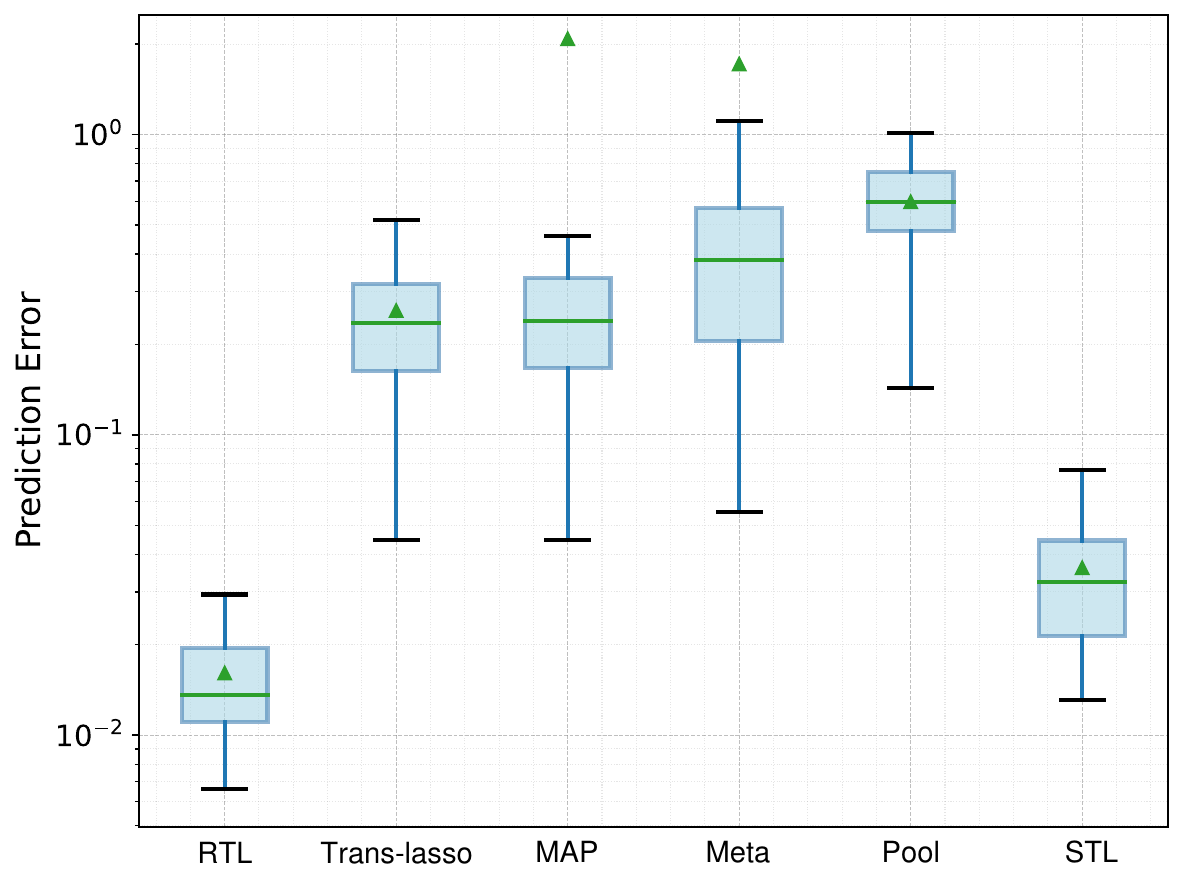}
  \end{subfigure}
  \begin{subfigure}{0.23\linewidth}
    \centering
    \caption*{\tiny\textbf{(B)} Putuo\raggedright}
    \includegraphics[width=\linewidth]{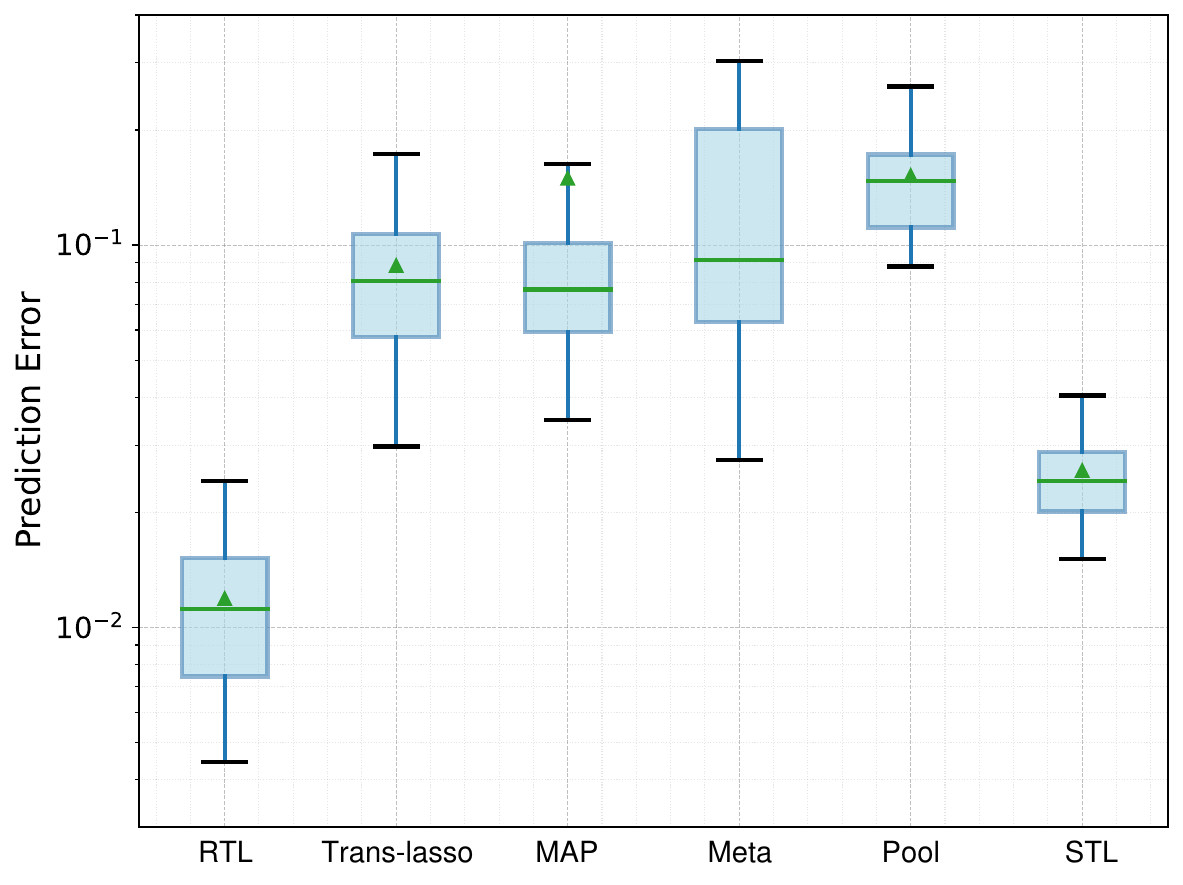}
  \end{subfigure}
  \begin{subfigure}{0.23\linewidth}
    \centering
    \caption*{\tiny\textbf{(C)} Fengxian\raggedright}
    \includegraphics[width=\linewidth]{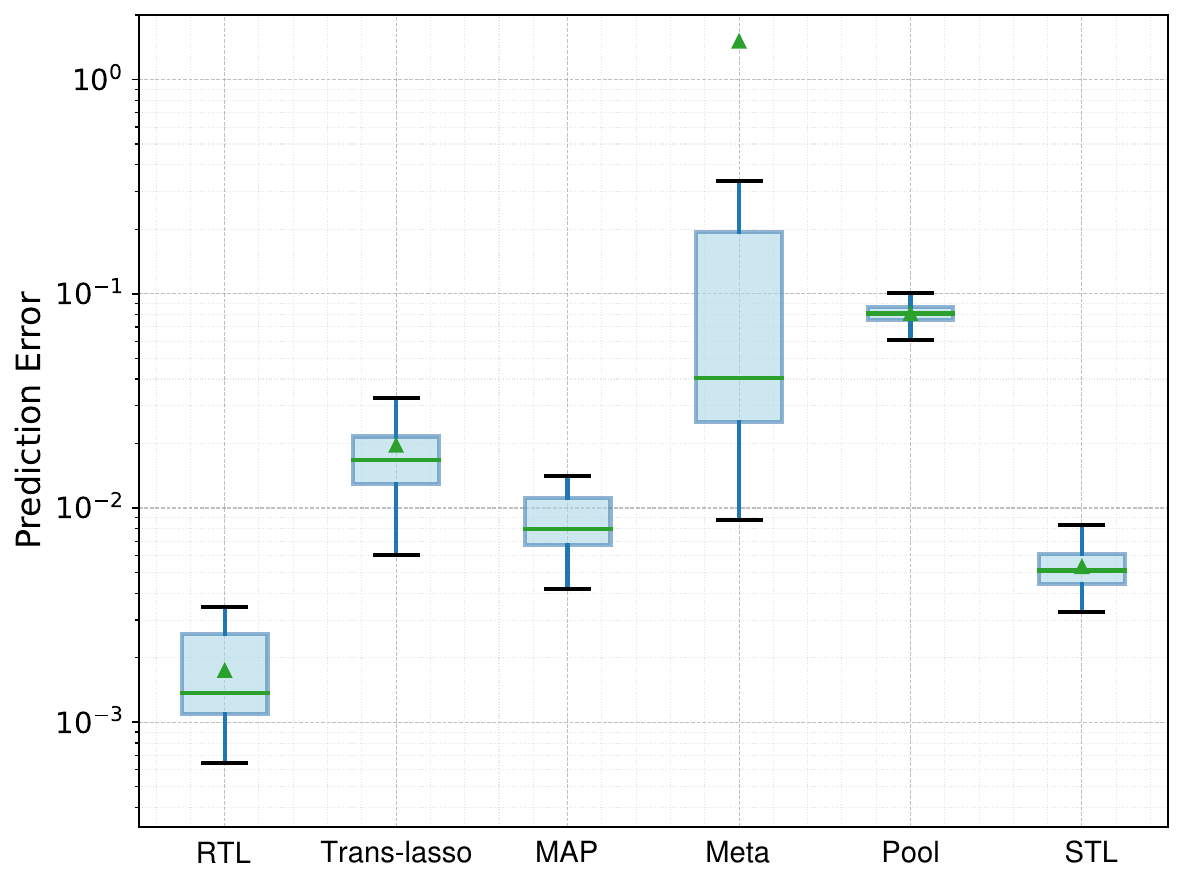}
  \end{subfigure}
  \begin{subfigure}{0.23\linewidth}
    \centering
    \caption*{\tiny\textbf{(D)} Yantian\raggedright}
    \includegraphics[width=\linewidth]{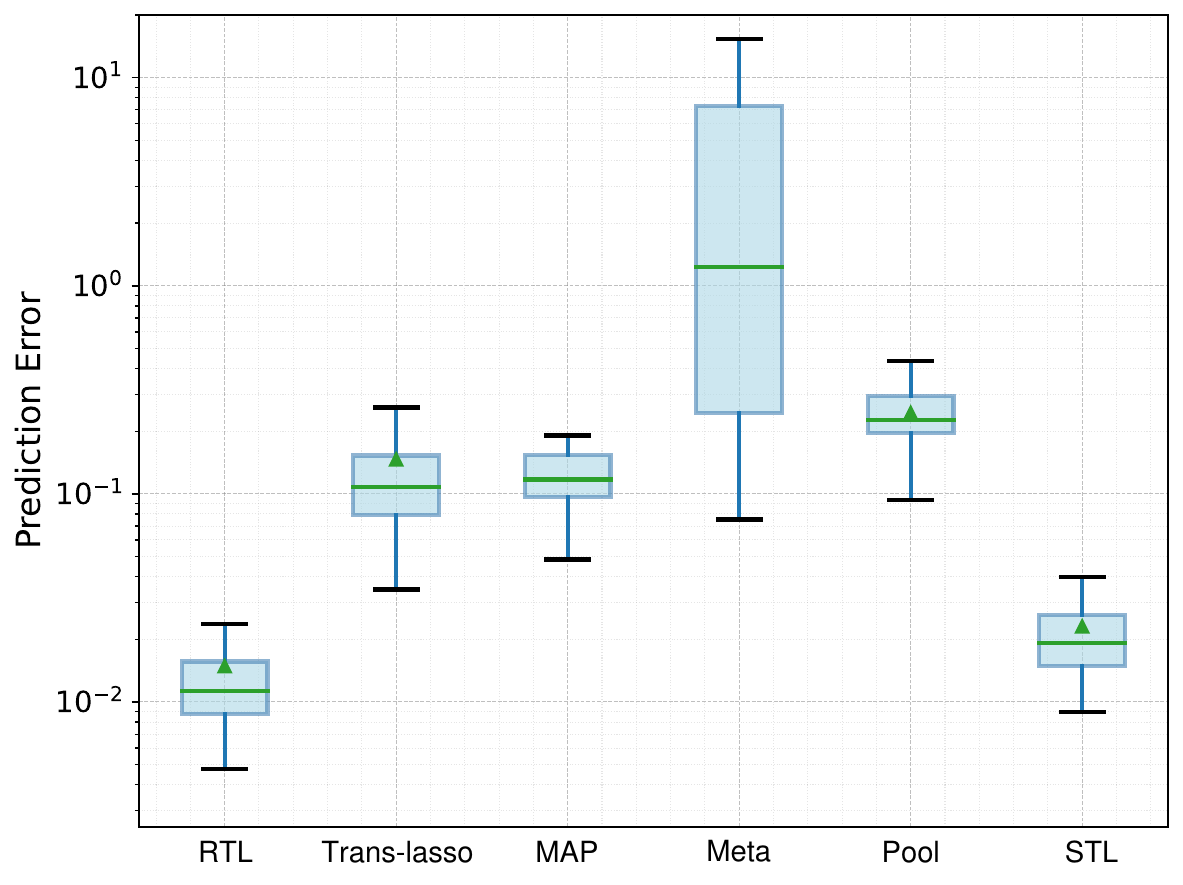}
  \end{subfigure}
  \caption{Prediction performance between the proposed RTL and the existing methods Pool, MAP, Trans-Lasso, Meta, and STL on the housing rental dataset.
  }
  \label{fig:house_pred}
\end{figure}

We further assess the prediction performance by randomly splitting the target dataset into training, validation, and testing sets. Specifically, the training and testing sets each consists of 30\% of the data, with the remaining 40\% designated as the validation set. We evaluate the prediction performance using the mean squared error (MSE) on the testing set.
Figure~\ref{fig:house_pred} illustrates the prediction errors of our proposed RTL method compared to existing methods, Pool, MAP, Trans-Lasso, Meta, and STL. The results demonstrate that RTL has lower prediction error than these existing methods, indicating its superior performance in predicting rental prices.

\section{Discussion} \label{sec:discuss}
In this work, we introduce a new approach to transfer learning within the context of semiparametric regression inference.The essence of our strategy lies in the transfer of knowledge from the source domains to the target domain via a representation function. Our goal is to enhance both prediction accuracy and estimation precision in the target domain by leveraging data from multiple source domains. The key idea of our method is the learning of a shared representation across various source tasks, which is then applied to a target task. We address data heterogeneity between the source and target domains by incorporating domain-specific parameters in their respective models. This strategy facilitates the integration of varied data representations while maintaining model interpretability and adaptability to heterogeneous datasets.

Our proposed RTL method has the potential to be adapted for use with other models, including semiparametric generalized linear and classification models. However, there are several challenging issues that warrant further investigation within our proposed framework. Firstly, a pivotal hyperparameter in our approach is the number of representations, for which the optimal selection remains an open question. The determination of this parameter significantly affects the model's performance and its ability to generalize. Our simulation studies indicate that the method performs adequately as long as the number of representations falls within a reasonable range. This observation suggests that it is helpful to consider a cross-validation type method for selecting this hyperparameter, which could provide a systematic approach to further enhance model performance. Secondly, our findings are based on a moderately high-dimensional regime of the model. Although this scenario is relevant in many applications, extending our method to handle sparse, high-dimensional settings is challenging. In such scenarios, where the model's dimensionality may exceed the sample size, it becomes crucial to integrate regularization techniques into the model fitting objective function via a penalty term to ensure effective model performance. Moreover, while our current model uses a linear mapping to integrate representations, there is potential to explore more flexible approaches. For instance, transitioning from task-specific linear functions to nonlinear functions could allow for the capture of more complex non-linear relationships between the representations and the target responses. Pursuing advancements in this area could significantly improve the model's capacity and its adaptability to more complex data. We hope to address these issues in our future work.

\bibliographystyle{apalike}
{\singlespacing
\bibliography{RTL-bib.bib}
}

\newpage
\appendix
\setcounter{section}{0}
\renewcommand{\thesection}{\Alph{section}}
\setcounter{equation}{0}
\renewcommand{\theequation}{A.\arabic{equation}}
\setcounter{page}{1}
\renewcommand{\thepage}{\arabic{page}}

\def\thetable{S\arabic{table}}
\def\thefigure{S\arabic{figure}}
{\centering
\textbf{\LARGE Appendix}

\section{The network architecture for the semi-synthetic MNIST data}

The architecture of the network used in the semi-synthetic data analysis is detailed in Figure~\ref{fig:struc_mnist}.
\begin{figure}[H]
  \centering
  \includegraphics[width=\linewidth]{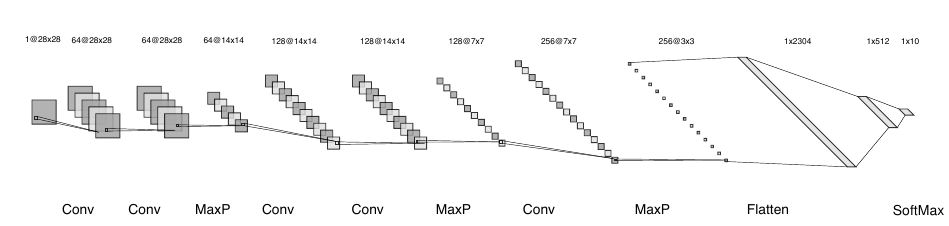}
  \caption{Structure of representation network used for synthetic data analysis. The convolution (Conv) transforms, max pooling (MaxP) transforms, tensor flattening (Flatten) and softmax transform are labeled in the bottom. The dimension of the output in each layer is labeled in the top.}
  \label{fig:struc_mnist}
\end{figure}

\section{Proofs of theoretical results}
Before we present the proofs of the results stated in the paper, we first introduce the definition of Rademacher and Gaussian complexity for $\calR$.
\begin{definition}
  We define the empirical  and population Rademacher complexities for a class of functions $\calR$ containing function $\bR: \mbR^q \rightarrow \mbR^p$ over $n$ data points, $(\bZ_1,\ldots,\bZ_n)$ as,
  \begin{equation*}
    \wcalF_{n}(\calR)=\rmE_{\varepsilon}\left[\sup\limits_{{\bR}\in \calR}\frac{1}{n}\sum_{j=1}^p\sum_{i=1}^{n}\varepsilon_{ij}{R_j}(\bZ_i)\right],
  \end{equation*}
  and
  \begin{equation*}
    \calF_{n}(\calR)=\rmE_{\bz}[\wcalF_{n}(\calR)],
  \end{equation*}
  respectively, where ${\rmE}_{\varepsilon}(\cdot)$ refers to the expectation operator taken over the randomness $\varepsilon_{ij}$'s, $\varepsilon_{ij}$'s are independent Rademacher random variables and $R_j(\cdot)$ is the $j$th element of $\bR(\cdot)$. Analogously, the empirical  and population Gaussian complexities are defined as
  \begin{equation*}
    \wcalG_{n}(\calR)=\rmE_{\iota}\left[\sup\limits_{{\bR}\in \calR}\frac{1}{n}\sum_{j=1}^p\sum_{i=1}^{n}\iota_{ij}{R_j}(\bZ_i)\right],
  \end{equation*}
  and
  \begin{equation*}
    \calG_{n}(\calR)=\rmE_{\bz}[\wcalG_{n}(\calR)],
  \end{equation*}
respectively, where $\iota_{ij}$'s are independent standard Gaussian random variables.
\end{definition}

\subsection{Auxiliary Lemmas}
Define $\ell(\calD_k;\bbeta_k,\bgamma_k,\bR)=\{Y_k-\bX_k\tr \bbeta_k-\bR\tr(\bZ_k) \bgamma_k\}_2^2$, where $\calD_k=(\bX_k,\bZ_k,Y_k)$.
Let $\Psi_{\ell,\delta}=\bigl\{\sum_{k=1}^K\ell(\calD_{k},\bpsi)/K-\sum_{k=1}^K\ell(\calD_{k},\bpsi^{\ast})/K,\bpsi\in \Psi_{\delta}\bigr\},$ where $\Psi_{\delta}=\{\bpsi: {\delta}/{2}\leq d_2(\bpsi,\bpsi^{\ast})\leq \delta,\delta>0, \bpsi\in\Psi\}$, $\bpsi=\{(\bbeta_{k},\bgamma_k)_{k=1}^K,\bR\}$, and  $\bpsi^{\ast}=\{(\bbeta_{\ast k},\bgamma_k^{\ast})_{k=1}^K,\bR^\ast\}$. Denote $\mathbb{P}_N$ and $\mathbb{P}$ as the empirical and probability measure of $\{\calD_{n,k}\}_{k=1}^K$ and $\{\calD_{k}\}_{k=1}^K$, where $\calD_{n,k}=\bigl\{\calD_{ki}=(\bxik,\bzik,Y_{ki}),i=1,\ldots, n_k\bigr\}$. We further define $\mathbb{G}_N=\sqrt{N}(\mathbb{P}_N-\mathbb{P})$. Denote the population $2$-norm  for function class $\calB^{\otimes K} +\Gamma^{\otimes K}(\calR)$ as,
\begin{align*}
d_2(\bpsi,\bpsi^{\prime})=\left(\frac{1}{K}\sum_{k=1}^K \rmE\left[\left(\bX_k^{\top}\bbeta_k+\bR^{\top}(\bZ_k)\bgamma_k-\bX_k^{\top}\bbeta_k^\prime-\bR^{\prime\top}(\bZ_k)\bgamma_{k}^{\prime}\right)^2\right]\right)^{1/2},
\end{align*}
where $\bpsi^{\prime}=\{(\bbeta_k^{\prime},\bgamma_k^{\prime})_{k=1}^K, \bR^{\prime}\}$. We use $\lesssim$, $\gtrsim$ and $\asymp$ to denote less than, greater than, and equal to up to a universal constant. For simplicity in notation, we remove the parts with the same parameters from the distance ${d}_2$ in the following contents.
\begin{lemma}\label{lemog}
Suppose Condition \ref{conb} holds, we have
\begin{align*}
  \rmE^{\ast}\|\mathbb{G}_N\|_{\Psi_{\ell,\delta}}\lesssim \delta \sqrt{s_1\log(s_2/\delta)}+\frac{s_1}{\sqrt{N}}\log(s_2/\delta),
\end{align*}
where $\rmE^{\ast}$ is the outer measure and $s_2={12B_{\gamma}(D+1)(B_{R}+1)(2B_{\theta})^{D+2}(\prod_{j=0}^D p_j)}{\left(\prod _{j=1}^D p_j!\right)^{-1/S}}$.
\end{lemma}
\begin{proof}[Proof of Lemma \ref{lemog}]
Using the triangle inequality, we can decompose the distance on function class $\calB^{\otimes K}+\Gamma^{\otimes K}(\calR)$ into a distance over $\calB^{\otimes K}$, $\Gamma^{\otimes K}$, and $\calR$. 
We have
\begin{align}\label{decomp}
  &{d}_{2}(\{(\bbeta_k^\prime,\bgamma_k^\prime)_{k=1}^K,\bR^\prime\},\{(\bbeta_k,\bgamma_k)_{k=1}^K,\bm{R}\})\nonumber\\
  \leq&d_{2}(\{(\bbeta_k^\prime,\bgamma_k^\prime)_{k=1}^K,\bR^\prime\},\{(\bbeta_k,\bgamma_k^\prime)_{k=1}^K,\bR^\prime\})\nonumber\\& +{d}_{2}(\{(\bbeta_k,\bgamma_k^\prime)_{k=1}^K,\bR^\prime\},\{(\bbeta_k,\bgamma_k)_{k=1}^K,\bR^\prime\})\nonumber\\&+{d}_{2}(\{(\bbeta_k,\bgamma_k)_{k=1}^K,\bR^\prime\},\{(\bbeta_k,\bgamma_k)_{k=1}^K,\bm{R}\})\nonumber\end{align}
    \begin{align}
  \leq&{d}_{2}(\{\bbeta_k^\prime\}_{k=1}^K,\{\bbeta_k\}_{k=1}^K)+{d}_{2}(\{\bgamma_k^\prime\}_{k=1}^K,\{\bgamma_k\}_{k=1}^K)\nonumber\\&+\max_{k\in[K]}\|\bgamma_k\|_2 {d}_{2}(\bR^\prime,\bR).
\end{align}
We then use a covering argument on each of the spaces  $\calB^{\otimes K}$, $\Gamma^{\otimes K}$, and $\calR$ to witness a covering of the composed space $\calB^{\otimes K}+\Gamma^{\otimes K}(\calR)$.
First, let $\calC_{\calB^{\otimes K}}$ be a $\tau_0$-covering for the function class $\calB^{\otimes K}$ of the norm ${d}_2$. Then for each $\bbeta \in \calC_{\calB^{\otimes K}}$, construct a $\tau_1$-covering, $\calC_{\calR}$ for the function class $\calR$ of the norm ${d}_2$. Last, for each $\bbeta \in \calC_{\calB^{\otimes K}}$ and $\bR \in \calC_{\calR}$, construct a $\tau_2$-covering $\calC_{\Gamma^{\otimes K}}$ for the function class $\Gamma^{\otimes K}$ of the norm ${d}_2$.  Using the decomposition of distance (\ref{decomp}), we can claim that set $$\calC_{\calB^{\otimes K}}\cdot \calC_{\Gamma^{\otimes K}(\calR)}=\cup_{\bbeta\in\calC_{\calB^{\otimes K}}}(\cup_{\bR\in \calC_{\calR}}(\calC_{\Gamma^{\otimes K}}),)$$
is a $(\tau_0+\max_{k\in [K]}\|\bgamma_k\|_2 \tau_1+\tau_2)$-covering for the function space $\calB^{\otimes K}+\Gamma^{\otimes K}(\calR)$ in the norm ${d}_2$.
 To see this, let $\{\bbeta_k\}_{k=1}^K\in \calB^{\otimes K},\{\bgamma_k\}_{k=1}^K\in \Gamma^{\otimes K}$, and $\bR\in\calR$ be arbitrary. Now let $\{\bbeta_k^\prime\}_{k=1}^K\in \calC_{\calB^{\otimes K}}$ be $\tau_0$ close to $\{\bbeta_k\}_{k=1}^K$; given this $\{\bbeta_k^\prime\}_{k=1}^K$, there exists  $\bR^\prime\in \calC_{\calR}$ be $\tau_1$ close to $\bR$; given this $\{\bbeta_k^\prime\}_{k=1}^K$ and $\bR^\prime$, there exists $\{\bgamma_k^\prime\}_{k=1}^K\in \calC_{\Gamma^{\otimes K}}$ such that $\{\bgamma_k^\prime\}_{k=1}^K$ be $\tau_2$ close to $\{\bgamma_k\}_{k=1}^K$. By the process of constructing $\{(\bbeta_k^\prime,\bgamma_k^\prime)_{k=1}^K,\bR^\prime\}$ and (\ref{decomp}), we have that,
\begin{align*}
  {d}_2(\{(\bbeta_k^\prime,\bgamma_k^\prime)_{k=1}^K,\bR^\prime\},\{(\bbeta_k,\bgamma_k)_{k=1}^K,\bm{R}\}) \leq \tau_0+\max_{k\in[K]}\|\bgamma_k\|_2 \tau_1+\tau_2.
\end{align*}
We now bound the cardinality of the coverings $\calC_{\calB^{\otimes K}}\cdot \calC_{\Gamma^{\otimes K}(\calR)}$.
\begin{align*}
  \left|\calC_{\calB^{\otimes K}}\cdot\calC_{\Gamma^{\otimes K}(\calR)}\right|
  \leq|\calC_{\calB^{\otimes K}}||\calC_{\calR}|\max_{\bR\in\calR}\left|\calC_{\Gamma^{\otimes K}_{\bR}}\right|.
\end{align*}
To control the cardinality of $\max_{\bR\in\calR}\left|\calC_{\Gamma^{\otimes K}_{\bR}}\right|$, note an $\epsilon$-covering of $\calC_{\Gamma^{\otimes K}_{\bR}}$ can be obtained from the cover $\calC_{\Gamma_{\bR}}\times\cdots\times \calC_{\Gamma_{\bR}}$. Hence,
\begin{align*}
  \left|\calC_{\calB^{\otimes K}}\right|
  \leq\left|\calC_{\calB}\right|^{K},\quad
  \max_{\bR\in\calR}\left|\calC_{\Gamma^{\otimes K}_{\bR}}\right|
\leq\left|\max_{\bR\in\calR}\calC_{\Gamma_{\bR}}\right|^{K}.
\end{align*}
Note that for any $\bpsi, \bpsi^{\prime} \in \Psi_{\delta}$, we have
\begin{align*}
  &\rmE\left[\frac{1}{K}\sum_{k=1}^K\ell(\calD_{k},\bpsi^{\prime})-\frac{1}{K}\sum_{k=1}^K\ell(\calD_{k},\bpsi)\right]\\
  =&\rmE\left[\frac{1}{K}\sum_{k=1}^K\ell(\calD_{k},\bpsi^{\prime})-\frac{1}{K}\sum_{k=1}^K\ell(\calD_{k},\bpsi^{\ast})\right]\\
  &+\rmE\left[\frac{1}{K}\sum_{k=1}^K\ell(\calD_{k},\bpsi)-\frac{1}{K}\sum_{k=1}^K\ell(\calD_{k},\bpsi^{\ast})\right]\\=&d_2^2(\bpsi,\bpsi^{\ast})+d_2^2(\bpsi^{\prime},\bpsi^{\ast})\lesssim\delta^2.
\end{align*}
By Theorem 3 in \citet{shen2023complexity}, we have
\begin{align*}
  \calN(\tau,{d}_{2},\calR)\leq &\calN(\tau,{d}_{\infty},\calR)\\\leq &\frac{\left(4(D+1)(B_{R}+1)(2B_{\theta})^{D+2}(\prod_{j=0}^D p_j)\tau^{-1}\right)^{S}}{\prod _{j=1}^D p_j!},
\end{align*}
where ${d}_{\infty}$ is the infinity norm. 
Furthermore, by the construction of covering net and Theorem 2.7.11 in \citet{van1996weak}, we have
\begin{align*}
  &\log \calN_{[]}(\tau,d_2,\Psi_{\ell,\delta})\\
  \lesssim &Kd\log\left(\frac{3\delta}{\tau}\right)+Kp\log\left(\frac{3\delta}{\tau}\right)\\
&+S\log\left(\frac{12B_{\gamma}(D+1)(B_{R}+1)(2B_{\theta})^{D+2}(\prod_{j=0}^D p_j)}{\tau\left(\prod _{j=1}^D p_j!\right)^{1/S}}\right).
\end{align*}
Using sub-additivity of the $\sqrt{}$ function, if $\delta\leq s_2/3$, then we have
\begin{align}\label{boudJ}
\calJ_{[ ]}(\delta,\Psi_{\ell,\delta}):=&\int_{0}^{\delta}\sqrt{1+\log \calN_{[ ]}(\tau,d_2,\Psi_{\ell,\delta})}\rmd\tau\nonumber\\
\lesssim&\int_{0}^{\delta}\sqrt{1+s_1\log(s_2/\tau)}\rmd\tau\nonumber\\
=&s_2\sqrt{\frac{s_1}{2}}\int_{\sqrt{2\log(s_2/\delta)}}^{\infty} v^2e^{-v^2/2}\rmd v\nonumber\\
\asymp&\delta \sqrt{s_1\log(s_2/\delta)}.
\end{align}
Under Condition \ref{conb}, combining (\ref{boudJ}) with lemma 3.4.2 in \citet{van1996weak} gives
\begin{align*}
  \rmE^{\ast}\|\mathbb{G}_N\|_{\Psi_{\ell,\delta}}\lesssim\calJ_{[ ]}(\delta,\Psi_{\ell,\delta})\left\{1+\frac{\calJ_{[ ]}(\delta,\Psi_{\ell,\delta})}{\delta^2\sqrt{N}}\right\}\lesssim \delta \sqrt{s_1\log(s_2/\delta)}+\frac{s_1}{\sqrt{N}}\log(s_2/\delta).
\end{align*}
\end{proof}

\begin{lemma}\label{bR}
Suppose Condition \ref{conb} holds, if $n_k\gtrsim d +\log K$, we have
\begin{align*}
  &\calF_{N}(\calB^{\otimes K} +\Gamma^{\otimes K}(\calR))\\
  \leq& 
  O\{\underline{n}^{-1/2}{d^{1/2}}\}+O\{\underline{n}^{-1/2}N^{-2}(p \log N)^{1/2}\}+O\left\{\underline{n}^{-1/2}p^{1/2} (\log N)\right\}
\end{align*}
  \begin{align*}
 & +O\{{N}^{-1/2}{(D+2+\log q)^{1/2}}p(\log N)^2\prod_{i=0}^D (p_i+1)\}\\&+O\left\{\underline{n}^{-1/2}N^{-2}\left(S\log\left(N (D+1)(2B_{\theta})^{D+2}\left(\prod_{j=0}^D p_j\right)\left(\prod _{j=1}^D p_j!\right)^{-1/S}\right)\right)^{1/2}\right\}.
\end{align*}
\end{lemma}

\begin{proof}[Proof of Lemma \ref{bR}]
 Under Condition \ref{conb}, by the definition of the empirical Gaussian complexity, we have
  \begin{align}\label{bgb}
    \wcalG_{n_k}(\calB)&=\rmE_{\iota}\left[\sup\limits_{\bbeta_k\in \calB}\frac{1}{n_k}\sum_{i=1}^{n_k}\iota_{ik}\bxik\bbeta_k\right]\nonumber
    \\
    &\leq \max\limits_{k\in[K]}\frac{B_{\beta}}{n_k}\sqrt{\rmE_{\iota}\left[\sum_{i=1}^{n_k}\|\iota_{ik}\bxik\|_2^2\right]}\nonumber\\&\leq\max\limits_{k\in[K]} \frac{B_{\beta}}{n_k}\sqrt{\sum_{i=1}^{n_k}\|\bxik\|_2^2}\nonumber\\
    &=\max\limits_{k\in[K]}\frac{B_{\beta}}{\sqrt{n_k}}\sqrt{{\rm tr}(\Sigma_{\bX_{k}})}\nonumber\\
    &=\max\limits_{k\in[K]}\frac{B_{\beta}}{\sqrt{n_k}}\sqrt{\sum_{j=1}^d\sigma_j(\Sigma_{\bX_{k}})},
  \end{align}
  where $\Sigma_{\bX_k}=\sum_{i=1}^{n_k}\bxik\bxik^\top/n_k$ and $\sigma_j(\Sigma_{\bX_k})$ is the $j$th largest eigenvalue of $\Sigma_{\bX_k}$.  Similarly, we can prove
  
\begin{align} \label{bgr}
  \wcalG_{n_k}(\Gamma) \leq \max\limits_{k\in [K]}\frac{B_\gamma}{\sqrt{n_k}}\sqrt{\sum_{j=1}^p\sigma_j(\Sigma_{\bR(\bZ_{k})})},
\end{align}
where $\sigma_j(\Sigma_{\bR(\bZ_{k})})$ is the $j$th largest eigenvalue of $\sum_{i=1}^{n_k}\bR(\bzik)\bR^{\top}(\bzik)/n_k$.
Furthermore, we can obtain that,
  \begin{align}\label{bgh}
    \wcalG_{N}(\calR)&=\frac{1}{N}\rmE_{\iota}\left[\sup\limits_{\bR\in \calR}\sum_{j=1}^p\sum_{k=1}^{K}\sum_{i=1}^{n_k}\iota_{ik}R_j(\bzik)\right]
    \leq \sum_{j=1}^p \wcalG_{N}(R_j)\leq (\log N)\sum_{j=1}^p \wcalF_{N}(R_j).
  \end{align}
By the definition of empirical Gaussian complexity, we can easily conclude
\begin{align}\label{decopG1}
  \wcalG_{N}(\calB^{\otimes K} +\Gamma^{\otimes K}(\calR))
  \leq \wcalG_{N}(\calB^{\otimes K}) +\wcalG_{N}(\Gamma^{\otimes K}(\calR)).
\end{align}
For simplicity in notation, denote $f_{ik}(\bgamma_k,\bR;\bgamma_k^\prime,\bR^\prime)=\left(\bR\tr (\bZ_{ki})\bgamma_k-\bR^{\prime\top} (\bZ_{ki}) \bgamma_k^\prime\right)^2$. Theorem 7 of \citet{tripuraneni2020theory} implies that
\begin{align}\label{decopG2}
  \wcalG_{N}(\Gamma^{\otimes K}(\calR))
  \leq & 4 \sup_{\{(\bgamma_k)_{k=1}^K,\bR\},\{(\bgamma_k^\prime)_{k=1}^K,\bR^\prime\}}\frac{1}{K}\sum_{k=1}^{K}\frac{1}{n_k}\sum_{i=1}^{n_k}f_{ik}(\bgamma_k,\bR;\bgamma_k^\prime,\bR^\prime)/N^2\nonumber\\& +128(\log N)\left[\left(\max_{k}\|\bgamma_k\|\right)\wcalG_{N}(\calR)+\max_{k}\wcalG_{n_k}(\Gamma)\right].
\end{align}
Similar to the calculation of the covering number in Lemma \ref{lemog}, we can obtain that
\begin{align*}
&\calN(\tau/3,d_{\infty},f_k(\Gamma(\calR)-\Gamma(\calR)))\nonumber\\
\leq &\left(\frac{48B_{R}^2B_{\gamma}}{\tau}\right)^{2p}\left( \frac{\left(192 B_{\gamma}^2 B_{R}(D+1)(B_{R}+1)(2B_{\theta})^{D+2}(\prod_{j=0}^D p_j)\right)^{2S}}{\tau^{2S} (\prod _{j=1}^D p_j!)^2} \right).
\end{align*}
Denote $f_k(\bgamma_k,\bR;\bgamma_k^\prime,\bR^\prime)=\left(\bR\tr (\bZ_{k})\bgamma_k-\bR^{\prime\top} (\bZ_{k}) \bgamma_k^\prime\right)^2$. By Lemma 9.1 of \cite{gyorfi2002distribution}, we have that, for any $\tau>0$
\begin{align*}
& P\left(\sup_{\{\bgamma_k,\bR\},\{\bgamma_k^\prime,\bR^\prime\}}\left|\frac{1}{n_k}\sum_{i=1}^{n_k}f_{ik}(\bgamma_k,\bR;\bgamma_k^\prime,\bR^\prime)- \mathbb{E}\left[f_k(\bgamma_k,\bR;\bgamma_k^\prime,\bR^\prime)\right]\right|> \tau\right)\nonumber\\
\leq &2 \calN(\tau/3,d_{\infty},f_k(\Gamma(\calR)-\Gamma(\calR))) \exp\left\{-\frac{n_k\tau^2}{18B_{\gamma}^2 B_{R}^2}\right\}.
\end{align*}
Denote $c=1/(18B_{\gamma}^2 B_{R}^2)$ and $C=2 \calN(1/(3n_k),d_{\infty},f_k(\Gamma(\calR)+\Gamma(\calR)))$. Note that ${\log C}/{(c n_k)}\geq 1/n_k$.  Then,
\begin{align}
&\mathbb{E}\left[\sup_{\{\bgamma_k,\bR\},\{\bgamma_k^\prime,\bR^\prime\}}\left\{\frac{1}{n_k}\sum_{i=1}^{n_k}f_{ik}(\bgamma_k,\bR;\bgamma_k^\prime,\bR^\prime)- \mathbb{E}\left[f_k(\bgamma_k,\bR;\bgamma_k^\prime,\bR^\prime)\right]\right\}\right]\nonumber\\
\leq &\sqrt{\mathbb{E}\left[\left(\sup_{\{\bgamma_k,\bR\},\{\bgamma_k^\prime,\bR^\prime\}}\left\{\frac{1}{n_k}\sum_{i=1}^{n_k}f_{ik}(\bgamma_k,\bR;\bgamma_k^\prime,\bR^\prime)- \mathbb{E}\left[f_k(\bgamma_k,\bR;\bgamma_k^\prime,\bR^\prime)\right]\right\}\right)^2\right]}\nonumber\\
\leq &\sqrt{\int_{0}^{\infty} P\left\{\left(\sup_{\{\bgamma_k,\bR\},\{\bgamma_k^\prime,\bR^\prime\}}\left\{\frac{1}{n_k}\sum_{i=1}^{n_k}f_{ik}(\bgamma_k,\bR;\bgamma_k^\prime,\bR^\prime)- \mathbb{E}\left[f_k(\bgamma_k,\bR;\bgamma_k^\prime,\bR^\prime)\right]\right\}\right)^2> \tau\right\} \rmd \tau}\nonumber\\
\leq & \sqrt{\frac{\log C}{c n_k}+ \int_{\frac{\log C}{c n_k}}^{\infty} 2 \calN(\tau/3,d_{\infty},f_k(\Gamma(\calR)-\Gamma(\calR))) \exp\left\{-\frac{n_k\tau}{18B_{\gamma}^2 B_{R}^2}\right\} \rmd \tau}\nonumber\\
\leq & \sqrt{\frac{\log C}{c n_k}+ \int_{\frac{\log C}{c n_k}}^{\infty} 2 \calN(1/(3n_k),d_{\infty},f_k(\Gamma(\calR)-\Gamma(\calR))) \exp\left\{-\frac{n_k\tau}{18B_{\gamma}^2 B_{R}^2}\right\} \rmd \tau}\nonumber\\
=& \sqrt{\frac{18B_{\gamma}^2 B_{R}^2\left(1+\log 2+\log\{ \calN(1/(3n_k),d_{\infty},f_k(\Gamma(\calR)-\Gamma(\calR)))\}\right)}{n_k}}\nonumber.
\end{align}
Under Condition \ref{conb}, we have
\begin{align}\label{expd}
&\mathbb{E}\left[\sup_{\{(\bgamma_k)_{k=1}^K,\bR\},\{(\bgamma_k^\prime)_{k=1}^K,\bR^\prime\}}\frac{1}{K}\sum_{k=1}^{K}\frac{1}{n_k}\sum_{i=1}^{n_k}f_{ik}(\bgamma_k,\bR;\bgamma_k^\prime,\bR^\prime)\right]\nonumber\\
\leq &\mathbb{E}\left[\sup_{\{\bgamma_k,\bR\},\{\bgamma_k^\prime,\bR^\prime\}}\frac{1}{n_k}\sum_{i=1}^{n_k}f_{ik}(\bgamma_k,\bR;\bgamma_k^\prime,\bR^\prime)\right]\nonumber\\
\leq & 4B_{\gamma}^2B_{R}^2+\sqrt{\frac{18B_{\gamma}^2 B_{R}^2(1+\log 2)}{n_k}}+\sqrt{\frac{36p B_{\gamma}^2B_{R}^2}{{n_k}}\log\left({48 n_k B_{R}^2B_{\gamma}}\right)}\nonumber\\&\quad \qquad  + \sqrt{\frac{36SB_{\gamma}^2B_{R}^2}{n_k}\log \left( \frac{192 n_kB_{R}B_{\gamma}^2 (D+1)(B_{R}+1)(2B_{\theta})^{D+2}(\prod_{j=0}^D p_j)}{\left(\prod _{j=1}^D p_j!\right)^{1/S}} \right)}.
\end{align}
Noting by \citet{ledoux2013probability} (p97), the empirically Rademacher complexity is upper bounded by empirical Gaussian complexity up to a factor, together with (\ref{decopG1}) and (\ref{decopG2}), we have
\begin{align}\label{boundhR}
  &\wcalF_{N}(\calB^{\otimes K} +\Gamma^{\otimes K}(\calR))\nonumber\\
  \leq &\sqrt{\frac{\pi}{2}}\left(\wcalG_{N}(\calB^{\otimes K} +\Gamma^{\otimes K}(\calR))\right)\nonumber\\
  \leq &\sqrt{\frac{\pi}{2}}\left(\wcalG_{N}(\calB^{\otimes K}) +\wcalG_{N}(\Gamma^{\otimes K}(\calR))\right)\nonumber\\
  \leq & \sqrt{\frac{\pi}{2}}\max_{k}\wcalG_{n_k}(\calB)+2\sqrt{{2\pi}} \sup_{\{(\bgamma_k)_{k=1}^K,\bR\},\{(\bgamma_k^\prime)_{k=1}^K,\bR^\prime\}}\frac{1}{K}\sum_{k=1}^{K}\frac{1}{n_k}\sum_{i=1}^{n_k}f_{ik}(\bgamma_k,\bR;\bgamma_k^\prime,\bR^\prime)/N^2\nonumber\\&+ 64\sqrt{{2\pi}}(\log N)\left[\left(\max_{k}\|\bgamma_k\|\right)\wcalG_{N}(\calR)+\max_{k}\wcalG_{n_k}(\Gamma)\right].
\end{align}
Under Condition \ref{conb} and $\btheta\in\Theta$, adapting from Theorem 2 of \citet{golowich2018size},
  \begin{align*}
    \calF_{N}(R_j)\leq& 2\prod\limits_{i=0}^D (p_i+1)B_{R}\sqrt{D+2+\log q}/\sqrt{N}.
  \end{align*}
If $n_k\gtrsim d +\log K$, applying Lemma 4 of \citet{tripuraneni2020theory} and using the concavity of $\sqrt{\cdot}$ function,  we have
  $\rmE\left[\sqrt{\sum_{j=1}^d\sigma_j(\Sigma_{\bX_{k}})}\right]\leq O(\sqrt{d})$ \mbox{and} $\rmE\left[\sqrt{\sum_{j=1}^p\sigma_j(\Sigma_{\bR(\bZ_{k})})}\right]\leq O(\sqrt{p})$. Thus, under Condition \ref{conb}, combing (\ref{boundhR}) with (\ref{bgb}), (\ref{bgr}), (\ref{bgh}), and (\ref{expd}), we prove Lemma \ref{bR}.
\end{proof}
\begin{lemma}\label{diffQmuJ}
 Suppose Conditions \ref{conb}-\ref{con5} hold, we have
  \begin{align}\label{diffJJ1} 
    \|\wbmu-\bmu^\ast\|_2
   \leq O_p(n_0^{-1/2})+ O_p(\Delta_N);\nonumber\\
  \|\wbJz^{-1}-\bJ_0^{-1}\|_2
    =O_p(n_0^{-1/2})+O_p(\Delta_N).
  \end{align}
\end{lemma}
\begin{proof}[Proof of Lemma \ref{diffQmuJ}]
{Theorem \ref{thediffh} indicates that $\|\wbR(\bZ)-\bR^{\ast}(\bZ)\|_2=O_p(\Delta_N)$ for all $\bZ\in\calZ$}. 
Denote $\wbQ=\sum_{i=1}^{n_0}\wbRiz(\wbRiz)\tr /n_0$ and $\bQ=\sum_{i=1}^{n_0}\bRaiz(\bRaiz)\tr /n_0$. Under Condition~\ref{conb}, we first derive
\begin{align}\label{diffhh}
&\left\|\wbQ-\bQ\right\|_2\nonumber\\
=&\left\|\frac{1}{n_0}\sum_{i=1}^{n_0}\left(\wbRiz-\bRaiz\right)\left(\wbRiz-\bRaiz\right)\tr \right.\nonumber\\
&+\frac{1}{n_0}\sum_{i=1}^{n_0}\bRaiz\left(\wbRiz-\bRaiz\right)\tr \nonumber\\&+
\frac{1}{n_0}\sum_{i=1}^{n_0}\left.\left(\wbRiz-\bRaiz\right)\left(\bRaiz\right)\tr \right\|_2\nonumber\\
=&O_p\{\|\wbRz-\bRaz\|_2\}=O_p(\Delta_N).
\end{align}
By \eqref{diffhh} and Condition \ref{conb}, we further derive that
\begin{align*}
&\left\|\wbQ^{-1}-\bQ^{-1}\right\|_2\\
=&\left\|\wbQ^{-1}(\wbQ-\bQ)\bQ^{-1}\right\|_2\\
\leq&\left\|\wbQ^{-1}\right\|_2\left\|\wbQ-
\bQ\right\|_2\left\|\bQ^{-1}\right\|_2\\
=&O_p(1)O_p(\Delta_N)O_p(1)=O_p(\Delta_N).
\end{align*}
Recall the construction of efficient scores for $\bbeta_{\ast 0}$, $\bmu^\ast$ is the minimizer of
\begin{equation*}
  \rmE[\|\bX_0-\bmu\bR^\ast(\bZ_0)\|_2^2].
\end{equation*}
Hence, under Conditions \ref{conb}-\ref{con5} and by the weak law of large numbers, we have
\begin{align}\label{diffwbmu}
  &\|\wbmu-\bmu^\ast\|_2\nonumber\\
 \leq&\left\|\frac{1}{n_0}\sum_{i=1}^{n_0}\bxiz(\wbRiz)\tr \wbQ^{-1}- \frac{1}{n_0}\sum_{i=1}^{n_0}\bxiz(\bRaiz)\tr \bQ^{-1}\right\|_2\nonumber\\
 &+\left\|\frac{1}{n_0}\sum_{i=1}^{n_0}\bxiz(\bRaiz)\tr \bQ^{-1}-\bmu^\ast\right\|_2\nonumber\\
 \leq&\left\|\frac{1}{n_0}\sum_{i=1}^{n_0}\bxiz\{\wbRiz-\bRaiz\}\tr \{\wbQ^{-1}-\bQ^{-1}\}\right\|_2\nonumber\\
 &+\left\|\frac{1}{n_0}\sum_{i=1}^{n_0}\bxiz\{\bRaiz\}\tr \{\wbQ^{-1}-\bQ^{-1}\}\right\|_2\nonumber\\
 &+\left\|\frac{1}{n_0}\sum_{i=1}^{n_0}\bxiz\{\wbRiz-\bRaiz\}\tr \bQ^{-1}\right\|_2\nonumber
\end{align}
\begin{align}
 &+\left\|\frac{1}{n_0}\sum_{i=1}^{n_0}\bxiz(\bRaiz)\tr \bQ^{-1}-\bmu^\ast\right\|_2\nonumber\\
 =&O_p(\Delta_N)+O_p(n_0^{-1/2}).
\end{align}

With the quadratic loss and by the property of $\ExbRaz$, we have
\begin{equation*}
  \ExbRaz=\bmu^\ast\bR^\ast(\bZ_0).
\end{equation*}
Hence, by (\ref{diffwbmu}) and $\|\wbRz-\bRaz\|_2=O_p(\Delta_N)$, and the independence between $\calD_{n, 0}$ and $\calD_{n,k}$ for $k=1,\ldots,K$, we have
\begin{align*}\label{diffJJ}
&\|\wbJz-\bJ_0\|_2\nonumber\\
=&\left\|\frac{1}{n_0}\sum_{i=1}^{n_0}\{\bxiz-\wbmu\wbRiz\}\{\bxiz-\wbmu\wbRiz\}\tr -\bJ_0\right\|_2\nonumber\\
\leq&\left\|\frac{1}{n_0}\sum_{i=1}^{n_0}\{\bxiz-\wbmu\wbRiz\}\{\bxiz-\wbmu\wbRiz\}\tr \right.\nonumber\\
 &\left.\quad-\frac{1}{n_0}\sum_{i=1}^{n_0}\{\bxiz-\bmu^\ast\wbRiz\}\{\bxiz-\bmu^\ast\wbRiz\}\tr \right\|_2\nonumber\\
 &+\left\|\frac{1}{n_0}\sum_{i=1}^{n_0}\{\bxiz-\bmu^\ast\wbRiz\}\{\bxiz-\bmu^\ast\wbRiz\}\tr \right.\nonumber\\
 &\left.\quad -E\left[\{\bX_0-\bmu^\ast\wbRz\}\{\bX_0-\bmu^\ast\wbRz\}\tr \right]\right\|_2\nonumber\\
 &+\left\|E\left[\{\bX_0-\bmu^\ast\wbRz\}\{\bX_0-\bmu^\ast\wbRz\}\tr \right.\right.\nonumber\\&\left.\left.\quad-\{\bX_0-\bmu^\ast\bR^\ast(\bZ_0)\}\{\bX_0-\bmu^\ast\bR^\ast(\bZ_0)\}\tr \right]\right\|_2\nonumber\\
 &+\left\|E\left[\{\bX_0-\bmu^\ast\bR^\ast(\bZ_0)\}\{\bX_0-\bmu^\ast\bR^\ast(\bZ_0)\}\tr \right]-\bJ_0\right\|_2\nonumber\\
=& O_p(n_0^{-1/2})+O_p(\Delta_N).
\end{align*}
\end{proof}

\subsection{Proof of Theorems and Corollaries} \label{sec:proof}
\begin{proof}[Proof of Theorem \ref{theidenti}]
For any $\bpsi$ satisfying \eqref{deff0}, we have
  \begin{align*}
     & \frac{1}{K} \sum_{k=1}^K \expect \bigl[ \ell(\mathcal{D}_k ; \bbeta_k, \bgamma_k, \bR) - \ell(\mathcal{D}_k ; \bbeta_{\ast k}, \bgamma_{\ast k}, \bR_\ast) \bigr] \\
      = {} & \frac{1}{K} \sum_{k=1}^K \expect \bigl\{ \bigl[\bR_\ast\tr (\bZ_k) \bgamma_{\ast k} - \bR\tr (\bZ_k) \bgamma_k + \bX_k^\top \bbeta_{\ast k}- \bX_k^\top \bbeta_k \bigr]^2 \bigr\} \\
      = &\frac{1}{K}\sum_{k=1}^K \rmE\left[\left((\bbeta_k-\bbeta_{\ast k})\tr \{\bX_k-\Exzk\}\right.\right.
    \end{align*}
      \begin{align*} &+\left.\left.(\bbeta_k-\bbeta_{\ast k})\tr \{\Exzk\}+\{\bR\tr (\bZ_k)\bgamma_k-\bR_\ast\tr(\bZ_k)\bgamma_{\ast k}\}\right)^2\right]\\
      =&\frac{1}{K}\sum_{k=1}^K \rmE\left[\left((\bbeta_k-\bbeta_{\ast k})\tr \{\bX_k-\Exzk\}\right)^2\right]\\&+\frac{1}{K}\sum_{k=1}^K \rmE\left[\left((\bbeta_k-\bbeta_{\ast k})\tr \{\Exzk\}+\{\bR\tr (\bZ_k)\bgamma_k-\bR_\ast\tr(\bZ_k)\bgamma_{\ast k}\}\right)^2\right].
  \end{align*}
  Together with Condition~\ref{assumKL}, we conclude that $$\frac{1}{K} \sum_{k=1}^K\expect \bigl[ \ell(\mathcal{D}_k ; \bbeta_k, \bgamma_k, \bR) - \ell(\mathcal{D}_k ; \bbeta_{\ast k}, \bgamma_{\ast k}, \bR_\ast) \bigr]>0,$$ for any $\bbeta_k \neq \bbeta_{\ast k}$.
  Hence, by the definition  of $\bpsi$ and $\bpsi_\ast$, we attain that $\bbeta_k = \bbeta_{\ast k}$, and
  \begin{equation}\label{equ1}
    \bR\tr (\bZ)\bgamma_k=\bR_\ast\tr (\bZ)\bgamma_{\ast k},
  \end{equation}
   for all $\bZ\in\calZ$, every $\bbeta_k\in \calB$ and $\bgamma_k\in \Gamma$, and $k\in[K]$. We ignore the upscript of $\bZ_k$ for erasing the confusion, as for all $\bZ_k\in \calZ$, equation (\ref{equ1}) holds. By Condition 
   \ref{assumind}, we construct an invertible matrix $U_0=\left[\bgamma_{\ast k_1},\cdots,\bgamma_{\ast k_p}\right]\in \mathbb{R}^{p \times p}$ such that $U_0\tr \bR_\ast(\bZ)=U\tr \bR(\bZ)$ for all $\bZ\in\calZ$, where $U=\left[\bgamma_{k_1},\cdots,\bgamma_{k_p}\right]\in \mathbb{R}^{p \times p}$. Then we have $\bR_\ast(\bZ)=(U_0\tr )^{-1}U\tr \bR(\bZ)$.
   Note that by Condition \ref{assumind}, there exist $p$ tasks with input $\bZ_1,\ldots, \bZ_p$ such that  $V_0=\left[\bR_\ast(\bZ_1),\ldots,\bR_\ast(\bZ_p)\right]\in \mathbb{R}^{p \times p}$ is an invertible matrix. Consequently, we can write as
   \begin{equation*}
    V_0=\Lambda V,
   \end{equation*}
   where $\Lambda=(U_0)^{-\rmT}U\tr $ and $V=\left[\bR(\bZ_1),\ldots,\bR(\bZ_p)\right]\in \mathbb{R}^{p \times p}$.
   Since $V_0$ is invertible and $U$ does not depend on the input, so are $\Lambda$ and $V$. This completes the proof.
  \end{proof}
  \begin{proof}[Proof of Theorem \ref{thediffh}]
     We center the functions to
    $$\ell_{ik}(\calD_{ki};\bpsi_k)=\ell(\calD_{ki};\bpsi_k)-\ell(\calD_{ki};\bf{0}),$$ 
    where $\bpsi_k=(\bbeta_{k}, \bgamma_{k}, \bR)$. Under Condition \ref{conb}, applying the contraction principle \citep[Theorem 4.12]{ledoux2013probability} over set $\{\bbeta_k^{\top}\bxik+\bgamma_k^{\top}\bRik,i\in [n_k],k\in [K]\} \subseteq \mbR^{N}$ shows that
    \begin{align}\label{contr}
      &\rmE_{\varepsilon}\left[\sup\limits_{\bpsi\in \Psi}\frac{1}{K}\sum_{k=1}^K\frac{1}{n_k}\sum_{i=1}^{n_k}\varepsilon_{ik}\ell_{ik}(\calD_{ki};\bpsi_k)\right]
      \leq 2 B_{\delta}\calF_{N}(\calB^{\otimes K} +\Gamma^{\otimes K}(\calR)),
    \end{align}
   with probability at least $1-\delta$, where $B_\delta=c\log{(1/\delta)}+ 4 B_{X} B_{\beta}+4B_{R}B_{\gamma}$ and $c$ is a constant.
    Additional, we can easily prove that  $|\ell_{ik}(\calD_{ki};{\bf{0}})|\leq B_{\ell}$ with probability $1-\delta$ under Condition \ref{conb}, where $B_{\ell}=(c\log{(1/\delta)}+ B_{X} B_{\beta}+B_{R}B_{\gamma})^2$. Further, the constant-shift property of Rademacher averages (\citet{wainwright2019high}, Exercise 4.7c) gives
    \begin{align}\label{csp}
      &\rmE_{\varepsilon}\left[\sup\limits_{\bpsi\in \Psi}\frac{1}{K}\sum_{k=1}^K\frac{1}{n_k}\sum_{i=1}^{n_k}\varepsilon_{ik}\ell(\calD_{ki};\bpsi_k)\right]\nonumber\\
      \leq& \rmE_{\varepsilon}\left[\sup\limits_{\bpsi\in \Psi}\frac{1}{K}\sum_{k=1}^K\frac{1}{n_k}\sum_{i=1}^{n_k}\varepsilon_{ik}\ell_{ik}(\calD_{ki};\bpsi_k)\right]+\frac{B_\ell}{\sqrt{N}},
    \end{align}
    with probability at least $1-\delta$.
  Theorem 4.10 of \citet{wainwright2019high} shows that
    \begin{align*}
      &\sup\limits_{\bpsi\in \Psi}\left|\frac{1}{K}\sum_{k=1}^{K}\calL(\calD_{n, k};\bbeta_k,\bgamma_k,\bR)-\frac{1}{K}\sum_{k=1}^{K}\rmE\left\{\ell(\calD_{k};\bbeta_k,\bgamma_k,\bR)\right\} \right|\\
      \leq& 2\calF_{N}(\ell(\calB^{\otimes K} +\Gamma^{\otimes K}(\calR)))+2B_\ell\sqrt{\frac{\log(1/\delta)}{{N}}}
    \end{align*}
    with probability at least $1-3\delta$, where $\calL(\calD_{n, k};\bbeta_k,\bgamma_k,\bR)= \sum_{i=1}^{n_k}\ell(\calD_{ki};
\bbeta_k,\bgamma_k,\bR) / n_k$.
    Consequently, combining (\ref{contr}) and (\ref{csp}), we have
    \begin{align}\label{diffLN0}
      &\sup\limits_{\bpsi\in \calF}\left|\frac{1}{K}\sum_{k=1}^{K}\calL(\calD_{n, k};\bbeta_k,\bgamma_k,\bR)-\frac{1}{K}\sum_{k=1}^{K}\rmE\left\{\ell(\calD_{k};\bbeta_k,\bgamma_k,\bR)\right\} \right|\nonumber\\
      \leq& 4 B_{\delta}\calF_{N}(\calB^{\otimes K} +\Gamma^{\otimes K}(\calR))+4B_\ell\sqrt{\frac{\log(1/\delta)}{N}},
    \end{align}
    with probability at least $1-4\delta$.
    Therefore, under Condition~\ref{con4}, combining (\ref{diffLN0}) with Lemma \ref{bR}, we can conclude that
    \begin{align}\label{diffLL}
      \sup\limits_{\bpsi\in \calF}\left|\frac{1}{K}\sum_{k=1}^{K}\calL(\calD_{n, k};\bbeta_k,\bgamma_k,\bR)-\frac{1}{K}\sum_{k=1}^{K}\rmE\left\{\ell(\calD_{k};\bbeta_k,\bgamma_k,\bR)\right\} \right|\pover 0.
    \end{align}
    Further applying Theorem \ref{theidenti}, there exists an invertible matrix $\Lambda_{\ast}$ such that, $\widehat{\bR}$ coverges to $\bR^\ast$, and $\wbgamma_k$ converges to $\bgamma^{\ast}_k$, where $\bR^\ast=\Lambda_{\ast}^{-1}\bR_\ast$ and $\bgamma_{k}^{\ast}=\Lambda_{\ast}^{\top}\bgamma_{\ast k}$ for $k\in[K]$.
    Define
    \begin{equation}\label{defth}
      \tbR=\arg\min_{\bR\in\calR} \frac{1}{K}\sum_{k=1}^K \rmE\left[\left(\bR\tr (\bZ_k)\bgamma_{\ast k}-\bR\tr _\ast(\bZ_k)\bgamma_{\ast k}\right)^2\right].
    \end{equation}
     Under Conditions \ref{conb}--\ref{conRieZ}, by the proof of Theorem 6.2 in \citet{jiao2023deep}, we know that if the network width and depth be $W=114(\lfloor\kappa\rfloor+1)(p\log q)^{\lfloor\kappa\rfloor+1}$ and $D=21(\lfloor\kappa\rfloor+1)^2\lceil N^{(p\log q)/2(p\log q +2\kappa)}\log_2(8N^{(p\log q)/2(p\log q +2\kappa)})\rceil$, then
    \begin{align*}
        &\frac{1}{K}\sum_{k=1}^K\rmE\left[\left(\bgamma^{\ast\top}_{k}\widetilde{\bR}^{\ast}(\bZ_k)-\bgamma^{\ast\top}_{k}\bR^{\ast}(\bZ_k)\right)^2\right]\\
      =&\frac{1}{K}\sum_{k=1}^K\rmE\left[\left(\bgamma^{\top}_{\ast k}\widetilde{\bR}(\bZ_k)-\bgamma^{\top}_{\ast k}\bR_\ast(\bZ_k)\right)^2\right]
    \end{align*}
        \begin{align*}
      \leq &\frac{1}{K}\sum_{k=1}^K\rmE\left[\left\|\widetilde{\bR}(\bZ_k)-\bR_\ast(\bZ_k)\right\|^2\|\bgamma_{\ast k}\|^2\right]\\=&O(p B_{\gamma}^2 N^{-\frac{2s}{2s+p\log q}}),
    \end{align*}
    where $\tbR^{\ast}(\cdot)=\Lambda_{\ast}\tbR(\cdot)$.
    Consequently, under Condition \ref{con4}, we have
    \begin{align}\label{diffEhh}
      &\frac{1}{K}\sum_{k=1}^{K}\rmE\left[\ell(\calD_{k};\bbeta_{\ast k},\bgamma_k^\ast,\tbR^{\ast})\right]-\frac{1}{K}\sum_{k=1}^{K}\rmE\left[\ell(\calD_{k};\bbeta_{\ast k},\bgamma_k^{\ast},\bR^{\ast})\right]\nonumber\\=& \frac{1}{K}\sum_{k=1}^K \rmE\left[\left(\tbR^{\ast\rmT}(\bZ_k)\bgamma_k^\ast-\bR^{\ast\rmT}_{}(\bZ_k)\bgamma_k^\ast\right)^2\right]\rightarrow 0.
    \end{align}
    Then, (\ref{diffLL}) and (\ref{diffEhh}) lead to
    \begin{align}\label{diffLhh}
    &\left|\frac{1}{K}\sum_{k=1}^{K}\calL(\calD_{n, k};\bbeta_{\ast k},\bgamma_k^{\ast},\tbR^{\ast})-\frac{1}{K}\sum_{k=1}^{K}\calL(\calD_{n, k};\bbeta_{\ast k},\bgamma_k^{\ast},\bR^{\ast})\right|\nonumber\\\leq&
    \left|\frac{1}{K}\sum_{k=1}^{K}\calL(\calD_{n, k};\bbeta_{\ast k},\bgamma_k^{\ast},\tbR^{\ast})-\frac{1}{K}\sum_{k=1}^{K}\rmE\left[\ell(\calD_{k};\bbeta_{\ast k},\bgamma_k^{\ast},\tbR^{\ast})\right]\right|\nonumber\\
    &+\left|\frac{1}{K}\sum_{k=1}^{K}\rmE\left[\ell(\calD_{k};\bbeta_{\ast k},\bgamma_k^{\ast},\tbR^{\ast})\right]-\frac{1}{K}\sum_{k=1}^{K}\rmE\left[\ell(\calD_{k};\bbeta_{\ast k},\bgamma_k^{\ast},\bR^{\ast})\right]\right|\nonumber\\
    &+\left|\frac{1}{K}\sum_{k=1}^{K}\rmE\left[\ell(\calD_{k};\bbeta_{\ast k},\bgamma_k^{\ast},\bR^{\ast})\right]-\frac{1}{K}\sum_{k=1}^{K}\calL(\calD_{n, k};\bbeta_{\ast k},\bgamma_k^{\ast},\bR^{\ast})\right|\nonumber\\=&o_p(1).
    \end{align}
    Noting that $\{(\wbbeta_k,\wbgamma_k)_{k=1}^K,\wbR\}$ is the minimizer of (\ref{estMul}), by (\ref{diffLhh}), we further have
    \begin{align}\label{diffhatLL}
      \frac{1}{K}\sum_{k=1}^{K}\calL(\calD_{n, k};\wbbeta_k,\wbgamma_k,\wbR)
      \leq& \frac{1}{K}\sum_{k=1}^{K}\calL(\calD_{n, k};\bbeta_{\ast k},\bgamma_k^{\ast},\tbR^{\ast})\nonumber\\
    \leq&\frac{1}{K}\sum_{k=1}^{K}\calL(\calD_{n, k};\bbeta_{\ast k},\bgamma_k^{\ast},\bR^{\ast})+o_p(1).
    \end{align}

    Since
    \begin{align*}
      &\frac{1}{K}\sum_{k=1}^{K}\rmE\left[\ell(\calD_{k};\bbeta_k,\bgamma_k,\bR)\right]-\frac{1}{K}\sum_{k=1}^{K}\rmE\left[\ell(\calD_{k};\bbeta_{\ast k},\bgamma_k^{\ast},\bR^{\ast})\right]\\=&d_2^2(\bpsi,\bpsi^{\ast}),
    \end{align*}
    then, for any small $\delta>0$, we have
    \begin{align}\label{diffLL0f}
     \inf\limits_{d(\bpsi,\bpsi^{\ast})\geq \delta} \frac{1}{K}\sum_{k=1}^{K}\rmE\left[\ell(\calD_{k};\bbeta_k,\bgamma_k,\bR)\right]
     >&\frac{1}{K}\sum_{k=1}^{K}\rmE\left[\ell(\calD_{k};\bbeta_{\ast k},\bgamma_{\ast k},\bR_\ast)\right]\nonumber\\
     >&\frac{1}{K}\sum_{k=1}^{K}\rmE\left[\ell(\calD_{k};\bbeta_{\ast k},\bgamma_k^{\ast},\bR^{\ast})\right].
    \end{align}
    Therefore, the Conditions of Theorem 5.7 in \citet{van2000asymptotic} follow from (\ref{diffLL}), (\ref{diffhatLL}), and (\ref{diffLL0f}), and this implies that $d(\widehat{\bpsi},\bpsi^{\ast})\pover 0$ as $n\rightarrow \infty$, where $\widehat{\bpsi}=\{(\wbbeta_k,\wbgamma_k)_{k=1}^K,\wbR\}$.
    Next, we show the convergence rates of $d(\widehat{\bpsi},\bpsi^{\ast})$. By lemma \ref{lemog}, we have
    \begin{align}\label{outer}
     \rmE^{\ast} &\left[\sup\limits_{\bpsi\in \calF_{\delta}}\sqrt{N}\left|\frac{1}{K}\sum_{k=1}^{K}\calL(\calD_{n, k};\bbeta_k,\bgamma_k,\bR)-\frac{1}{K}\sum_{k=1}^{K}\rmE\left[\ell(\calD_{k};\bbeta_k,\bgamma_k,\bR)\right]\right.\right.\nonumber\\&\left.\left.-\left\{\frac{1}{K}\sum_{k=1}^{K}\calL(\calD_{n, k};\bbeta_{\ast k},\bgamma_k^{\ast},\bR^{\ast})-\frac{1}{K}\sum_{k=1}^{K}\rmE\left[\ell(\calD_{k};\bbeta_{\ast k},\bgamma_k^{\ast},\bR^{\ast})\right]\right\}\right|\right]\nonumber\\
     \lesssim&\bphi_N(\delta),
      \end{align}
      where $\bphi_N(\delta)=\delta \sqrt{s_1\log(s_2/\delta)}+\frac{s_1}{\sqrt{N}}\log(s_2/\delta)$.

    Denote $\upsilon_N=s_1^{1/2}(\log N)^2N^{-1/2}$. With some calculations, we have that $\bphi_N(\upsilon_N)\leq \upsilon_N^2\sqrt{N}$ and $\bphi_N(p^{1/2} B_{\gamma}N^{-{s}/({2s+p\log q})}) \leq  \sqrt{N}p B_{\gamma}^2 N^{-\frac{2s}{2s+p\log q}}$. On the other hand, by the definition of  $\widetilde{\bR}$ in (\ref{defth}) and analogy to (\ref{outer}),
    \begin{align*}
      &\left|\frac{1}{K}\sum_{k=1}^{K}\calL(\calD_{n, k};\bbeta_{\ast k},\bgamma_k^{\ast},\widetilde{\bR}^{\ast})-\frac{1}{K}\sum_{k=1}^{K}\calL(\calD_{n, k};\bbeta_{\ast k},\bgamma_k^{\ast},\bR^{\ast})\right|\\
      \lesssim&\left|\frac{1}{K}\sum_{k=1}^{K}\rmE\left\{\ell(\calD_{k};\bbeta_{\ast k},\bgamma_k^{\ast},\widetilde{\bR}^{\ast})\right\}-\frac{1}{K}\sum_{k=1}^{K}\rmE\left\{\ell(\calD_{k};\bbeta_{\ast k},\bgamma_k^{\ast},\bR^{\ast})\right\}\right|\\
     &+O_p({N}^{-1/2}\bphi_N(\upsilon_N))\\
     \lesssim& d_2^2(\{(\bbeta_{\ast k},\bgamma_k^{\ast})_{k=1}^K,\widetilde{\bR}^{\ast}\},\{(\bbeta_{\ast k},\bgamma_k^{\ast})_{k=1}^K,{\bR^{\ast}}\})+O_p({N}^{-1/2}\bphi_N(\upsilon_N))\\
     \leq& O(p B_{\gamma}^2 N^{-\frac{2s}{2s+p\log q}})+O_p(\upsilon_N^2).
    \end{align*}
    Then by the definition of $\widehat{\bpsi}$, we have
    \begin{align*}
      &\frac{1}{K}\sum_{k=1}^{K}\calL(\calD_{n, k};\wbbeta_k,\wbgamma_k,\wbR)\\
      \leq &\frac{1}{K}\sum_{k=1}^{K}\calL(\calD_{n, k};\bbeta_{\ast k},\bgamma_k^{\ast},\tbR^{\ast})\\
      \leq& \frac{1}{K}\sum_{k=1}^{K}\calL(\calD_{n, k};\bbeta_{\ast k},\bgamma_k^{\ast},\bR^{\ast})+O_p(p B_{\gamma}^2 N^{-\frac{2s}{2s+p\log q}}+\upsilon_N^2).
    \end{align*}
    By Theorem 3.4.1 in \citet{van1996weak}, we have $d(\widehat{\bpsi},\bpsi^{\ast})=O(p^{1/2}N^{-{s}/({2s+p\log q})}+\upsilon_N)$.
    Furthermore, we have
    \begin{align*}
      &d_2^2(\widehat{\bpsi},\bpsi^{\ast})\\=&\frac{1}{K}\sum_{k=1}^K \rmE\left[\left((\wbbeta_k-\bbeta_{\ast k})\tr \{\bX_k-\ExbRak\}+\right.\right.\\&\left.\left.(\wbbeta_k-\bbeta_{\ast k})\tr \{\ExbRak\}+\{\wbR\tr (\bZ_k)\wbgamma_k-{\bR^{\ast}}\tr (\bZ_k)\bgamma_k^{\ast}\}\right)^2\right]\\
      =&\frac{1}{K}\sum_{k=1}^K \rmE\left[\left((\wbbeta_k-\bbeta_{\ast k})\tr \{\bX_k-\ExbRak\}\right)^2\right]\\&+\frac{1}{K}\sum_{k=1}^K \rmE\left[\left((\wbbeta_k-\bbeta_{\ast k})\tr \{\ExbRak\}+\{\wbR\tr (\bZ_k)\wbgamma_k-{\bR^{\ast}}\tr (\bZ_k)\bgamma_k^{\ast}\}\right)^2\right].
    \end{align*}
    Thus, by Conditions \ref{assumKL}--\ref{conb}, it follows that $\max_{k\in[K]}\|\wbbeta_k-\bbeta_{\ast k}\|_2=O_p(p^{1/2}N^{-{s}/({2s+p\log q})}+\upsilon_N)$
    and
    \begin{align*} 
      \frac{1}{K}\sum_{k=1}^K \rmE\left[\{\wbR\tr (\bZ_k)\wbgamma_k-{\bR^{\ast}}\tr (\bZ_k)\bgamma_k^{\ast}\}^2\right]=O(p B_{\gamma}^2 N^{-\frac{2s}{2s+p\log q}}+\upsilon_N^2).
    \end{align*}
Denote
$$\{\widetilde{\boldsymbol{\gamma}}_k^\prime,\widetilde{\boldsymbol{\gamma}}_k\}=\arg\sup\limits_{\bgamma_k^\prime\in\Gamma}\inf\limits_{\bgamma_k\in\Gamma}\frac{1}{K}\sum_{k=1}^K \rmE\left[\{\wbR\tr (\bZ_k)\bgamma_k-{\bR^{\ast}}\tr (\bZ_k)\bgamma_k^\prime\}^2\right].$$
    Furthermore,
    \begin{align*}
    &\frac{1}{K}\sum_{k=1}^K \rmE\left[\{\wbR\tr (\bZ_k)\widetilde{\boldsymbol{\gamma}}_k-{\bR^{\ast}}\tr (\bZ_k)\widetilde{\boldsymbol{\gamma}}_k\}^2\right]\nonumber\\
    \leq&\sup\limits_{\bgamma_k^\prime\in\Gamma}\inf\limits_{\bgamma_k\in\Gamma}\frac{1}{K}\sum_{k=1}^K \rmE\left[\{\wbR\tr (\bZ_k)\bgamma_k-{\bR^{\ast}}\tr (\bZ_k)\bgamma_k^\prime\}^2\right]\nonumber\\
        \leq& C \inf\limits_{\bgamma_k\in\Gamma}\frac{1}{K}\sum_{k=1}^K \rmE\left[\{\wbR\tr (\bZ_k)\bgamma_k-{\bR^{\ast}}\tr (\bZ_k)\bgamma_k^\ast\}^2\right] \\
        \leq& \frac{C}{K}\sum_{k=1}^K \rmE\left[\{\wbR\tr (\bZ_k)\wbgamma_k-{\bR^{\ast}}\tr (\bZ_k)\bgamma_k^\ast\}^2\right],
    \end{align*}
    where the second inequality is implied by  Lemma 6 of \citet{tripuraneni2020theory}. Consequently, by Condition \ref{conb}, we have $d_2(\wbR,\bR^{\ast})=O(\Delta_N)$.
    \end{proof}

\begin{proof}[Proof of Theorem \ref{thenorm}]
Under Condition \ref{conb}, by Theorem \ref{thediffh} and Lemma \ref{diffQmuJ}, it is easily to conclude that
  \begin{equation}\label{tem1}
    \{\wbJz\}^{-1}\left[\frac{1}{n_0}\sum_{i=1}^{n_0}\{\bxiz-\wbmu\wbRiz\}
    \{\bRaiz-\wbRiz\}\tr \bgamma_0^\ast\right]=O_p(n_0^{-1/2}\Delta_N+\Delta_N^2).
  \end{equation}
  By (\ref{diffwbmu}), and the independence between $\calD_{n,k}$ and $\calD_{n, 0}$, we have
  \begin{align*}
  &\left\|\frac{1}{n_0}\sum_{i=1}^{n_0}\{\bmu^\ast\wbRiz-\wbmu\wbRiz\}\epsiz\right\|_2\\
  \leq &\left\|\wbmu-\bmu^\ast\right\|_2\left\|\frac{1}{n_0}\sum_{i=1}^{n_0}\wbRiz\epsiz\right\|_2\\
  =&O_p(n_0^{-1/2}\Delta_N+n_0^{-1}).
  \end{align*}
  Consequently, we can easily obtain
  \begin{align} \label{tem3}
  &\{\wbJz\}^{-1}\left[\frac{1}{n_0}\sum_{i=1}^{n_0}
    \{\bmu^\ast\bRaiz-\bmu^\ast\wbRiz+\bmu^\ast\wbRiz-\wbmu\wbRiz\}\epsiz\right]\nonumber\\
    =&O_p(n_0^{-1/2}\Delta_N+n_0^{-1}).
  \end{align}

  Recall the first-order optimality Conditions
  \begin{align*} 
    &\frac{1}{n_0}\sum_{i=1}^{n_0}\bxiz\{Y_{0i}-\bxiz\tr \wbbetaz-(\wbRiz)\tr \wbgamma_0\}=0;  \nonumber\\
    &\frac{1}{n_0}\sum_{i=1}^{n_0}\wbRiz\{Y_{0i}-\bxiz\tr \wbbetaz-(\wbRiz)\tr \wbgamma_0\}=0.
  \end{align*}
  Then the empirical orthogonal score  for $\bbeta_{\ast 0}$ is then given by
  \begin{equation*}
    \frac{1}{n_0}\sum_{i=1}^{n_0}\{\bxiz-\wbmu\wbRiz\}\{Y_{0i}-\bxiz\tr \wbbetaz-(\wbRiz)\tr \wbgamma_0\}=0.
  \end{equation*}
  With some simple calculations, we can obtain that
  \begin{align*}
    \frac{1}{n_0}\sum_{i=1}^{n_0}&\{\bxiz-\wbmu\wbRiz\}\{\bxiz\tr (\bbeta_{\ast 0}-\wbbetaz)-
    (\wbRiz)\tr (\wbgamma_0-\bgamma_0^\ast)\\&+(\bRaiz-\wbRiz)\tr \bgamma_0^\ast+\epsiz\}=0.
  \end{align*}
  Then  we conclude
  \begin{align*} 
    \wbbetaz-\bbeta_{\ast 0}=&\{\wbJz\}^{-1}\left[\frac{1}{n_0}\sum_{i=1}^{n_0}\{\bxiz-\wbmu\wbRiz\}
    (\bRaiz-\wbRiz)\tr \bgamma_0^\ast\right]\nonumber\\
    &+\{\wbJz\}^{-1}\left[\frac{1}{n_0}\sum_{i=1}^{n_0}
    \{\bxiz-\wbmu\wbRiz\}\epsiz\right]\nonumber\\
  =&\{\wbJz\}^{-1}\left[\frac{1}{n_0}\sum_{i=1}^{n_0}\{\bxiz-\wbmu\wbRiz\}
  (\bRaiz-\wbRiz)\tr \bgamma_0^\ast\right]\nonumber
\end{align*}
  \begin{align*}
  &+\{\wbJz\}^{-1}\left[\frac{1}{n_0}\sum_{i=1}^{n_0}
    \{\bxiz-\bmu^\ast\bRaiz\}\epsiz\right]\nonumber\\
    &+\{\wbJz\}^{-1}\left[\frac{1}{n_0}\sum_{i=1}^{n_0}
    \{\bmu^\ast\bRaiz-\bmu^\ast\wbRiz+\bmu^\ast\wbRiz-\wbmu\wbRiz\}\epsiz\right].
  \end{align*}
  By (\ref{diffJJ1}), (\ref{tem1}), and (\ref{tem3}), we obtain that
  \begin{align*} 
    \wbbetaz-\bbeta_{\ast 0}=\{\bJ_0\}^{-1}\left[\frac{1}{{n}_0}\sum_{i=1}^{n_0}
    \{\bxiz-\bmu^\ast\bRaiz\}\epsiz\right]+O_p(\Delta_N^2).
  \end{align*}
  Consequently, we have proved that
  \begin{align*}
  \sqrt{n_0}(\wbbetaz-\bbeta_{\ast 0})\dover N(0,\sigma_0^2 {\bJ_0}^{-1}).
  \end{align*}
  \end{proof}
  \begin{proof}[Proof of Corollary \ref{corvar}]
 Denote the efficient score  for $\bbeta_{\ast 0}$ as
  \begin{align*}
    \bPhi(\calD_0;\bmu_0,\bbeta_{\ast 0},\bgamma_0^\ast,\bR^\ast)=\{\bX_0-\bmu^{\ast}\bR^{\ast}(\bZ_0)\}\{Y_0-\bX_0\tr \bbeta_{\ast 0}-(\bR^{\ast}(\bZ_0))\tr \bgamma_0^{\ast}\}.
  \end{align*}
  Then the efficient orthogonal score  for $\bbeta_{\ast 0}$ is then given by
  \begin{align*}
    \bPhi(\calD_0;\wbmu,\wbbetaz,\wbgamma_0,\wbR)=\{\bX_0-\wbmu\wbRiz\}\{Y_0-\bX_0\tr \wbbetaz-(\wbRiz)\tr \wbgamma_0\}.
  \end{align*}
  Note that $\bPhi$ is a $d$-dimensional vector. Let $\Phi_j$ denote the $j$th element of $\bPhi$, $j\in[d]$. To simplify the notation, write
  $\Phi_j(\calD_0)=\Phi_j(\calD_0;\bmu^\ast,\bbeta_{\ast 0},\bgamma_0^\ast,\bR^\ast)$ and $\widehat{\Phi}_j(\calD_0;\wbmu,\wbbetaz,\wbgamma_0,\wbR)$.
  For any $j\in[d]$, with some calculations, under Conditions \ref{conb} and \ref{inv}, the results of Theorem \ref{thediffh} and Theorem \ref{thenorm} imply that
  \begin{align}\label{diffpsi}
    \frac{1}{n_0}\sum_{i=1}^{n_0}\left\|\widehat{\bPhi}(\calD_0)-\bPhi(\calD_0)\right\|_2^2=O_p(\Delta_N+n_0^{-1/2}).
  \end{align}
  By Condition \ref{inv}, we have that
  \begin{align*}
    &{\rmE}\left[\left\|\bPhi(\bmu^\ast,\bbeta_{\ast 0},\bgamma_0^\ast,\bR^\ast)\right\|_2^2\right]\\=&\rmE[\sigma_{0}^2]\rmE[\{\bX_0-\bm{m}^\ast(\bZ_0)\}\tr \{\bX_0-\bm{m}^\ast(\bZ_0)\}]=O(1).
  \end{align*}
  Therefore, for any $j_1,j_2\in[d]$,
  \begin{align}\label{boundpsi}
    \left(\frac{1}{n_0}\sum_{i=1}^{n_0}\left|\Phi_{j_1}(\calD_{0i})\right|^2\bigvee\left|\Phi_{j_2}(\calD_{0i})\right|^2\right)^{1/2}=O_p(1),
  \end{align}
  where $a \bigvee b=\max\{a,b\}$.
  Consequently, by (\ref{diffpsi}) and (\ref{boundpsi}), we have
  \begin{align}\label{boundpsi2}
    &\left|\frac{1}{n_0}\sum_{i=1}^{n_0}\widehat{\Phi}_{j_1}(\calD_{0i})\widehat{\Phi}_{j_2}(\calD_{0i})-\frac{1}{n_0}\sum_{i=1}^{n_0}\Phi_{j_1}(\calD_{0i})\Phi_{j_2}(\calD_{0i})\right|\nonumber\\
    \leq& \frac{1}{n_0}\sum_{i=1}^{n_0}\left|\widehat{\Phi}_{j_1}(\calD_{0i})\widehat{\Phi}_{j_2}(\calD_{0i})-\Phi_{j_1}(\calD_{0i})\Phi_{j_2}(\calD_{0i})\right|\nonumber \\
    \leq& \frac{1}{n_0}\sum_{i=1}^{n_0}\left(\left|\widehat{\Phi}_{j_1}(\calD_{0i})-\Phi_{j_1}(\calD_{0i})\right|\bigvee\left|\widehat{\Phi}_{j_2}(\calD_{0i})-\Phi_{j_2}(\calD_{0i})\right|\right)\nonumber\\
    &\times\left(\left|\Phi_{j_1}(\calD_{0i})\right|\bigvee\left|\Phi_{j_2}(\calD_{0i})\right|+\left|\widehat{\Phi}_{j_1}(\calD_{0i})-\Phi_{j_1}(\calD_{0i})\right|\bigvee\left|\widehat{\Phi}_{j_2}(\calD_{0i})-\Phi_{j_2}(\calD_{0i})\right|\right)\nonumber\\
    \leq& O_p(\Delta_N+n_0^{-1/2}).
  \end{align}
  Furthermore, note that
  \begin{align}\label{diffEpsi}
    \frac{1}{n_0}\sum_{i=1}^{n_0}\Phi_{j_1}(\calD_{0i})\Phi_{j_2}(\calD_{0i})-\rmE(\epsilon_0^2V_{j_1}^\ast V_{j_2}^\ast)=O_p(n_0^{-1/2}),
  \end{align}
   where $V_j^\ast$ is the $j$th element of $\bX_0-\bm{m}^\ast(\bZ_0)$.
   Therefore, by (\ref{boundpsi2}) and (\ref{diffEpsi}), we have
   \begin{align*}
      &\left|\frac{1}{n_0}\sum_{i=1}^{n_0}\widehat{\Phi}_{j_1}(\calD_{0i})\widehat{\Phi}_{j_2}(\calD_{0i})-\rmE(\epsilon_0^2V_{j_1}^\ast V_{j_2}^\ast)\right|\\
      \leq& O_p(\Delta_N+n_0^{-1/2}).
   \end{align*}
   Note that
   \begin{align*}
    \widehat{\bJ}_0\pover \bJ_0.
   \end{align*}
    Then applying the continuous mapping theorem completes the proof of Corollary \ref{corvar}.
  \end{proof}

\end{document}